%% file: local_interaction.tex
\documentclass[11pt]{aart}

\usepackage[letterpaper, hmargin=1in, top=1in, bottom=1.2in, footskip=0.6in]{geometry} 

\usepackage{titlesec}
\titleformat{\section}[block]{\filcenter\normalfont\bfseries\large}{\thesection.}{.5em}{}\titlespacing*{\section}{0pt}{2\baselineskip}{1\baselineskip}
\titleformat{\subsection}[runin]{\normalfont\bfseries}{\thesubsection.}{.4em}{}[.]\titlespacing{\subsection}{0pt}{2ex plus .1ex minus .2ex}{.8em}
\titleformat{\subsubsection}[runin]{\normalfont\itshape}{\thesubsubsection.}{.3em}{}[.]\titlespacing{\subsubsection}{0pt}{1ex plus .1ex minus .2ex}{.5em}
\titleformat{\paragraph}[runin]{\normalfont\itshape}{\theparagraph.}{.3em}{}[.]\titlespacing{\paragraph}{0pt}{1ex plus .1ex minus .2ex}{.5em}

\usepackage[T1]{fontenc}
\usepackage[utf8]{inputenc}
\usepackage{lmodern}

\usepackage[labelfont=bf,font=small,labelsep=period]{caption}
\let\originalleft\left
\let\originalright\right
\renewcommand{\left}{\mathopen{}\mathclose\bgroup\originalleft}
\renewcommand{\right}{\aftergroup\egroup\originalright}

\usepackage{amsmath}
\usepackage{amssymb}
\usepackage{amsfonts}
\usepackage{latexsym}
\usepackage{amsthm}
\usepackage{amsxtra}
\usepackage{amscd}
\usepackage{bbm}
\usepackage{mathrsfs}
\usepackage{bm}
\usepackage{mathtools}

\usepackage{graphicx, color}

\definecolor{darkred}{rgb}{0.9,0,0.3}
\definecolor{darkblue}{rgb}{0,0.3,0.9}

\definecolor{vdarkred}{rgb}{0.6,0,0.2}
\definecolor{vdarkblue}{rgb}{0,0.2,0.6}
\usepackage[pdftex, colorlinks, linkcolor=vdarkblue,citecolor=vdarkred]{hyperref}

\usepackage{booktabs}
\usepackage[nottoc,notlof,notlot]{tocbibind}
\usepackage{cite} 

\numberwithin{equation}{section}
\numberwithin{figure}{section}

\usepackage{enumitem}

\theoremstyle{plain} 
\newtheorem{theorem}{Theorem}[section]
\newtheorem*{theorem*}{Theorem}
\newtheorem{lemma}[theorem]{Lemma}
\newtheorem*{lemma*}{Lemma}
\newtheorem{corollary}[theorem]{Corollary}
\newtheorem*{corollary*}{Corollary}
\newtheorem{proposition}[theorem]{Proposition}
\newtheorem*{proposition*}{Proposition}

\newtheorem*{conjecture*}{Conjecture}

\theoremstyle{definition} 
\newtheorem{definition}[theorem]{Definition}
\newtheorem*{definition*}{Definition}
\newtheorem{example}[theorem]{Example}
\newtheorem*{example*}{Example}
\newtheorem{remark}[theorem]{Remark}
\newtheorem*{remark*}{Remark}
\newtheorem{assumption}[theorem]{Assumption}
\newtheorem*{assumption*}{Assumption}

\renewcommand{\r}{\mathrm} 
\newcommand{\bb}{\mathbb} 
\renewcommand{\cal}{\mathcal} 
 
\newcommand{\fra}{\mathfrak} 
\newcommand{\ol}[1]{\overline{#1} \!\,} 
\newcommand{\wh}{\widehat}
\newcommand{\wt}{\widetilde}

\renewcommand{\P}{\mathbb{P}}
\newcommand{\E}{\mathbb{E}}
\newcommand{\R}{\mathbb{R}}
\newcommand{\C}{\mathbb{C}}
\newcommand{\N}{\mathbb{N}}
\newcommand{\Z}{\mathbb{Z}}

\newcommand{\ee}{\r e}
\newcommand{\ii}{\r i}
\newcommand{\dd}{\r d}

\newcommand{\col}{\vcentcolon}
\newcommand*{\deq}{\mathrel{\vcenter{\baselineskip0.65ex \lineskiplimit0pt \hbox{.}\hbox{.}}}=}

\newcommand{\eqdist}{\overset{\r d}{=}}

\renewcommand{\leq}{\leqslant}
\renewcommand{\geq}{\geqslant}
\renewcommand{\epsilon}{\varepsilon}

\newcommand{\ind}[1]{\f 1_{#1}}

\newcommand{\p}[1]{(#1)}
\newcommand{\pb}[1]{\bigl(#1\bigr)}
\newcommand{\pB}[1]{\Bigl(#1\Bigr)}
\newcommand{\pbb}[1]{\biggl(#1\biggr)}
\newcommand{\pBB}[1]{\Biggl(#1\Biggr)}

\newcommand{\qb}[1]{\bigl[#1\bigr]}
\newcommand{\qB}[1]{\Bigl[#1\Bigr]}

\newcommand{\qBB}[1]{\Biggl[#1\Biggr]}
\newcommand{\qa}[1]{{\left[#1\right]}}

\newcommand{\abs}[1]{\lvert #1 \rvert}
\newcommand{\absb}[1]{\bigl\lvert #1 \bigr\rvert}
\newcommand{\absB}[1]{\Bigl\lvert #1 \Bigr\rvert}
\newcommand{\absbb}[1]{\biggl\lvert #1 \biggr\rvert}
\newcommand{\absBB}[1]{\Biggl\lvert #1 \Biggr\rvert}

\newcommand{\norm}[1]{\lVert #1 \rVert}
\newcommand{\normb}[1]{\bigl\lVert #1 \bigr\rVert}
\newcommand{\normB}[1]{\Bigl\lVert #1 \Bigr\rVert}

\newcommand{\scalar}[2]{\langle#1 \mspace{2mu}, #2\rangle}

\DeclareMathOperator{\tr}{Tr}

\DeclareMathOperator{\re}{Re}
\DeclareMathOperator{\im}{Im}

\DeclareMathOperator{\Span}{Span}

\newcommand{\wick}[1]{{\col\!#1\!\col}}
\newcommand{\f}[1]{\boldsymbol{\mathrm{#1}}}

\title{The Euclidean $\phi^4_2$ theory as a limit of an interacting Bose gas}

\author{	
J\"urg Fr\"ohlich
\and Antti Knowles
\and Benjamin Schlein
\and Vedran Sohinger
}

\begin{document}

\maketitle

\begin{abstract}
We prove that the complex Euclidean field theory with local quartic self-interaction in two dimensions arises as a limit of an interacting Bose gas at positive temperature, when the density of the gas becomes large and the range of the interaction becomes small. The field theory is supported on distributions of negative regularity, which requires a renormalization by divergent mass and energy counterterms. We obtain convergence of the relative partition function and uniform convergence of the renormalized reduced density matrices. The proof is based on three main ingredients: (a) a quantitative analysis of the infinite-dimensional saddle point argument for the functional integral introduced in \cite{FKSS_2020} using continuity properties of Brownian paths, (b) a Nelson-type estimate for a general nonlocal field theory in two dimensions, and (c) repeated Gaussian integration by parts in field space to obtain uniform control on the renormalized correlation functions. As a byproduct of our proof, in two and three dimensions we also extend the results on the mean-field limit from \cite{LNR3, FKSS_2020} to unbounded interaction potentials satisfying the optimal integrability conditions proposed by Bourgain \cite{bourgain1997invariant}.
\end{abstract}

\section{Introduction}

\subsection{Overview of Euclidean field theory} \label{sec:overview}
  
A Euclidean field theory of a scalar field on a domain $\Lambda \subset \R^d$ is specified by a formal probability measure on a space of fields\footnote{Rigorously, the space of fields is the Schwartz distribution space $\cal S'(\Lambda, \R^K)$.} $\phi \col \Lambda \to \R^K$ given by
\begin{equation} \label{mu_def}
\mu(\dd \phi) = \frac{1}{c} \, \ee^{-S(\phi)} \, \r D \phi\,,
\end{equation}
where $\r D \phi = \prod_{x \in \Lambda} \dd \phi(x)$ is the formal uniform measure on the space of fields, and $S$ is the action. The latter is typically the integral over $\Lambda$ of a local function of the field $\phi$ and its gradient. One of the simplest field theories with nontrivial interaction is the $K$-component \emph{Euclidean $\phi^4_d$ theory}, whose action is given by
\begin{equation} \label{S_def}
S(\phi) \deq - \int_\Lambda \dd x \, \phi(x) \cdot (\theta + \Delta / 2) \phi(x) + \frac{\lambda}{2} \int_\Lambda \dd x \, \abs{\phi(x)}^4\,,
\end{equation}
where $\theta$ is a constant,
$\lambda$ is a coupling constant, $\Delta$ is the Laplacian on $\Lambda$ with appropriate boundary conditions, and $\abs{\cdot}$ denotes the Euclidean norm on $\R^K$.

Euclidean field theories originally arose in high-energy physics in $d = 4$ space-time dimensions, through an analytic continuation of the time variable of the quantum field $\phi$, which replaces the Minkowski space-time metric with a Euclidean one \cite{schwinger1958euclidean, nakano1959quantum}. Subsequently, Euclidean field theories have proven of great importance in statistical mechanics in $d \leq 3$ dimensions, in particular through their connection with the theory of phase transitions and critical phenomena. The works \cite{symanzik1966euclidean, symanzik1969euclidean} recognized the analogy between Euclidean field theories and classical statistical mechanics, which was followed by a purely probabilistic formulation of Euclidean field theories in \cite{nelson1973free, nelson1973probability}. The rigorous study of field theories of the form \eqref{mu_def} has been a major topic in mathematical physics since the late sixties; see e.g. \cite{glimm2012quantum, Simon74} as well as the more recent \cite{Hairer2016} for reviews.

Euclidean field theories also play a central role in the theory of stochastic nonlinear partial differential equations. Formally, \eqref{mu_def} is the stationary measure of the stochastic nonlinear heat equation
\begin{equation*}
\partial_t \phi = - \frac{1}{2} \nabla S(\phi) + \xi = (\theta + \Delta / 2) \phi - \lambda \abs{\phi}^2 \phi + \xi
\end{equation*}
with space-time white noise $\xi$, which can be regarded as the Langevin equation for a time-dependent field $\phi$ with potential given by the action $S$ in \eqref{S_def}.
Constructing measures of the form \eqref{mu_def} by exhibiting them as stationary measures of stochastic nonlinear partial differential equations is the goal of stochastic quantization developed in \cite{nelson1966derivation, faris1982large, parisi1981perturbation, lebowitz1988statistical}. See for instance \cite{hairer2014theory,gubinelli2015paracontrolled, kupiainen2016renormalization,da2003strong} for recent developments.

In addition, Euclidean field theories are of great importance in the probabilistic Cauchy theory of nonlinear dispersive equations. For $K = 2$ and identifying $\R^2 \equiv \C$, the measure \eqref{mu_def} is formally invariant under the nonlinear Schrödinger (NLS) equation
\begin{equation} \label{NLS}
\ii \partial_t \phi = \frac{1}{2} \nabla S(\phi) = -(\theta + \Delta / 2) \phi + \lambda \abs{\phi}^2 \phi \,.
\end{equation}
Gibbs measures \eqref{mu_def} for the NLS \eqref{NLS} have proven a powerful tool for constructing almost sure global solutions with random initial data of low regularity. One considers the flow of the NLS \eqref{NLS} with random initial data distributed according to \eqref{mu_def}. The invariance of the measure \eqref{mu_def} under the NLS flow (in low dimensions) serves as a substitute for energy conservation, which is not available owing to the low regularity of the solutions.
See for instance the seminal works \cite{bourgain1994periodic,bourgain1994_z,bourgain1996_2d,bourgain1997invariant,bourgain2000_infinite_volume,lebowitz1988statistical} as well as \cite{Carlen_Froehlich_Lebowitz_2016, Carlen_Froehlich_Lebowitz_Wang_2019,GLV1,GLV2,McKean_Vaninsky2,NORBS,NRBSS,BourgainBulut4,BrydgesSlade,BurqThomannTzvetkov,FOSW,Thomann_Tzvetkov} and references given there for later developments.

The main difficulty in all of the works cited above is that, in dimensions larger than one, under the measure \eqref{mu_def} the field $\phi$ is almost surely a distribution of negative regularity, and hence the interaction term
\begin{equation*}
V(\phi) \deq \frac{\lambda}{2} \int_\Lambda \dd x \, \abs{\phi(x)}^4
\end{equation*}
in \eqref{S_def} is ill-defined. This is an \emph{ultraviolet} problem: a divergence for large wave vectors (i.e.\ spatial frequencies) producing small-scale singularities in the field. As the dimension $d = 1,2,3$ increases, the difficulty of making sense of the measure in \eqref{mu_def} increases significantly.

To outline the rigorous construction of the measure in \eqref{mu_def}, we introduce an ($\R^K$-valued) Gaussian free field on $\Lambda$ whose law $\P$ is the Gaussian measure on the space of fields with mean zero and covariance $(2 \kappa - \Delta)^{-1}$, where $\kappa > 0$ is some positive constant. Then we write
\begin{equation} \label{mu_P}
\mu(\dd \phi) = \frac{1}{\zeta} \ee^{-V(\phi)} \, \P(\dd \phi)
\end{equation}
for some normalization constant $\zeta > 0$. For $d = 1$, the right-hand side of \eqref{mu_P} makes sense as is, since, under 
$\P$, the field $\phi$ is almost surely a continuous function and hence $V(\phi)$ is almost surely nonnegative and finite. This provides a simple construction of \eqref{mu_def} for $d = 1$ and $\kappa = -\theta$. See e.g.\ \cite{Simon74} for a careful treatment.

For $d > 1$, the simple approach just sketched no longer works, since $\phi$ is almost surely of negative regularity, and the interaction term $V(\phi)$ has to be renormalized by subtracting suitably chosen infinite counterterms. The most elementary renormalization is Wick ordering of $V(\phi)$ with respect to the Gaussian measure $\P$, denoted by $\wick{\cdot}$ (see Appendix \ref{sec:wick}). After Wick ordering, the interaction term becomes
\begin{align} 
V(\phi) &= \frac{\lambda}{2} \int_\Lambda \dd x \, \wick{\abs{\phi(x)}^4}
\notag \\ \label{Wick_intro}
&= \frac{\lambda}{2} \int_\Lambda \dd x \pbb{\abs{\phi(x)}^4 - \frac{4 + 2K}{K} \, \E \qb{\abs{\phi(x)}^2}\, \abs{\phi(x)}^2 + \frac{K+2}{K} \, \E \qb{\abs{\phi(x)}^2}^2}\,,
\end{align}
where $\E$ denotes expectation with respect to $\P$. The second and third terms on the right-hand side of \eqref{Wick_intro} are infinite counterterms, which may be regarded as mass and energy renormalizations, respectively. Hence, for $d > 1$, the constant $\theta$ in \eqref{S_def} is formally $-\infty$. To make rigorous sense of \eqref{Wick_intro} in dimension $d=2$, one has to mollify $\phi$ by convolving it with an approximate delta function, and then show that as the mollifier is removed, the right-hand side of \eqref{Wick_intro} converges in $L^2(\P)$ (see Section \ref{sec:classical_field} below for more details). It is not hard to show that for $d = 2$ the renormalization on the right-hand side of \eqref{Wick_intro} yields a well-defined interaction term $V(\phi) \in L^2(\P)$. However, owing to the mass renormalization in \eqref{Wick_intro}, after Wick ordering, $V(\phi)$ is unbounded from below, and the integrability of $\ee^{-V(\phi)}$ with respect to $\P$ represents a nontrivial problem, which was successfully solved in the landmark work of Nelson \cite{nelson1973free, nelson1973probability}.

For $d = 3$, it is easy to see that, even after Wick ordering, $V(\phi)$ almost surely does not exist in $L^2(\P)$. Further, a simple expansion of the exponential $\ee^{-V(\phi)}$ in the two-point correlation function, $\frac{1}{\zeta} \int  \phi_i(x) \, \phi_j(y) \, \ee^{-V(\phi)} \, \P(\dd \phi)$, yields a divergent term already at second order, associated with the so-called sunset diagram of quantum field theory. Hence, a further mass renormalization of $V(\phi)$ is required, which results in a measure 
$\mu$ that is mutually singular with respect to the free-field Gaussian measure $\P$. The mathematically rigorous construction of the Euclidean $\phi^4_3$ theory, first achieved in the seminal work of Glimm and Jaffe \cite{glimm1973positivity},
is one of the major successes of the constructive field theory programme started in the sixties. By now, several different constructions of this theory have been developed, based on, first, phase cell expansions \cite{glimm1973positivity,feldman1976wightman,park1977convergence,glimm2012quantum}, then on renormalization group methods \cite{brydges1995short,gawedzki1985asymptotic,benfatto1978some}, later on correlation inequalities \cite{brydges1983new}, and, most recently, on paracontrolled calculus \cite{catellier2018paracontrolled, gubinelli2018pde}, as well as variational methods \cite{barashkov2020variational}.

For $d \geq 4$,  it is expected, and indeed proven in some cases,
that the $\phi^4_d$ theory is trivial: any renormalization of $V(\phi)$ resulting in a well-defined measure $\mu$ yields a (generalized free-field) Gaussian measure.
For $d \geq 5$, this triviality was proven in \cite{aizenman1982geometric, frohlich1982triviality}. Recently, the triviality of $\phi^4_4$ for $K = 1$ was established in \cite{aizenman2021marginal}.

\subsection{The $\phi^4_2$ theory as a limit of a Bose gas} \label{sec:overview_limit}
In this paper, we establish for the first time a relationship beween a local Euclidean field theory in dimension larger than one and an interacting quantum gas. We show that the complex Euclidean $\phi^4_2$ theory describes the limiting behaviour of an interacing Bose gas at positive temperature. The limiting regime is a high-density limit in a box\footnote{For conciseness, in this paper we assume that $\Lambda$ is the unit torus, although the actual shape of $\Lambda$ and the boundary conditions are not essential for our proof; see Remark \ref{rem:general_Lambda} below.} of fixed size, where the range of the interaction is much smaller than the diameter of the box. This result provides a rigorous derivation of the $\phi^4_2$ theory starting from a realistic model of statistical mechanics. Viewed differently, we introduce a new regularization of the $\phi^4_2$ theory in terms of an interacting Bose gas, in addition to the commonly used smooth mollifiers or lattice approximations.

To explain our result more precisely, we recall that a quantum system of $n$ spinless non-relativistic bosons of mass $m$ in $\Lambda$ is described by the Hamiltonian
\begin{equation*}
\bb H_n \deq - \sum_{i = 1}^n \frac{\Delta_i}{2m} + \frac{g}{2} \sum_{i,j = 1}^n v(x_i - x_j)
\end{equation*}
acting on the space $\cal H_n$ of square-integrable wave functions that are symmetric in their arguments $x_1, \dots, x_n$ and supported in $\Lambda^n$. Here $\Delta_i$ is the Laplacian in the variable $x_i$, $g$ is a coupling constant, and $v$ is a repulsive (i.e.\ with nonnegative Fourier transform) two-body interaction potential. We consider a system in the grand canonical ensemble at positive temperature, characterized by the density matrix
\begin{equation} \label{gc_density}
\frac{1}{Z} \bigoplus_{n \in \N} \ee^{-\beta(\bb H_n - \theta n)}
\end{equation}
acting on Fock space $\cal F = \bigoplus_{n \in \N} \cal H_n$, where $\beta < \infty$ is the inverse temperature, $\theta$ is the chemical potential, and $Z$ is a normalization factor. We refer to e.g.\ \cite{lieb2005mathematics,benedikter2016effective} for reviews on interacting Bose gases.

The limiting regime of this paper is obtained by introducing two parameters, $\nu, \epsilon > 0$, where $\nu = \frac{\beta}{m} = \sqrt{\beta g}$, and the potential $v$ is taken to be an approximate delta function of range $\epsilon$. We suppose that $\nu,\epsilon \to 0$ under the technical constraint $\epsilon \geq \exp\pb{- (\log \nu^{-1})^{1/2 - c}}$, for some constant $c > 0$. We show that there exists a suitable renormalization of the chemical potential $\theta \equiv \theta_\nu^\epsilon$ such that the reduced density matrices of the quantum state \eqref{gc_density} converge to the correlation functions of the field theory \eqref{mu_def}, \eqref{S_def}.

Previously, this result was obtained for $d = 1$ in \cite{lewin2015derivation, lewin2018gibbs, frohlich2017gibbs}, where, as explained in Section \ref{sec:overview}, no renormalization is required. In higher dimensions $d = 2,3$, the \emph{mean-field limit} was investigated in \cite{lewin2018classical, LNR3, FKSS_2020, frohlich2017gibbs, frohlich2019microscopic,sohinger2019microscopic}, where the parameter $\epsilon$ was fixed as $\nu \to 0$. The resulting limiting field theory differs from $\phi^4_d$ in that the interaction term $V(\phi)$ is nonlocal, given by a convolution with a bounded two-body interaction potential $v$. This nonlocal interaction term is considerably less singular than the local one of $\phi^4_2$ theory. 
The stronger singularity of $V(\phi)$ requires additional renormalization as compared to the nonlocal potential. This makes the local problem significantly more difficult than the nonlocal one. In particular, the renormalized interaction term $V(\phi)$ is unbounded from below, whereas in the nonlocal regime it is almost surely nonnegative.

The above lower bound on the range of the interaction $\epsilon$ is technical in nature (see Remark \ref{rem:range} below for a more detailed discussion on its origin). We expect that it can be improved, however at the cost of a considerably more complicated argument. In this paper we wish to emphasize that, using relatively simple methods, one can establish a connection between local Euclidean field theories and interacting quantum gases. We leave quantitative improvements of such results to future work.

Using our methods, we also extend the results on the mean-field limit for a nonlocal interaction term $V(\phi)$ in 
\cite{LNR3, FKSS_2020} from bounded two-body interaction potentials, $v$, to unbounded ones. 
Our integrability assumptions on the function $v$ are optimal, as given in \cite{bourgain1997invariant}. 
We refer to Section \ref{Extension of results to unbounded nonlocal interactions} below for details.

\subsection{Outlook}
The close relationship between Euclidean field theory and interacting Bose gases established in this paper leads 
to a web of conjectures concerning properties of 
$\phi^4_d$ theories inspired by results on Bose gases and, conversely and perhaps more interestingly, properties of 
interacting Bose gases inspired by known results on $\phi^4_d$ theories. In the following, we outline some of these conjectures.

We remark that an analysis very similar to the one in this paper yields an analogous relationship between 
the $\phi^4_2$ theory with $N$ complex components (that is, with $K = 2N$ real components) and an interacting Bose 
gas with $N$ species of identical Bosons; see Remark \ref{rem:species} below. 

\begin{enumerate}
\item
It is known (see \cite{berezin1961remark, albeverio2012solvable, geiler1995potentials}) that systems of non-relativisitic 
quantum particles moving in $d$-dimensional Euclidean 
space and interacting through  delta function potentials are equivalent to systems of 
\textit{free} (i.e.\ non-interacting) particles, provided that $d\geq 4$. Given the connection between interacting Bose
 gases and  $\phi^4_d$ theories exhibited in this paper, this suggests that the latter theories are
equivalent to free (i.e.\ Gaussian) field theories in dimensions $d\geq 4$, for a field $\phi$ with an arbitrary number 
of complex components.

\item
In $d=3$ dimensions, $\phi^4_d$ theories with $N$ complex components are known to undergo 
a phase transition accompanied by spontaneous $O(2N)$-symmetry breaking and the emergence of Goldstone 
bosons \cite{frohlich1976infrared}; 
(see also \cite{frohlich1983berezinskii}, as well as \cite{garban2021continuous} for recent results on related lattice models with disorder). 
Given our results for $d = 2$, as well as analogous results for $d = 3$ to appear in a future paper, the existence of a phase transition in the Euclidean field theory strongly suggests that
 translation-invariant Bose gases with repulsive two-body interactions in three dimensions exhibit 
 Bose-Einstein condensation accompanied by the appearance of massless quasi-particles with approximately relativistic dispersion at small wave vectors.

In two dimensions, the Mermin-Wagner theorem implies that such phase transitions do not exist, and the 
$O(2N)$-symmetry remains unbroken for arbitrary values of the coupling constant $\lambda$. A similar result is expected to hold for two-dimensional interacting Bose gases (and easy to see for ideal Bose gases).

\item
The \textit{one-component} complex $\phi^4_d$ - theory in $d = 2$ dimension is expected to exhibit a \textit{Berezinskii-Kosterlitz-Thouless transition.} This is 
rigorously known for the classical $XY$-model on a square lattice, which is the limiting theory of lattice 
$\phi^4_2$-theory, as $\lambda$ tends to $\infty$, with $\kappa=2\lambda$; see \cite{frohlich1981kosterlitz, frohlich1983berezinskii}.
 In view of the results proven in this paper, this suggests that two-dimensional Bose gases of \textit{one species} of 
 particles might exhibit a transition to a low-temperature phase where reduced density matrices exhibit 
 \textit{slow decay}, analogous to the Berezinskii-Kosterlitz-Thouless transition.
 
 In contrast, for a two-dimensional $\phi^4_d$ theory with two or more complex components, with an $O(2N)$-symmetry, it is expected that connected correlations exhibit exponential decay for arbitrary values of 
 the coupling constant $\lambda$; see \cite{polyakov1975interaction}. This suggests that two-dimensional Bose gases of several species 
 of identical particles exhibit rapidly decaying correlations at all temperatures and densities.
 
 \item
 For $\phi^4_d$ theories with $N$ complex components, 
  there exists a systematic $1/N$-expansion; see \cite{itzykson1991statistical1, itzykson1991statistical2} and \cite[Chapter 30]{zinn2021quantum}. The model obtained in the limit, as 
  $N\rightarrow \infty$, is the spherical model, which is exactly solved. It is tempting to extend 
 the method of the $1/N$-expansion to Bose gases of $N$ species of identical particles interacting 
 through two-body interactions of strength $O(1/N)$. The model obtained in the limit,
 as $N \rightarrow \infty$, appears to be equivalent to an ideal Bose gas, but with a renormalized chemical
 potential. In attempting to prove Bose-Einstein condensation for translation-invariant
interacting Bose gases, therefore, it seems judicious to begin by studying Bose gases with a large number of species of identical particles. The connection between Bose-Einstein condensation and phase transitions in classical field theory has been discussed in e.g.\ \cite{baym2001bose, holzmann2003condensate}.
\end{enumerate}

\section{Setup and results}

\subsection{Classical field theory} \label{sec:classical_field}
In this subsection we define the Euclidean field theory and its correlation functions. We note that the measure $\mu$  from \eqref{mu_def} can be formally viewed as the thermal equilibrium measure of a \textit{classical} field theory with Hamilton function given by $S(\phi)$ from \eqref{S_def}.
We work on the $d$-dimensional torus $\Lambda \deq [-1/2,1/2)^d$. We use the Euclidean norm $\abs{\cdot}$ for elements of $\Lambda$ regarded as a subset of $\R^d$. We use the shorthand $\int \dd x \, (\cdot)\deq \int_{\Lambda} \dd x \, (\cdot)$ to denote integration over $\Lambda$ with respect to Lebesgue measure. We abbreviate $\cal H \deq L^2(\Lambda; \C)$ and denote by $\scalar{\cdot}{\cdot}$ the inner product of the space $\cal H$, which is by definition linear in the second argument. On $\cal H$ we use the standard Laplacian $\Delta$ with periodic boundary conditions.

The classical free field $\phi$ is by definition the complex-valued Gaussian field with covariance $(\kappa - \Delta/2)^{-1}$, where $\kappa > 0$ is a constant. Explicitly, the free field may be constructed as follows.
We use the spectral decomposition $\kappa - \Delta/2 = \sum_{k \in \Z^d} \lambda_k u_k u_k^*$, with eigenvalues $\lambda_k > 0$ and normalized eigenfunctions $u_k \in \cal H$ (see also \eqref{lambda_k} below). Let $X = (X_k)_{k \in \Z^d}$ be a family of independent standard complex Gaussian random variables\footnote{We recall that $Z$ is a standard complex Gaussian if it is Gaussian and satisfies $\E Z = 0$, $\E Z^2 = 0$, and $\E \abs{Z^2} = 1$, or, equivalently, if it has law $\pi^{-1} \ee^{- \abs{z}^2} \dd z$ on $\C$, where $\dd z$ denotes Lebesgue measure.}, whose law and associated expectation are denoted by $\P$ and $\E$, respectively. The \emph{classical free field} is then given by
\begin{equation*}
\phi = \sum_{k \in \Z^d} \frac{X_k}{\sqrt{\lambda_k}} \, u_k\,,
\end{equation*}
which is easily seen to converge\footnote{In fact, an application of Wick's rule shows that the convergence holds in $L^m$ for any $m < \infty$.} in $L^2(\P)$ of the $L^2$-Sobolev space $H^{1 - d/2 - c}$ for any $c > 0$.

In order to define the interacting theory, it is necessary to regularize the field $\phi$ by convolving it with a smooth mollifier. To that end, choose a nonnegative function $\vartheta \col \R^d \to \R_+$ of rapid decay satisfying $\vartheta(0) = 1$, and for $0 < N < \infty$ define the regularized field
\begin{equation} \label{def_phi_N}
\phi_N \deq \sum_{k \in \Z^d} \frac{X_k}{\sqrt{\lambda_k}} \sqrt{\vartheta(k/N)}\, u_k\,,
\end{equation}
which is almost surely a smooth function on $\Lambda$. We define the regularized interaction
\begin{equation*}
V_N \deq \frac{1}{2} \int \dd x \, \wick{\abs{\phi_N(x)}^4}\,,
\end{equation*}
where $\wick{\cdot}$ denotes Wick ordering with respect to the Gaussian measure $\P$ (see Appendix \ref{sec:Wick}). Explicitly,
\begin{equation*}
\wick{\abs{\phi_N(x)}^4} = \abs{\phi_N(x)}^4 - 4 \E \qb{\abs{\phi_N(x)}^2} \, \abs{\phi_N(x)}^2 + 2 \E \qb{\abs{\phi_N(x)}^2}^2\,.
\end{equation*}
Here, the deterministic factor $\E \qb{\abs{\phi_N(x)}^2} = \sum_{k \in \Z^d} \frac{\vartheta(k/N)}{\lambda_k}$ diverges as $N \to \infty$ for $d > 1$.

For $d = 2$, using Wick's theorem, it is easy to see that $V_N$ converges as $N \to \infty$ in $L^2(\P)$ to a random variable, denoted by $V$, which does not depend on the choice of $\vartheta$. See e.g.\ \cite[Lemma 1.5]{frohlich2017gibbs} for details. The interacting field theory is given as the probability measure
\begin{equation} \label{field theory}
\frac{1}{\zeta} \,\ee^{-V} \, \dd \P \,, \qquad \zeta \deq \E[\ee^{-V}]\,.
\end{equation}
By the well-known Nelson bounds \cite{nelson1973probability, nelson1973free} mentioned in Section \ref{sec:overview}, $\ee^{-V}$ is integrable with respect to $\P$.

We characterize the interacting field theory through its correlation functions, defined as follows.
For $p \in \N$ and $\f x,\tilde{\f x} \in \Lambda^p$, we define the \emph{$p$-point correlation function} as
\begin{equation}
\label{gamma_p}
(\gamma_p)_{\f x,\tilde{\f x}} \deq \frac{1}{\zeta} \, \E \qB{\bar{\phi}(\tilde x_1) \cdots \bar{\phi}(\tilde x_p)\,\phi(x_1)\cdots \phi(x_p) \, \ee^{-V}}\,.
\end{equation}
which is the $2p$-th moment of the field $\phi$ under the probability measure \eqref{field theory}. This measure is sub-Gaussian, and is hence determined by its moments $(\gamma_p)_{p \in \N^*}$. (Note that any moment containing a different number of $\bar \phi$s and $\phi$s vanishes by invariance of the measure \eqref{field theory} under the gauge transformation $\phi \mapsto \alpha \phi$, where $\abs{\alpha} = 1$.)

As explained in \cite[Section 1.5]{FKSS_2020}, the correlation function $\gamma_p$ is divergent on the diagonal, even for the free field. Hence, for instance, it cannot be used to analyse the distribution of the mass density $\abs{\phi(x)}^2$. As in \cite[Section 1.5]{FKSS_2020}, we remedy this issue by introducing the \emph{Wick-ordered $p$-point correlation function}
\begin{equation}
\label{hat_gamma_p}
(\wh{\gamma}_p)_{\f x,\tilde{\f x}} \deq \frac{1}{\zeta}\, \E \qB{\wick{\bar{\phi}(\tilde x_1) \cdots \bar{\phi}(\tilde x_p)\,\phi(x_1)\cdots \phi(x_p)} \, \ee^{-V}}\,,
\end{equation}
which has a regular behaviour on the diagonal. The Wick-ordered correlation function \eqref{hat_gamma_p} can be expressed explicitly in terms of the correlation functions \eqref{gamma_p} and the correlation functions of the free field; see \eqref{hat_gamma_p_2} below.

\subsection{Quantum many-body system}
In this subsection we define the quantum many-body system and its reduced density matrices. For $n \in \N$, we denote by $P_n$ the orthogonal projection onto the symmetric subspace of $\cal H^{\otimes n}$; explicitly, for $\Psi_n \in \cal H^{\otimes n}$,
\begin{equation} \label{def_Pn}
P_n \Psi_n(x_1, \dots, x_n) \deq \frac{1}{n!} \sum_{\pi \in S_n} \Psi_n(x_{\pi(1)}, \dots, x_{\pi(n)})\,,
\end{equation}
where $S_n$ is the group of permutations on $\{1, \dots, n\}$. For $n \in \N^*$, we define the $n$-particle space as $\cal H_n \deq P_n \cal H^{\otimes n}$. We define Fock space as the Hilbert space $\cal F \equiv \cal F (\cal H) \deq \bigoplus_{n \in \N} \cal H_n$. We denote by $\tr_{\cal F}(X)$ the trace of an operator $X$ acting on $\cal F$. For $f \in \cal H$ we define the bosonic annihilation and creation operators $a(f)$ and $a^*(f)$ on $\cal F$ through their action on a dense set of vectors $\Psi = (\Psi_n)_{n \in \N} \in \mathcal{F}$ as
\begin{align}
\label{def_b2}
\pb{a(f) \Psi}_n(x_1, \dots, x_n) &= \sqrt{n+1} \int \dd x \, \bar f(x) \, \Psi_{n+1} (x,x_1, \dots, x_n)\,,
\\
\label{def_b1}
\pb{a^*(f) \Psi}_n(x_1, \dots, x_n) &= \frac{1}{\sqrt{n}} \sum_{i = 1}^n f(x_i) \Psi_{n - 1}(x_1, \dots, x_{i - 1}, x_{i+1}, \dots, x_n)
\,.
\end{align}
The operators $a(f)$ and $a^*(f)$ are unbounded closed operators on $\cal F$, and are each other's adjoints. They satisfy the canonical commutation relations
\begin{equation} \label{CCR_b}
[a(f), a^*(g)] = \scalar{f}{g} \, 1 \,, \qquad [a(f), a(g)] = [a^*(f), a^*(g)] =0\,,
\end{equation}
where $[X,Y] \deq XY - YX$ denotes the commutator. We regard $a$ and $a^*$ as operator-valued distributions and use the notations
\begin{equation} \label{phi_tau f}
a(f) = \int \dd x \, \bar f(x) \, a(x)\,, \qquad
a^*(f) = \int \dd x \, f(x) \, a^*(x)\,.
\end{equation}
The distribution kernels $a^*(x)$ and $a(x)$ satisfy the canonical commutation relations
\begin{equation} \label{CCR}
[a(x),a^*(\tilde x)] = \delta(x - \tilde x) \,, \qquad
[a(x),a(\tilde x)] = [a^*(x),a^*(\tilde x)] = 0\,.
\end{equation}
For $\nu > 0$, we define the free quantum Hamiltonian $H^{(0)} \equiv H^{(0)}_\nu$ through
\begin{equation} \label{free_Hamiltonian_H}
H^{(0)} \deq \nu \int \dd x \, a^*(x) ((\kappa - \Delta / 2) a) (x)\,.
\end{equation}

To describe the interaction potential of the Bose gas, we choose $v \col \R^d \to \R$ to be an even, smooth, compactly supported function of positive type\footnote{This means that the Fourier transform of $v$ is a positive measure. Note that we do not assume $v$ to be pointwise nonnegative.} whose integral is equal to one. For $\epsilon > 0$ we define the rescaled interaction potential on $\Lambda$ as
\begin{equation}
\label{v_epsilon}
v^\epsilon(x) = \sum_{n \in \Z^d} \frac{1}{\epsilon^d} \, v \pbb{\frac{x - n}{\epsilon}}\,.
\end{equation}
For $\epsilon, \nu > 0$ we define the interacting quantum Hamiltonian $H \equiv H^{\epsilon}_{\nu}$ through
\begin{equation} \label{def_H}
H \deq 
H^{(0)} + \frac{\nu^2}{2} \int \dd x \, \dd \tilde x \, a^*(x) a(x)  \, v^\epsilon(x - \tilde x) \,  a^*(\tilde x) a (\tilde x)
- \nu \alpha^\epsilon_\nu \int \dd x \, a^*(x) a(x)+ \theta^\epsilon_\nu 
\,,
\end{equation}
where $\alpha^\epsilon_\nu$ and  $\theta^\epsilon_\nu$ are real renormalization parameters that we shall define shortly in \eqref{alpha_beta_def} below.

Using \eqref{def_H}, the quantum grand canonical density matrix from \eqref{gc_density} can be expressed as the operator
\begin{equation*}
\eqref{gc_density} =  \frac{\ee^{-H}}{Z}\,, \qquad Z \deq \tr_{\cal F} (\ee^{-H})\,,
\end{equation*}
where $Z$ is the grand canonical partition function. Analogously, the free grand canonical partition function is
\begin{equation*}
Z^{(0)} \deq \tr_{\cal F}(\ee^{-H^{(0)}})\,.
\end{equation*}
We shall also use the relative partition function
\begin{equation} \label{Z^epsilon_quantum}
\cal Z \deq \frac{Z}{Z^{(0)}}\,.
\end{equation}

In order to define the renormalization parameters $\alpha^\epsilon_\nu$ and  $\theta^\epsilon_\nu$, we introduce the Green function $G$ of the free field $\phi$, i.e.\ the integral kernel of the operator $(\kappa - \Delta / 2)^{-1}$. Since $\kappa - \Delta/2$ is invariant under translations, we can write $G_{x,y} = G(x - y)$. Explicitly, in the sense of distributions,
\begin{equation*}
G(x - y) = \E [\phi(x) \bar \phi(y)]\,.
\end{equation*}
The Green function $G$ exhibits a logarithmic singularity at the origin (see Lemma \ref{lem:G_smoothness} below).
Moreover, we denote by
\begin{equation} \label{rho_nu_definition}
\varrho_\nu \deq \nu \tr_{\cal F} \pbb{a^*(0) a(0) \, \frac{\ee^{-H^{(0)}}}{Z^{(0)}}}
\end{equation}
the expected rescaled particle density in the free quantum state. Then we set
\begin{equation} \label{alpha_beta_def}
\alpha^\epsilon_\nu \deq \varrho_\nu + \tau^\epsilon\,, \qquad
\theta^\epsilon_\nu \deq \frac{1}{2} \varrho_\nu^2 + \tau^\epsilon \varrho_\nu - E^\epsilon\,,
\end{equation}
where
\begin{equation}
\label{tau_epsilon}
\tau^\epsilon \deq \int \dd x \, v^\epsilon(x) \, G(x)\,,
\qquad
E^\epsilon \deq \frac{1}{2} \int \dd x  \, v^\epsilon (x) \, G(x)^2\,.
\end{equation}
The parameter $\alpha^\epsilon_\nu$ describes a renormalization of the chemical potential, and $\theta^\epsilon_\nu$ corresponds to an energy renormalization.  As $\epsilon ,\nu \to 0$, the renormalization of the chemical potential behaves as $\alpha_\nu^\epsilon \to +\infty$. We remark that, using the quantities \eqref{rho_nu_definition} and \eqref{tau_epsilon}, we can rewrite the Hamiltonian \eqref{def_H} in the form
\begin{multline}
\label{Hamiltonian_H}
H = H^{(0)}  + \frac{1}{2} \int \dd x \, \dd \tilde x \, \pb{\nu a^*(x) a(x) - \varrho_\nu} \, v^\epsilon(x - \tilde x) \, \pb{ \nu a^*(\tilde x) a (\tilde x) - \varrho_\nu}
\\
- \tau^\epsilon \, \int \dd x \, \pb{\nu a^*(x) a(x) - \varrho_\nu} - E^\epsilon\,.
\end{multline}

Next, we define the \emph{$p$-particle reduced density matrix} as
\begin{equation} \label{Gamma_p}
(\Gamma_p)_{\f x,\tilde{\f x}} \deq
\tr_{\cal F} \Biggl(a^{*}(\tilde{x}_1) \cdots a^{*}(\tilde{x}_p)\,a(x_1)\cdots a(x_p)\,\frac{\ee^{-H}}{Z}\Biggr)\,.
\end{equation}
As for the correlation function \eqref{gamma_p} and its Wick-ordered version \eqref{hat_gamma_p}, we would like to replace \eqref{Gamma_p} with its Wick-ordered version. To that end, we regard the expressions \eqref{gamma_p} and \eqref{hat_gamma_p} as integral kernels of operators acting on $\cal H_p$, and observe that  (see \cite[Lemma A.4]{FKSS_2020})
\begin{equation}
\label{hat_gamma_p_2}
\wh{\gamma}_p =\sum_{k=0}^{p}\binom{p}{k}^2\,(-1)^{p-k}\,P_p(\gamma_k \otimes \gamma^{(0)}_{p-k})P_p\,,
\end{equation}
where $\gamma^{(0)}_{m}$ denotes the $m$-point correlation function from \eqref{gamma_p} with $V=0$. In analogy with \eqref{hat_gamma_p_2}, we therefore define the \emph{Wick-ordered $p$-particle reduced density matrix} as
\begin{equation}
\label{hat_Gamma_p_2}
(\wh{\Gamma}_p)_{\f x,\tilde{\f x}} \deq \sum_{k=0}^{p}\binom{p}{k}^2\,(-1)^{p-k}\,P_p(\Gamma_k \otimes \Gamma^{(0)}_{p-k})P_p\,,
\end{equation}
where $\Gamma^{(0)}_{m}$ denotes the $m$-particle reduced density matrix of the free grand canonical density matrix $\ee^{-H^{(0)}} / Z^{(0)}$. 
(For an interpretation of \eqref{hat_Gamma_p_2} as a result of Wick ordering \eqref{Gamma_p} with respect to the free field in the functional integral representation of quantum many-body theory, we refer the reader to the discussion in \cite[Section 1.7]{FKSS_2020}).

\subsection{Results}

We may now state our main result.

\begin{theorem} \label{thm:main}
Suppose that $d = 2$ and $\epsilon \equiv \epsilon(\nu)$ satisfies
\begin{equation} \label{epsilon_assumption}
\epsilon \geq \exp\pb{- (\log \nu^{-1})^{1/2 - c}}
\end{equation}
for some constant $c > 0$. Then as $\epsilon, \nu \to 0$ we have the convergence of the partition function
\begin{equation}
\cal Z \to \zeta
\end{equation}
and of the Wick-ordered correlation functions
\begin{equation}
\nu^p \, \wh \Gamma_p \overset{C}{\longrightarrow} \wh \gamma_p
\end{equation}
for all $p \in \N$, where $\overset{C}{\longrightarrow}$ denotes convergence in the space of continuous functions on $\Lambda^p \times \Lambda^p$ with respect to the supremum norm.
\end{theorem}

We refer to \cite[Section 1.5]{FKSS_2020} for an in-depth discussion on applications of Theorem \ref{thm:main}. In particular, Theorem \ref{thm:main} yields the following result for unrenormalized correlation functions.

\begin{corollary}
Under the assumptions of Theorem \ref{thm:main},
\begin{equation*}
\nu^p \Gamma_p \overset{L^r}{\longrightarrow} \gamma_p
\end{equation*}
for all $p \in \N$ and $r < \infty$, where $\overset{L^r}{\longrightarrow}$ denotes convergence in the $L^r(\Lambda^p \times \Lambda^p)$-norm.
\end{corollary}

Another application of Theorem \ref{thm:main} is the convergence of the joint distribution of the Wick-ordered quantum particle densities $a^*(x) a(x)$ to those of the Wick-ordered mass densities $\abs{\phi(x)}^2$; see \cite[Theorem 1.4]{FKSS_2020}.

\begin{remark} \label{rem:general_Lambda}
In this paper we set $\Lambda$ to be the unit torus for definiteness, but our methods extend without complications to more general domains and boundary conditions. In particular, they also apply to the full space $\R^2$ with one-body Hamiltonian $-\Delta/2 + U(x)$, where the particles are confined by a suitable external potential $U \col \R^2 \to \R$ satisfying $U(x) \geq b \abs{x}^\theta$ for some $b > 0$ and $\theta > 2$. We refer to \cite[Sections 1.6 and 7]{FKSS_2020} and \cite[Section 3.2]{LNR3} for an in-depth discussion of the analogous extension for the mean-field scaling. The corresponding counterterm problem, which relates the bare and renormalized external potentials, was formulated and solved in \cite{frohlich2017gibbs} for the mean-field scaling. It, along with the arguments of \cite[Section 5]{frohlich2017gibbs}, can be adapted to the local scaling of the current paper; we omit further details.
\end{remark}

\begin{remark} \label{rem:species}
The proof of Theorem \ref{thm:main} can be extended to establish the convergence of the interacting Bose gas of $N$ species of identical Bosons to the $\phi^4_2$ theory with $N$ complex components. (Theorem \ref{thm:main} corresponds to $N = 1$.) More precisely, we introduce the species index $i = 1, \dots, N$, and augment the creation and annihilation operators to $a_i^*(x), a_i(x)$ satisfying the canonical commutation relations
\begin{equation} \label{CCR_species}
[a_i(x),a_j^*(\tilde x)] = \delta_{ij} \delta(x - \tilde x) \,, \qquad
[a_i(x),a_j(\tilde x)] = [a_i^*(x),a_j^*(\tilde x)] = 0\,,
\end{equation}
which generalize \eqref{CCR}. The Hamiltonian from \eqref{free_Hamiltonian_H} and \eqref{def_H} is generalized to
\begin{equation*}
H^{(0)} \deq \nu \sum_{i = 1}^N \int \dd x \, a_i^*(x) ((\kappa - \Delta / 2) a_i) (x)
\end{equation*}
and
\begin{equation*}
H \deq 
H^{(0)} + \frac{\nu^2}{2} \sum_{i = 1}^N \int \dd x \, \dd \tilde x \, a_i^*(x) a_i(x)  \, v^\epsilon(x - \tilde x) \,  a_i^*(\tilde x) a_i (\tilde x)
- \nu \alpha^\epsilon_\nu \sum_{i = 1}^N \int \dd x \, a_i^*(x) a_i(x)+ \theta^\epsilon_\nu
\,.
\end{equation*}
Then we find that the reduced density matrices of the $N$-species quantum Bose gas converge to the correlation functions of $\phi^4_2$ theory with $N$ complex components, in the sense of Theorem \ref{thm:main}.
\end{remark}

\begin{remark} \label{rem:range}
We conclude this section with a discussion on the technical condition \eqref{epsilon_assumption} on the range $\epsilon$ of the interaction potential. It is instructive to compare the right-hand side of \eqref{epsilon_assumption} to the typical inter-particle distance, which we claim is of order
\begin{equation} \label{particle-distance}
\ell = \nu^{1/2} (\log \nu^{-1})^{-1/2}\,.
\end{equation}
To show \eqref{particle-distance}, it suffices to show that the expected number of particles, given by $\tr (\Gamma_1)$, is of order $\nu^{-1} \log \nu^{-1}$. By Theorem \ref{thm:main} and the definition \eqref{hat_Gamma_p_2}, we find $\tr (\Gamma_1) = \tr (\Gamma_1^{(0)}) + O(\nu^{-1})$. Hence, it suffices to show that $\tr (\Gamma_1^{(0)})$ is of order $\nu^{-1} \log \nu^{-1}$, which follows using the Wick theorem for quasi-free quantum states (see e.g.\ \cite[Remark 1.5 and Lemma 2.9]{FKSS_2020}), since $\tr (\Gamma_1^{(0)}) = \tr \pb{\frac{1}{\ee^{\nu (\kappa - \Delta/2)} - 1}} \asymp \nu^{-1} \log \nu^{-1}$, as claimed.


Comparing \eqref{particle-distance} and the right-hand side of \eqref{epsilon_assumption}, we conclude that the range of the interaction $\epsilon$ may vanish much faster than any power of $(\log \ell^{-1})^{-1}$ but much slower than any power of $\ell$. As stated in Section \ref{sec:overview_limit}, we expect that it can be improved, however at the cost of a considerably more complicated argument. We leave such quantitative improvements to future work, focusing here on the first result of this kind while aiming for a relatively simple proof.

The origin of the condition \eqref{epsilon_assumption} in our proof arises from controlling oscillatory integrals. It can be traced to the functional Fourier representation from Lemmas \ref{Partition_function_rate_of_convergence_1}  and \ref{Partition_function_rate_of_convergence_2} below. There, the partition function is expressed in terms of an integral over a Gaussian field, where the integrand includes a phase with a diverging prefactor. To compensate this rapidly oscillating phase and obtain a quantity of order one, the integral is multiplied by the large factor 
\begin{equation*}
\ee^{(\tau^\epsilon)^2 / 2+E^{\epsilon}}\,,
\end{equation*}
(see \eqref{T^epsilon}) which is propagated throughout the estimates of Section \ref{The rate of convergence}, for example in Lemma \ref{Partition_function_rate_of_convergence_3}. This large factor needs to be compensated by powers of $\nu$ which arise from our error estimates, leading to the lower bound on $\epsilon$ in terms of $\nu$.

\end{remark}

\section{Structure of the proof}

The rest of this paper is devoted to the proof of Theorem \ref{thm:main}. We begin with a short section that lays out the general strategy. We use $c,C$ to denote generic positive constants, which may change from one expression to the next, and may depend on fixed parameters.  We write $x \lesssim y$ or $x=O(y)$ to mean $x \leq C y$. If $C$ depends on a parameter $\alpha$, we write $x \lesssim_{\alpha} y$, $x \leq C_{\alpha} y$, or $x=O_{\alpha}(y)$. We abbreviate $[n] = \{1,\dots,n\}$.

We shall need two different interacting field theories approximating \eqref{field theory}, obtained by replacing the interaction $V$ with regularized variants, denoted by $W^\epsilon$ and $V^\epsilon$, respectively. They are defined by
\begin{align}
\label{W^epsilon}
W^\epsilon &\deq \frac{1}{2} \int \dd x \, \dd \tilde x \, \wick{\abs{\phi(x)}^2} \, v^\epsilon(x - \tilde x) \, \wick{\abs{\phi(\tilde x)}^2} - \tau^\epsilon \int \dd x \, \wick{\abs{\phi(x)}^2} - E^\epsilon\,,
\\ \label{V^epsilon}
V^\epsilon &\deq \frac{1}{2} \int \dd x \, \dd \tilde x \, v^\epsilon(x - \tilde x) \, \wick{ \abs{\phi(x)}^2 \, \abs{\phi(\tilde x)}^2}\,.
\end{align}
The rigorous construction of the random variables $W^\epsilon, V^\epsilon$ proceeds exactly like that of $V$ explained in Section \ref{sec:classical_field}: one introduces truncated versions $W^\epsilon_N, V^\epsilon_N$ defined in terms of the truncated free field $\phi_N$ (see e.g.\ \eqref{def_V_e_M} below), and proves using Wick's theorem that as $N \to \infty$ they converge in $L^2(\P)$ to their respective  limits $W^\epsilon, V^\epsilon$. Throughout the following, we shall make use of such constructions of Wick-ordered functions of the free field without further comment. The integrability of $\ee^{-W^\epsilon}$ and $\ee^{-V^\epsilon}$ is established in Section \ref{sec:field_theory} below.

To emphasize the dependence of the quantities \eqref{field theory} and \eqref{hat_gamma_p} on the interaction $V$, we sometimes include the interaction $V$ in our notation as a superscript, writing $\zeta^V$ and $\wh \gamma_p^V$, respectively.

The proof consists of two main steps.
\begin{description}
\item[Step 1.]
We compare $\cal Z$ and $\nu^p \, \wh \Gamma_p$ with $\zeta^{W^\epsilon}$ and $\wh \gamma_p^{W^\epsilon}$, respectively, in the limit $\nu,\epsilon \to 0$ under the condition \eqref{epsilon_assumption}.
\item[Step 2.]
We compare $\zeta^{W^\epsilon}$ and $\wh \gamma_p^{W^\epsilon}$ with $\zeta^{V}$ and $\wh \gamma_p^{V}$, respectively, in the limit $\epsilon \to 0$. This step is done by passing via the further intermediate interaction $V^\epsilon$.
\end{description}

Step 1 relies on a quantitative analysis of the infinite-dimensional saddle point argument for the functional integral introduced in \cite{FKSS_2020}.

Step 2 relies on three main ingredients. First, we show integrability of $\ee^{-V^\epsilon}$, uniformly in $\epsilon$. Second, we use that $V^\epsilon - W^\epsilon$ is small in $L^2(\P)$ and it lies in the second polynomial chaos (see Section \ref{sec:hypercontr} below), which allows us to deduce integrability of $\ee^{-W^\epsilon}$ by expansion in $V^\epsilon - W^\epsilon$ and hypercontractive moment bounds. Third, to obtain uniform control on the Wick-ordered correlation functions, we use Gaussian integration by parts, analogous to Malliavin calculus, to derive a representation of the correlation functions in terms of expectations of derivatives of the interaction potential.

The results of these two steps are summarized in the following two propositions.

\begin{proposition} \label{prop:step1}
Suppose that $d = 2$ and that $\nu,\epsilon \to 0$ under the constraint \eqref{epsilon_assumption}. Then $\cal Z - \zeta^{W^\epsilon} \to 0$. Moreover, for all $p \in \N^*$,
\begin{equation}
\normB{\nu^p \, \wh \Gamma_p - \wh \gamma_p^{W^\epsilon}}_C \to 0\,.
\end{equation}
\end{proposition}

\begin{proposition} \label{prop:step2}
Suppose that $d = 2$ and that $\epsilon \to 0$. Then $\zeta^{W^\epsilon} \to \zeta^V$. Moreover, for all $p \in \N^*$,
\begin{equation} \label{convergence_gamma_field}
\normB{\wh \gamma_p^{W^\epsilon} - \wh \gamma_p^V}_C \to 0\,.
\end{equation}
\end{proposition}

We remark that Proposition \ref{prop:step1} holds also for $d = 3$, with the same proof, provided that the condition \eqref{epsilon_assumption} is suitably adjusted. We refer to Section \ref{The rate of convergence} for more details and for the proof.

\section{Proof of Proposition \ref{prop:step2}} \label{sec:field_theory}

In this section we prove Proposition \ref{prop:step2}. We set $d = 2$ throughout.

\subsection{$L^2$-estimates} \label{sec_L2estimates}
In this subsection we derive $L^2$-estimates for the differences $V^\epsilon - V$ and $V^\epsilon - W^\epsilon$.

\begin{lemma} \label{lem:R_variance}
We have $\norm{V^\epsilon - W^\epsilon}_{L^2(\P)} \to 0$ as $\epsilon \to 0$.
\end{lemma}
\begin{proof}
A straightforward calculation using \eqref{wick_expanded} below yields
\begin{align}
V^\epsilon &= \frac{1}{2} \int \dd x \, \dd \tilde x \, \wick{\abs{\phi(x)}^2} \, v^\epsilon(x - \tilde x) \, \wick{\abs{\phi(\tilde x)}^2}
 - \int \dd x \, \dd \tilde x \, v^\epsilon(x - \tilde x) \, G(x - \tilde x) \, \bar \phi(x) \phi(\tilde x) + E^\epsilon
\notag \\
 &= 
 \frac{1}{2} \int \dd x \, \dd \tilde x \, \wick{\abs{\phi(x)}^2} \, v^\epsilon(x - \tilde x) \, \wick{\abs{\phi(\tilde x)}^2}
 - \int \dd x \, \dd \tilde x \, v^\epsilon(x - \tilde x) \, G(x - \tilde x) \, \wick{\bar \phi(x) \phi(\tilde x)} - E^\epsilon
\notag \\ \label{W_V_estimate}
&= W^\epsilon -  \int \dd x \, \dd \tilde x \, v^\epsilon(x - \tilde x) \, G(x -  \tilde x) \, \pb{\wick{\bar \phi(x) \phi(\tilde x)} - \wick{\abs{\phi(x)}^2}}\,.
\end{align}
By Lemma \ref{lem:Wick} below (see also Example \ref{ex:two_blocks} below), we find
\begin{align}
\norm{V^\epsilon - W^\epsilon}_{L^2(\P)}^2 &= \int \dd x \, \dd \tilde x \, \dd y \, \dd \tilde y \, v^\epsilon(x - \tilde x) \, v^\epsilon(y - \tilde y) \, G(\tilde x - x) \, G(\tilde y - y)
\notag \\ \label{R_e_square}
&\qquad \times \pb{G(x - \tilde y) - G(x - y)} \pb{G(\tilde x - y) - G(x - y)}\,.
\end{align}
We emphasize the crucial double cancellation on the second line of \eqref{R_e_square}, which will ensure convergence of the right-hand side, even though the first line of the right-hand side on its own is divergent as $\epsilon \to 0$.
From Lemma \ref{lem:G_smoothness} below, we find
\begin{equation} \label{G-G_estimate}
|G (x- \tilde{y}) - G(x-y)| \lesssim |y-\tilde{y}| + \absbb{ \log \frac{|x-\tilde{y}|}{|x-y|}} 
\end{equation}
and similarly for $| G(\tilde{x} - y) - G (x-y)|$. Switching to new integration variables $h = (\tilde x- x)/\epsilon$, $k = (y - \tilde{y})/\epsilon$, and $z = x-y$, we obtain 
\[ \| V^\epsilon - W^\epsilon \|_{L^2 (\mathbb{P})}^2 \lesssim \int \dd h \,  \dd k \, \dd z \, v(h) \, v(k) \, G(\epsilon h) \, G(\epsilon k) \pbb{ \epsilon |h| + \absbb{\log \frac{|z + \epsilon  h|}{|z|}}}  \pbb{ \epsilon |k| + \absbb{ \log \frac{|z +\epsilon k|}{|z|} } } \]
We multiply out the two parentheses on the right-hand side and treat each of the four terms separately. The term arising from $\epsilon \abs{h} \epsilon \abs{k}$ is easily estimated by $O_\alpha(\epsilon^{2 \alpha})$ for any $\alpha \in (0,1)$, using that $G(x) \abs{x} \lesssim_\alpha \abs{x}^\alpha$, by Lemma \ref{lem:G_smoothness}.

For the other three terms containing the logarithmic factor, we use
\begin{equation} \label{log_estimate}
\absbb{ \log \frac{|x + y|}{|x|}} \lesssim
\begin{cases}
| \log |x+ y| | + | \log |x| | & \text{if } |x| \leq 2 |y|
\\
\frac{|y|}{|x|}  & \text{if } |x| > 2 |y|\,,
\end{cases}
\end{equation}
We estimate the mixed terms, for any $\alpha \in [0,1)$, as
\begin{align*}
&\mspace{-10mu}\int \dd h \,  \dd k \, \dd z \, v(h) \, v(k) \, G(\epsilon h) \, G(\epsilon k) \,  \epsilon |k| \, \absbb{ \log \frac{|z +\epsilon h|}{|z|} }
\\
&\lesssim \int \dd k \,  v(k) \, G(\epsilon k) \, \epsilon |k| \int  \dd h  \, v(h)  \, G(\epsilon h) \int \dd z \,   \pbb{\pb{| \log |z+ \epsilon h| | + | \log |z| |} \ind{|z| \leq 2 \epsilon |h|} + \frac{\epsilon |h|}{|z|} \ind{|z| > 2 \epsilon |h|}}
\\
&\lesssim_\alpha
\int \dd k \,  v(k) \, G(\epsilon k) \, \epsilon |k| \int \dd h \, v(h)  \, G(\epsilon h) \pB{\epsilon^{2 \alpha} \abs{h}^{2 \alpha} + \epsilon \abs{h}}
\\
&\lesssim_\alpha \epsilon^{2 \alpha}\,.
\end{align*}
Finally, we estimate, for any $\alpha' < \alpha < 1$,
\begin{align*}
&\mspace{-10mu} \int \dd h \,  \dd k \, \dd z \, v(h) \, v(k) \, G(\epsilon h) \, G(\epsilon k) \,  \absbb{ \log \frac{|z +\epsilon h|}{|z|} } \, \absbb{ \log \frac{|z +\epsilon k|}{|z|} }
\\
&\lesssim
\int \dd h \,  \dd k \, v(h) \, v(k) \, G(\epsilon h) \, G(\epsilon k) \int \dd z \, \pbb{\pb{| \log |z+ \epsilon h| | + | \log |z| |} \ind{|z| \leq 2 \epsilon |h|} + \frac{\epsilon |h|}{|z|} \ind{|z| > 2 \epsilon |h|}}
\\
&\qquad \times \pbb{\pb{| \log |z+ \epsilon k| | + | \log |z| |} \ind{|z| \leq 2 \epsilon |k|} + \frac{\epsilon |k|}{|z|} \ind{|z| > 2 \epsilon |k|}}
\\
&\lesssim_\alpha \int \dd h \,  \dd k \, v(h) \, v(k) \, G(\epsilon h) \, G(\epsilon k)  \pbb{[\epsilon (\abs{h} \wedge \abs{k})]^{2 \alpha} + 2 \, (\epsilon \abs{h})^\alpha \, \epsilon \abs{k} + (\epsilon \abs{h})^\alpha (\epsilon \abs{k})^\alpha}
\\
&\lesssim \int \dd h \,  \dd k \, v(h) \, v(k) \, G(\epsilon h) \, G(\epsilon k)  \,(\epsilon \abs{h})^\alpha (\epsilon \abs{k})^\alpha
\\
&\lesssim_{\alpha'} \epsilon^{2 \alpha'}\,,
\end{align*}
where in the second step we used the estimates
\begin{align*}
\int \dd z \, \pB{| \log |z+ \epsilon h| |^2 + | \log |z| |^2} \ind{|z| \leq 2 \epsilon |h|} \ind{|z| \leq 2 \epsilon |k|} &\lesssim_\alpha [\epsilon (\abs{h} \wedge \abs{k})]^{2 \alpha}\,,
\\
\int \dd z \, \frac{| \log |z+ \epsilon h| | + | \log |z| |}{\abs{z}} \ind{|z| \leq 2 \epsilon |h|} &\lesssim_\alpha (\epsilon \abs{h})^\alpha
\end{align*}
and that $\frac{\epsilon |h|}{|z|} \leq \pb{\frac{\epsilon |h|}{|z|}}^{\alpha}$ for $|z| > 2 \epsilon |h|$. We conclude that  $\|  V^\epsilon - W^\epsilon \|_{L^2 (\mathbb{P})} \lesssim_\alpha \epsilon^{\alpha}$ for any $\alpha \in (0,1)$. 
\end{proof}

\begin{lemma} \label{lem:conv_V}
We have $\norm{V^{\epsilon} - V}_{L^2(\P)} \to 0$ as $\epsilon \to 0$.
\end{lemma}
\begin{proof}
Clearly,
\begin{equation*}
V^\epsilon - V = \frac{1}{2} \int \dd x \, \dd \tilde x \, \pb{v^\epsilon(x - \tilde x) - \delta(x - \tilde x)} \, \wick{ \abs{\phi(x)}^2 \, \abs{\phi(\tilde x)}^2}
\end{equation*}
Using Lemma \ref{lem:Wick} below (see also Example \ref{ex:two_blocks} below) we therefore find
\begin{multline*}
\E (V^\epsilon - V)^2 = \frac{1}{2} \int \dd x \, \dd \tilde x \, \dd y \, \dd \tilde y \, \pb{v^\epsilon (x - \tilde x)  - \delta(x - \tilde x)} \, \pb{v^\epsilon(y - \tilde y) - \delta(y - \tilde y)}
\\
\times G(x - y) \, G(\tilde x - \tilde y) \, \pb{G(x - y) \, G(\tilde x - \tilde y) + G(\tilde x - y) \, G(x - \tilde y)}\,.
\end{multline*}
The right-hand side splits into two terms. We only consider the first one; the analysis of the second one is analogous.
With the change of variables $\tilde x - x = h$, $\tilde y - y = k$, and $z = x - y$, the first term reads
\begin{multline*}
\frac{1}{2} \int \dd z  \, \dd h \, \dd k \, \pb{v^\epsilon (h)  - \delta(h)} \, \pb{v^\epsilon(k) - \delta(k)} \, G(z)^2 \, G(z + h - k)^2
\\
= \frac{1}{2} \int \dd z  \, G(z)^2 \int \dd h \, \dd k \, v^\epsilon (h) \, v^\epsilon (k) \pB{G(z + h - k)^2 - G(z - k)^2 - G(z + h)^2 + G(z)^2}\,.
\end{multline*}
We now estimate the first two terms and the last two terms separately. (A more careful second-order analysis could be done to take into account a further cancellation between all four terms, yielding a bound $\epsilon^{2 \alpha}$ for any $\alpha < 1$ instead of $\epsilon$, but we shall not need it.) The sum of the first two terms on the right-hand side is estimated by
\begin{align*}
&\mspace{-10mu}\frac{1}{2} \int \dd z  \, G(z)^2 \int \dd h \, \dd k \, v^\epsilon (h) \, v^\epsilon (k) \absB{G(z + h - k)^2 - G(z - k)^2}
\\
&\lesssim
\frac{1}{2} \int \dd h \, \dd k \, v^\epsilon (h) \, v^\epsilon (k) \int \dd z  \, G(z)^2 \, \pb{G(z + h - k) + G(z - k)} \pbb{\abs{h} + \absbb{\log \frac{\abs{z + h - k}}{\abs{z - k}}}}
\\
&\lesssim
\frac{1}{2} \int \dd h \, \dd k \, v^\epsilon (h) \, v^\epsilon (k) \int \dd z  \, G(z)^2 \, \pb{G(z + h - k) + G(z - k)}
\\
&\qquad \times \pbb{\abs{h} + \pb{\abs{\log \abs{z + h - k}} + \abs{\log \abs{z - k}}} \ind{\abs{z - k} \leq 2 \abs{h}} + \frac{\abs{h}}{\abs{z - k}} \ind{\abs{z - k} > 2 \abs{h}}}
\end{align*}
where in the second step we used the estimate \eqref{G-G_estimate}, and in the third the estimate \eqref{log_estimate}. Using that on the support of the integral over $h$ we have $\abs{h} \lesssim \epsilon$, we may perform the integral over $z$, followed by the integrals over $h$ and $k$, to deduce that the above expression is bounded by $O(\epsilon)$. 
Here we also used that $G$ has a logarithmic singularity at the origin, as established in Lemma \ref{lem:G_smoothness} below and the fact that  $\bigl\|\log |x| \,\ind{|x| \lesssim \epsilon}\bigr\|_{L^p} \lesssim_p \epsilon$ for $p \in (1,2)$, combined with H\"{o}lder's inequality.
This concludes the proof.
\end{proof}

\subsection{Integrability of $\ee^{-V^\epsilon}$} \label{sec:Nelson}
In this subsection we establish the integrability of $\ee^{-V^\epsilon}$, uniformly in $\epsilon$.
This is an adaptation of Nelson's argument \cite{nelson1973free} (see also \cite{Hairer2016} for a recent pedagocial account) to a nonlocal interaction.

\begin{proposition} \label{prop:Nelson}
There is a constant $c > 0$ such that for all $\epsilon > 0$ and $t \geq 1$ we have
\begin{equation*}
\P(\ee^{-V^{\epsilon}} > t) \lesssim \exp(-\ee^{c \sqrt{\log t}})\,.
\end{equation*}
The same estimate holds for $V^\epsilon$ replaced with $V$.
\end{proposition}
In particular, $\ee^{-V^{\epsilon}}$ is uniformly integrable in $\epsilon > 0$.

The rest of this subsection is devoted to the proof of Proposition \ref{prop:Nelson}. We start by noting that $\kappa - \Delta / 2$ has eigenfunctions $u_k \in \cal H$ and eigenvalues $\lambda_k$ indexed by $k \in \Z^d$ and given by
\begin{equation}
\label{lambda_k}
\lambda_k = \kappa + 2 \pi^2 \abs{k}^2 \,, \qquad u_k = \ee^{2 \pi \ii k \cdot x}\,.
\end{equation}
We shall use the truncated field $\phi_N$ from \eqref{def_phi_N} with a suitable truncation $\vartheta$, which is smooth in Fourier space. To that end, we fix $\rho$ to be a smooth, nonnegative, rotation invariant function, that has integral $1$ and is supported in the unit ball. We suppose that its Fourier transform
\begin{equation} \label{Fourier}
\vartheta(\xi) \deq \fra F \rho (\xi) \deq \int_{\R^2} \dd x \, \ee^{-2 \pi \ii \xi \cdot x} \, \rho(x)
\end{equation}
is nonnegative and radially nonincreasing (this can always be achieved by taking $\rho$ as a convolution of two nonnegative functions).

We define the truncated version of $V^\epsilon$ from \eqref{V^epsilon} through
\begin{equation} \label{def_V_e_M}
V^\epsilon_N \deq \frac{1}{2} \int \dd x \, \dd \tilde x \, v^\epsilon(x - \tilde x) \, \wick{ \abs{\phi_N(x)}^2 \, \abs{\phi_N(\tilde x)}^2}\,,
\end{equation}
which converges in $L^2(\P)$ to $V^\epsilon$ as $N \to \infty$.

Next, let $(Y_k)_{k \in \Z^d}$ be a family of i.i.d.\ standard complex Gaussian random variables, which is independent of the family $(X_k)_{k \in \Z^d}$. For  $0 < N \leq M \leq \infty$ we define the field
\begin{equation*}
\psi_{N,M} \deq \sum_{k \in \Z^d} \frac{Y_k}{\sqrt{\lambda_k}} \, \sqrt{\vartheta(k / M) - \vartheta(k / N)} \, u_k\,.
\end{equation*}
By construction, $\phi_N$ and $\psi_{N,M}$ are independent. For $M < \infty$, they are almost surely smooth on ${\Lambda}$. We define the truncated Green function
\begin{equation} \label{def_G_N}
G_N \deq G * \rho_N\,, \qquad \rho_N(x) \deq \sum_{n \in \Z^d} N^2 \rho(N(x + n))\,,
\end{equation}
and find by Poisson summation that, for $N \leq M$,
\begin{align} \label{G_N_phi}
\E \qb{\phi_N(x) \bar \phi_N(y)} &= \sum_{k \in \Z^d} \frac{1}{\lambda_k} \ee^{2 \pi \ii k \cdot (x - y)} \vartheta(k / N) = G_N(x-y)
\\ \label{G_N_M_phi}
\E \qb{\psi_{N,M}(x) \bar \psi_{N,M}(y)}  &= \sum_{k \in \Z^d} \frac{1}{\lambda_k} \ee^{2 \pi \ii k \cdot (x - y)} \pb{\vartheta(k / M) - \vartheta(k / N)} = G_M(x-y) - G_N(x - y)\,.
\end{align}
By independence of $\phi_N$ and $\psi_{N,M}$, we therefore find that for any $N \leq M$ we have the decomposition into low and high frequencies
\begin{equation} \label{phi_decomposition}
\phi_N + \psi_{N,M} \eqdist \phi_{M}\,,
\end{equation}
and in particular setting $M = \infty$ we get
\begin{equation*}
\phi_N + \psi_{N,\infty} \eqdist \phi\,.
\end{equation*}
Here $\eqdist$ denotes equality in law.

By \eqref{phi_decomposition} we have, for any $N \leq M$,
\begin{equation*}
V^{\epsilon}_M \eqdist \frac{1}{2} \int \dd x \, \dd \tilde x \, v^\epsilon(x - \tilde x) \, \wick{ \abs{\phi_N(x) + \psi_{N,M}(x)}^2 \, \abs{\phi_N(\tilde x) + \psi_{N,M}(\tilde x)}^2}\,.
\end{equation*}
For any $N \leq M$ we therefore have
\begin{equation*}
V^{\epsilon}_M \eqdist \sum_{a , \tilde a , b, \tilde b \in \{0,1\}} V^{\epsilon}_{ N,M}(a,\tilde a, b, \tilde b)\,,
\end{equation*}
where
\begin{multline} \label{def_V_epsilon}
V^{\epsilon}_{N,M}(a,\tilde a, b, \tilde b) \deq \frac{1}{2} \int \dd x \, \dd \tilde x \, v^\epsilon(x - \tilde x)
\\ \times  \wick{\phi_N(x)^{1 - a} \phi_N(\tilde x)^{1 - \tilde a} \bar \phi_N(x)^{1 - b} \bar \phi_N(\tilde x)^{1 - \tilde b} \,\psi_{N,M}(x)^{a} \psi_{N,M}(\tilde x)^{\tilde a} \bar \psi_{N,M}(x)^{b} \bar \psi_{N,M}(\tilde x)^{\tilde b}}\,.
\end{multline}

Hence, for $N \leq M$ we have
\begin{equation} \label{V_M-V_N}
V^{\epsilon}_M - V^{\epsilon}_N = \sum_{a , \tilde a , b, \tilde b \in \{0,1\}} \ind{a + \tilde a +  b + \tilde b > 0} \, V^{\epsilon}_{N,M}(a,\tilde a, b, \tilde b)\,.
\end{equation}

\begin{lemma} \label{lem:lbound_V}
There is a constant $C$ depending on $v$ such that almost surely
\begin{equation*}
V^{\epsilon}_{N} \geq - C (\log N)^2
\end{equation*}
for all $\epsilon > 0$.
\end{lemma}
\begin{proof}
Abbreviate $S = 1 + \int_{\R^2} \dd x \, \abs{v(x)}$.
Using the explicit form \eqref{wick_expanded} of the Wick power in \eqref{def_V_e_M} as well as \eqref{G_N_phi}, we find
\begin{align*}
V^{\epsilon}_N &= \frac{1}{2} \int \dd x \, \dd \tilde x \, v^\epsilon(x - \tilde x) \, \pB{\abs{\phi_N(x)}^2 \abs{\phi_N(\tilde x)}^2 - G_N(0) \abs{\phi_N(x)}^2 - G_N(0) \abs{\phi_N(\tilde x)}^2
\\
&\qquad - 2 \re G_N(x - \tilde x) \phi_N(x) \bar \phi_N(\tilde x) + G_N(0)^2 + G_N(x - \tilde x)^2}
\\
&= \frac{1}{2} \int \dd x \, \dd \tilde x \, v^\epsilon(x - \tilde x) \, \qB{\pb{\abs{\phi_N(x)}^2 - S G_N(0)} \pb{\abs{\phi_N(\tilde x)}^2 - S G_N(0)}
\\
&\qquad + (S - 1) G_N(0) \pb{\abs{\phi_N(x)}^2 + \abs{\phi_N(\tilde x)}^2 }
- 2 \re G_N(x - \tilde x) \phi_N(x) \bar \phi_N(\tilde x)
\\
&\qquad
- (S^2 - 1) G_N(0)^2 + G_N(x - \tilde x)^2
}
\\
&\geq (S - 1) G_N(0) \int \dd x \, \abs{\phi_N(x)}^2 - \re \int \dd x \, \dd \tilde x \, v^\epsilon(x - \tilde x)  G_N(x - \tilde x) \phi_N(x) \bar \phi_N(\tilde x) - \frac{S^2}{2} G_N(0)^2\,,
\end{align*}
where in the last step we used that $v$ (and hence also $v^\epsilon$) is of positive type with integral one.
Using $\abs{G_N(x)} \leq G_N(0)$ by \eqref{G_N_phi} and Cauchy-Schwarz combined with Young's inequality, we find
\begin{equation*}
\absbb{ \int \dd x \, \dd \tilde x \, v^\epsilon(x - \tilde x)  G_N(x - \tilde x) \phi_N(x) \bar \phi_N(\tilde x)} \leq (S - 1) \, G_N(0) \int \dd x \, \abs{\phi_N(x)}^2\,,
\end{equation*}
and the claim follows from Lemma \ref{lem:prop_G_N}.
\end{proof}

Next, we derive an estimate for the $L^2$-norm of $V^{\epsilon}_M - V^{\epsilon}_N$.

\begin{lemma} \label{lem:V_variance}
For any fixed $\delta > 0$ and for any $0 < N \leq M < \infty$ we have
\begin{equation*}
\normb{V^{\epsilon}_M - V^{\epsilon}_N}_{L^2(\P)} \lesssim N^{-1+\delta} 
\end{equation*}
\end{lemma}
\begin{proof}
By \eqref{V_M-V_N} and Minkowski's inequality, it suffices to estimate
\begin{equation} \label{variance_V}
\cal R \deq \E \qB{\abs{V^{\epsilon}_{N,M}(a,\tilde a, b, \tilde b)}^2} = \E \qB{V^{\epsilon}_{N,M}(a,\tilde a, b, \tilde b) \, \ol{V^{\epsilon}_{N,M}(a,\tilde a, b, \tilde b)}\,}
\end{equation}
for any fixed $a, \tilde a, b, \tilde b \in \{0,1\}$ satisfying $a + \tilde a + b + \tilde b > 0$.

Using Lemma \ref{lem:prop_G_N} below we find the bounds $G_N(x) \lesssim p(x)$ and $\abs{G_M(x) - G_N(x)} \lesssim q(x)$, where
\begin{equation*}
p(x) \deq 1 + \abs{\log \abs{x}}\,, \qquad q(x) \equiv q_N(x) \deq (1 + \abs{\log (N \abs{x})}) \wedge \frac{1}{N^2 \abs{x}^2}\,.
\end{equation*}
Note that $q(x) \lesssim p(x)$.
Using Wick's theorem, Lemma \ref{lem:Wick} below, and Young's inequality, we find
\begin{align*}
\cal R &\lesssim \int \dd x \, \dd \tilde x \, \dd y \, \dd \tilde y \, \abs{v^\epsilon(x - \tilde x)} \, \abs{v^\epsilon(y - \tilde y)} \, q(x - y) \, p(x - y) \, p(\tilde x - \tilde y)^2
\\
&\quad +
\int \dd x \, \dd \tilde x \, \dd y \, \dd \tilde y \, \abs{v^\epsilon(x - \tilde x)} \, \abs{v^\epsilon(y - \tilde y)} \, q(x - y) \, p(\tilde x - \tilde y) \, p(x - \tilde y) \, p(\tilde x - y)
\\ &\lesssim \sup_{y \in \Lambda} \int \, \dd x \, q(x) \, p(x-y)^3 + \int \dd x  \, \dd \tilde x \, q(x) \, p(\tilde x)^3
\\
&\lesssim \sup_{y \in \Lambda} \int \, \dd x \, q(x) \, p(x-y)^3
\\ &\lesssim \sup_{y \in \Lambda} \left[ \int_{|x| \leq C/N} \dd x \, \pb{ 1 + |\log |x| | }^4 + \frac{1}{N^2}  \int_{|x| > C/N} \dd x \, \pb{ 1 + |\log |x-y| | }^3 \frac{1}{|x|^2} \right] \lesssim_\delta  N^{-2+\delta}
\end{align*}
for any $\delta > 0$, where in the third step we used that $\int \dd \tilde x \, p(\tilde x)^3 \lesssim 1 \leq p(x)^3$, and in the last step we used H\"older's inequality.
\end{proof}

\begin{proof}[Proof of Proposition \ref{prop:Nelson}]
For any $N \geq 1$ we have, by Lemma \ref{lem:lbound_V},
\begin{equation*}
\P(\ee^{-V^{\epsilon}} > t) = \P\pb{V^{\epsilon} - V^{\epsilon}_N < - \log t - V^{\epsilon}_N}
\leq \P\pb{V^{\epsilon} - V^{\epsilon}_N < - (\log t - C (\log N)^2)}\,.
\end{equation*}
Now choose $N \geq 1$ such that
\begin{equation*}
\log t - C (\log N)^2 = 1\,,
\end{equation*}
which is always possible for $t$ large enough.

Next, we find that $V^{\epsilon}_M$ (or more precisely its real and imaginary parts) is in the $4$th polynomial chaos (see Section \ref{sec:hypercontr}), by using Lemma \ref{lem:chaos_decomp} and the easy fact that $V^{\epsilon}_M$ is orthogonal to the $n$th chaos for $n \neq 4$, which is a consequence of Wick's theorem in Lemma \ref{lem:Wick}. Hence, from Remark \ref{rem:hypercontr_complex} and Lemma \ref{lem:V_variance} we deduce that for any $0 < N \leq M < \infty$ and $p \in 2 \N$ we have
\begin{equation} \label{V_M-V_N2}
\normb{V^{\epsilon}_M - V^{\epsilon}_N}_{L^p(\P)} \lesssim \frac{p^2}{N^{2/3}}\,.
\end{equation}
Since $V_N^\epsilon \to V^\epsilon$ in $L^2(\P)$ as $N \to \infty$, by Lemma \ref{lem:hyper} we find that \eqref{V_M-V_N2} holds also for $M = \infty$ (i.e.\ replacing $V^\epsilon_M$ with $V^\epsilon$). Hence we get from Chebyshev's inequality, for any $p \in 2\N$,
\begin{equation*}
\P(\ee^{-V^{\epsilon}} > t) \leq \E \qb{\abs{V^{\epsilon} - V^{\epsilon}_N}^p} \leq \pbb{\frac{C p^2}{N^{2/3}}}^p \leq \pbb{\frac{p^2}{\sqrt{N}}}^p \leq \pb{p^2 \ee^{-c \sqrt{\log t}}}^p\,,
\end{equation*}
for large enough $t$ (and hence $N$). Choosing $p$ to be the largest element of $2 \N$ smaller than $\ee^{c/2 \sqrt{\log t} - 1/2}$ yields the claim for $V^\epsilon$.

Finally, the claim for $V$ easily follows from the one for $V^\epsilon$ and Lemma \ref{lem:conv_V}.
\end{proof}

\subsection{Convergence of the partition function} \label{sec:field_part_funct}
The first claim of Proposition \ref{prop:step2}, the convergence $\zeta^{W^\epsilon} \to \zeta^{V}$, follows immediately from the following result.

\begin{lemma} \label{lem:V_E_L2}
For any $1 \leq p < \infty$ we have $\normb{\ee^{-V} - \ee^{-W^\epsilon}}_{L^p(\P)} \to 0$ as $\epsilon \to 0$.
\end{lemma}

\begin{proof}
We begin by estimating $\norm{\ee^{-W^\epsilon}}_{L^p(\P)}$ by comparing it to $\norm{\ee^{-V^\epsilon}}_{L^p(\P)}$ and recalling Proposition \ref{prop:Nelson}. To that end, we note that $V^\epsilon - W^\epsilon$ is in the second polynomial chaos by \eqref{W_V_estimate} (more precisely, by Lemma \ref{lem:Wick}, $V^\epsilon - W^\epsilon$ is orthogonal to the $n$th polynomial chaos for $n \neq 2$ and the claim hence follows from Lemma \ref{lem:chaos_decomp}). Hence, by the hypercontractive bound from Lemma \ref{lem:hyper} we obtain
\begin{multline*}
\normb{\ee^{V^\epsilon - W^\epsilon} - 1}_{L^p(\P)} \leq \sum_{k \geq 1} \frac{1}{k!}\, \norm{(V^\epsilon - W^\epsilon)^k}_{L^p(\P)}
\\
= \sum_{k \geq 1} \frac{1}{k!}\, \norm{V^\epsilon - W^\epsilon}_{L^{pk}(\P)}^k
\lesssim
\sum_{k \geq 1} \frac{1}{k!} \, (pk)^k \norm{V^\epsilon - W^\epsilon}_{L^2(\P)}^k \leq
\sum_{k \geq 1} \pb{C p \norm{V^\epsilon - W^\epsilon}_{L^2(\P)}}^k 
\end{multline*}
for some constant $C > 0$, by Stirling's approximation for $k!$.
Using Lemma \ref{lem:R_variance} and Proposition \ref{prop:Nelson}, we conclude that for small enough $\epsilon$ (depending on $p$), $\norm{\ee^{-W^\epsilon}}_{L^p(\P)}$ is uniformly bounded in $\epsilon$. The claim now follows by writing
\begin{equation*}
\normb{\ee^{-V} - \ee^{-W^\epsilon}}_{L^p(\P)} \leq \int_0^1 \dd t \, \normB{(V - W^\epsilon) \, \ee^{-t W^\epsilon - (1 - t) V}}_{L^p(\P)}\,,
\end{equation*}
applying Hölder's inequality to the right-hand side, and combining Lemmas \ref{lem:R_variance} and \ref{lem:conv_V}  with Lemma \ref{lem:hyper} and the observation that $V - W^\epsilon$ lies in the span of the polynomial chaoses up to order four.
\end{proof}

\subsection{Convergence of correlation functions}
In this subsection we prove the second claim of Proposition \ref{prop:step2}, the convergence of the correlation functions in \eqref{convergence_gamma_field}. In order to obtain the uniform convergence of the Wick-ordered correlation functions $(\wh \gamma^{W^\epsilon}_p)_{\f x, \tilde {\f x}}$ to $(\wh \gamma^V_p)_{\f x, \tilde {\f x}}$, we use a representation obtained by repeated Gaussian integration by parts. To that end, we shall introduce a differential operator, denoted by $L_{N,x}$, such that
\begin{equation} \label{LN_prop1}
L_{N,x} \bar \phi(y) = G_N(x-y)
\end{equation}
and hence, formally,
\begin{equation} \label{LN_prop2}
L_{N,x} = \int \dd y \, G_N(x - y) \frac{\delta}{\delta \bar \phi(y)}\,.
\end{equation}
Our argument may be viewed as an instance of Malliavin calculus, with $L_{N,x}$ playing the role of the Malliavin derivative.

We choose the regularizing function $\vartheta$ to have compact support. Recall (see \eqref{def_phi_N}) that the underlying probability space consists of elements $X = (X_k)_{k \in \Z^d}$ with $X_k \in \C$. Define $\cal T$ to be the space of random variables of the form $f(X)$, where $f$ is smooth in the sense that all of its partial derivatives exist. We denote by $\partial_{X_k}$ and $\partial_{\bar X_k}$ the usual holomorphic and antiholomorphic partial derivatives in the complex variable $X_k$. On the space $\cal T$ we define the first order differential operators
\begin{equation*}
L_{N,x} \deq \sum_{k \in \Z^d} \frac{1}{\sqrt{\lambda_k}}\,  \sqrt{\vartheta(k / N)} \, u_k(x) \, \partial_{\bar X_k}\,, \qquad
\bar L_{N,x} \deq \sum_{k \in \Z^d} \frac{1}{\sqrt{\lambda_k}}\,  \sqrt{\vartheta(k / N)} \, \bar u_k(x) \, \partial_{X_k}\,.
\end{equation*}
where $x \in \Lambda$ and $0 < N < \infty$. Here we recall the definitions \eqref{lambda_k}. Note that, owing to our choice of $\vartheta$, each sum is finite. That this definition indeed satisfies \eqref{LN_prop1} is verified in Lemma \ref{lem:L_phi} below.

We record a few simple properties of $L_{N,x}$. The first property is Gaussian integration by parts for the operator $L_{N,x}$.

\begin{lemma} \label{lem:IBP}
Let $f(X) \in \cal T \cap L^1(\P)$. Then
\begin{equation*}
\E[L_{N,x} f(X)] = \E [\phi_N(x) f(X)] \,, \qquad \E[\bar L_{N,x} f(X)] = \E [\bar \phi_N(x) f(X)]\,.
\end{equation*}
\end{lemma}
\begin{proof}
We only prove the first identity. We use that if $Z$ is a standard complex Gaussian random variable, then $\E[Z f(Z)] = \E[ \partial_{\bar Z} f(Z)]$, as can be seen by integration by parts. Thus, using that each $X_k$ is a standard complex Gaussian random variable independent of the others, we get
\begin{align*}
\E[L_{N,x} f(X)] &= \sum_{k \in \Z^d} \frac{1}{\sqrt{\lambda_k}}\, \sqrt{\vartheta(k / N)} \, u_k(x) \, \E [\partial_{\bar X_k} f(X)]
\\
&= \sum_{k \in \Z^d} \frac{1}{\sqrt{\lambda_k}}\, \sqrt{\vartheta(k / N)} \, u_k(x) \, \E [X_k f(X)]
\\
&= \E[\phi_N(x) f(X)]\,. \qedhere
\end{align*}
\end{proof}

The second property is the verification of the condition \eqref{LN_prop1}.
\begin{lemma} \label{lem:L_phi}
For $x,y \in \Lambda$, we have
\begin{equation*}
L_{N,x} \, \phi(y) = 0\,, \qquad L_{N,x}\,  \bar \phi(y) = G_N(x - y)
\end{equation*}
in the sense of distributions in the variable $y$. Similar identities hold for $\bar L_{N,x}$.
\end{lemma}
\begin{proof}
We only prove the second identity. We compute
\begin{equation*}
L_{N,x}\,  \bar \phi(y) = \sum_{k \in \Z^d} \frac{1}{\sqrt{\lambda_k}}\,  \vartheta(k / N) \, u_k(x) \, \frac{1}{\sqrt{\lambda_k}} \, \bar u_k(y) = G_N(x-y)\,,
\end{equation*}
where the last step follows from \eqref{G_N_phi}\,.
\end{proof}

We may now prove the representation of \eqref{hat_gamma_p} underlying our proof. We denote the regularized Wick-ordered correlation function by
\begin{equation*}
(\wh{\gamma}_{N,p}^V)_{\f x,\tilde{\f x}} \deq \frac{1}{\zeta^V}\, \E \qB{\ee^{-V} \wick{\bar{\phi}_N(\tilde x_1) \cdots \bar{\phi}_N(\tilde x_p)\,\phi_N(x_1)\cdots \phi_N(x_p)}}\,.
\end{equation*}

\begin{lemma} \label{lem:IBP_representation}
We have
\begin{equation*}
(\wh{\gamma}_{N,p}^V)_{\f x,\tilde{\f x}} = \frac{1}{\zeta^V} \E \qB{\bar L_{N, \tilde x_1} \cdots \bar L_{N, \tilde x_p} L_{N, x_1} \cdots L_{N, x_p} \ee^{-V}}\,.
\end{equation*}
The same holds for $V$ replaced with $W^\epsilon$.
\end{lemma}
\begin{proof}
Using the recursive characterization of Wick ordering from \eqref{Wick_recursion}, we find
\begin{multline*}
\E \qBB{\ee^{-V} \wick{ \prod_{i \in [p]} \bar{\phi}_N(\tilde x_i) \prod_{i \in [p]} \phi_N(x_i)}}
= \E \qBB{\ee^{-V} \wick{ \prod_{i \in [p]} \bar{\phi}_N(\tilde x_i) \prod_{i \in [p-1]} \phi_N(x_i)} \, \phi_N(x_p)}
\\
- \sum_{j \in [p]} G_N(\tilde x_j - x_p) \, \E \qBB{\ee^{-V} \wick{ \prod_{i \in [p] \setminus \{j\}}\bar{\phi}_N(\tilde x_i) \prod_{i \in [p-1]} \phi_N(x_i)}}\,.
\end{multline*}
Here we used \eqref{G_N_phi} and that $\E[\phi_N(x) \phi_N(y)] = 0$.
Using Lemmas \ref{lem:IBP} and \ref{lem:L_phi} as well as the Leibniz rule for the operator $L_{N,x}$, we write the first term on the right-hand side as
\begin{multline*}
\E \qBB{L_{N,x_p} \pBB{\ee^{-V} \wick{ \prod_{i \in [p]} \bar{\phi}_N(\tilde x_i) \prod_{i \in [p-1]} \phi_N(x_i)}}}
= \E \qBB{L_{N,x_p} (\ee^{-V}) \, \wick{ \prod_{i \in [p]} \bar{\phi}_N(\tilde x_i) \prod_{i \in [p-1]} \phi_N(x_i)}}
\\
+ \sum_{j \in [p]} G_N(\tilde x_j - x_p) \, \E \qBB{\ee^{-V} \wick{ \prod_{i \in [p] \setminus \{j\}}\bar{\phi}_N(\tilde x_i) \prod_{i \in [p-1]} \phi_N(x_i)}}\,.
\end{multline*}
We conclude that
\begin{equation*}
\E \qBB{\ee^{-V} \wick{ \prod_{i \in [p]} \bar{\phi}_N(\tilde x_i) \prod_{i \in [p]} \phi_N(x_i)}} = 
\E \qBB{L_{N,x_p} (\ee^{-V}) \, \wick{ \prod_{i \in [p]} \bar{\phi}_N(\tilde x_i) \prod_{i \in [p-1]} \phi_N(x_i)}}\,.
\end{equation*}
Repeating this argument $2p$ times yields the claim.
\end{proof}

The following result is the main analytical tool behind the proof of \eqref{convergence_gamma_field}.

\begin{lemma} \label{lem:derivatives-convergence}
Let $\ell \geq 0$ and $\f z = (z_1, \dots, z_\ell) \in \Lambda^\ell$. Abbreviate
\begin{equation*}
\cal L_{N, \f z} \deq L_{N, z_1}^\# \cdots L_{N, z_\ell}^\#\,,
\end{equation*}
where each $L^\#$ stands for either $L$ or $\bar L$. Then the following holds for any $r \geq 1$.
\begin{enumerate}[label=(\roman*)]
\item \label{lab:der1}
$\sup_N \sup_{\f z} \norm{\cal L_{N, \f z}V}_{L^r(\P)} < \infty$.
\item \label{lab:der2}
As $\epsilon \to 0$ we have $\sup_N \sup_{\f z} \norm{\cal L_{N, \f z} W^\epsilon - \cal L_{N, \f z} V}_{L^r(\P)} \to 0$.
\item \label{lab:der3}
As $M,N \to \infty$ we have $\sup_{\f z} \norm{(\cal L_{N, \f z} - \cal L_{M, \f z}) V}_{L^r(\P)} \to 0$.
\item \label{lab:der4}
For any $\epsilon > 0$, as $M,N \to \infty$ we have $\sup_{\f z} \norm{(\cal L_{N, \f z} - \cal L_{M, \f z}) W^\epsilon}_{L^r(\P)} \to 0$.
\end{enumerate}
\end{lemma}

Before proving Lemma \ref{lem:derivatives-convergence}, we use it to conclude the proof of \eqref{convergence_gamma_field}, and hence also of Proposition \ref{prop:step2}.

\begin{proof}[Proof of \eqref{convergence_gamma_field}]
We begin by showing that $(\wh \gamma^V_{N,p})_{N > 0}$ is a Cauchy sequence in $\norm{\cdot}_C$, and it hence converges to a limit that is by definition $\wh \gamma^V_p$. In the notation of Lemma \ref{lem:derivatives-convergence}, we deduce from Lemma \ref{lem:IBP_representation} that
\begin{equation} \label{cauchy_est}
\zeta^V \normb{\wh{\gamma}_{N,p}^V - \wh{\gamma}_{M,p}^V}_{C} = \sup_{\f z} \absB{\E\qB{(\cal L_{N, \f z} - \cal L_{M, \f z}) \ee^{-V}}}\,,
\end{equation}
where we abbreviate $\f z = (\f x, \tilde {\f x}) \in \Lambda^{2p}$ and choose the supercripts $\#$ in the definition of $\cal L_{N, \f z}$ appropriately. Applying the chain rule and the Leibniz rule to the $2 p$ derivatives in $\cal L_{N, \f z}$ and $\cal L_{M,\f z}$, we estimate the right-hand side by a finite number of terms of the form
\begin{equation} \label{cauchy_prod_est}
\sup_{\f z_1, \dots, \f z_k} \E\absBB{\pBB{\prod_{i = 1}^k (\cal L_{N, \f z_i} V) - \prod_{i = 1}^k (\cal L_{M, \f z_i} V)} \ee^{-V}}\,,
\end{equation}
where $k \leq 2p$ and, for all $i = 1, \dots, k$, $\f z_i \in \Lambda^{\ell_i}$ for some $\ell_i \leq 4$. Using the identity $\prod_i a_i - \prod_i b_i = \sum_{i} (a_i - b_i) \prod_{j < i} a_j \prod_{j > i} b_j$, we estimate \eqref{cauchy_prod_est} by
\begin{equation} \label{cauchy_prod_est2}
\sum_{i = 1}^k \sup_{\f z_1, \dots, \f z_k} \E \absBB{\pb{\cal L_{N, \f z_i} V - \cal L_{M, \f z_i} V} \prod_{j < i} (\cal L_{N, \f z_j} V) \prod_{j > i} (\cal L_{M, \f z_j} V) \, \ee^{-V}}\,.
\end{equation}
Applying Hölder's inequality, we conclude using Lemma \ref{lem:derivatives-convergence} \ref{lab:der1} and \ref{lab:der3}, as well as Proposition \ref{prop:Nelson}, that \eqref{cauchy_est} converges to $0$ as $M,N \to \infty$. We conclude that $(\wh \gamma^V_{N,p})_{N > 0}$ is a Cauchy sequence in $\norm{\cdot}_C$.

The same argument with $V$ replaced with $W^\epsilon$, using Lemma \ref{lem:derivatives-convergence} \ref{lab:der4}, shows that, for any $\epsilon > 0$, $\wh \gamma^{W^\epsilon}_{N,p}$ converges in $\norm{\cdot}_C$ to $\wh \gamma^{W^\epsilon}_p$. Here, we also use the observation that
$\sup_N \sup_{\f z} \norm{\cal L_{N, \f z}W^\epsilon}_{L^r(\P)}<\infty$, which follows from Lemma \ref{lem:derivatives-convergence}  \ref{lab:der1} and  \ref{lab:der2}.

Now writing
\begin{equation*}
\normb{\zeta^{W^\epsilon} \wh \gamma^{W^\epsilon}_{p} -\zeta^V \wh \gamma^V_p }_C \leq
\normb{\zeta^{W^\epsilon} \wh \gamma^{W^\epsilon}_{p} - \zeta^{W^\epsilon} \wh \gamma^{W^\epsilon}_{N,p} }_C
+
\normb{\zeta^{W^\epsilon} \wh \gamma^{W^\epsilon}_{N,p} -\zeta^V \wh \gamma^V_{N,p}}_C
+
\normb{\zeta^V \wh \gamma^V_{N,p} - \zeta^V \wh \gamma^V_p}_C\,,
\end{equation*}
we therefore conclude, using $\lim_{\epsilon \to 0} \zeta^{W^\epsilon} = \zeta^V$ by Lemma \ref{lem:V_E_L2}, that the convergence \eqref{convergence_gamma_field} holds provided that
\begin{equation} \label{eps_conv_gamma}
\lim_{\epsilon \to 0} \sup_N \normb{\zeta^{W^\epsilon} \wh \gamma^{W^\epsilon}_{N,p} -\zeta^V \wh \gamma^V_{N,p}}_C = 0\,.
\end{equation}

To prove \eqref{eps_conv_gamma}, we write, using Lemma \ref{lem:IBP_representation},
\begin{equation*}
 \normb{\zeta^{W^\epsilon} \wh \gamma^{W^\epsilon}_{N,p} -\zeta^V \wh \gamma^V_{N,p}}_C = \sup_{\f z} \absB{\E \qB{\cal L_{N, \f z} \pb{\ee^{-W^\epsilon} - \ee^{V}}}}\,.
\end{equation*}
Similarly to \eqref{cauchy_est}--\eqref{cauchy_prod_est} above, we estimate the right-hand side by a sum of terms of the form
\begin{equation*}
\sup_{\f z_1, \dots, \f z_k} \E\absBB{\prod_{i = 1}^k (\cal L_{N, \f z_i} W^\epsilon) \, \ee^{-W^\epsilon} - \prod_{i = 1}^k (\cal L_{N, \f z_i} V) \, \ee^{-V}}\,,
\end{equation*}
which converges to $0$ as $\epsilon \to 0$, uniformly in $N$, by telescoping (analogously to \eqref{cauchy_prod_est2}), Hölder's inequality, as well as Lemma \ref{lem:derivatives-convergence} \ref{lab:der1}--\ref{lab:der2} and Proposition \ref{lem:V_E_L2}. This concludes the proof.
\end{proof}

\begin{proof}[Proof of Lemma \ref{lem:derivatives-convergence}]
Since all quantities appearing inside the $L^r(\P)$-norms in Lemma \ref{lem:derivatives-convergence} are a superpositions of random variables in polynomial chaoses of order at most four (see Section \ref{sec:hypercontr}), from the hypercontractivity estimate of Remark \ref{rem:hypercontr_complex} we find that it suffices to consider $r = 2$.

The proof of \ref{lab:der1} is similar to that of \ref{lab:der2}, and we omit it. To prove \ref{lab:der2}, we proceed by telescoping via $V^\epsilon$, writing $W^\epsilon - V = (W^\epsilon - V^\epsilon) + (V^\epsilon - V)$. Thus, we have to differentiate the random variables
\begin{align*}
W^\epsilon - V^\epsilon &=  \int \dd x \, \dd \tilde x \, v^\epsilon(x - \tilde x) \, G(x -  \tilde x) \, \pb{\wick{\bar \phi(x) \phi(\tilde x)} - \wick{\abs{\phi(x)}^2}}\,,
\\
V^\epsilon - V &= \frac{1}{2} \int \dd x \, \dd \tilde x \, \pb{v^\epsilon(x - \tilde x) - \delta(x - \tilde x)} \,\wick{ \abs{\phi(x)}^2 \, \abs{\phi(\tilde x)}^2}
\end{align*}
by $\cal L_{N, \f v}$.

The zeroth order derivatives, $\ell = 0$, were estimated in Lemmas \ref{lem:R_variance} and \ref{lem:conv_V}. For the high-order derivatives, let us start with $W^\epsilon - V^\epsilon$. For the first order derivative, from \eqref{Wick_derivative} and Lemma \ref{lem:L_phi}, we find
\begin{equation*}
L_{N,z} (W^\epsilon - V^\epsilon) = 
\int \dd x \, \dd \tilde x \, v^\epsilon(x - \tilde x) \, G(x -  \tilde x) \, G_N(z - x) \, \p{\phi(\tilde x) - \phi(x)}\,,
\end{equation*}
so that Lemma \ref{lem:Wick} (see also Example \ref{ex:two_blocks}) yields
\begin{multline*}
\normb{L_{N,z} (W^\epsilon - V^\epsilon)}_{L^2(\P)}^2 = \int \dd x \, \dd \tilde x \, \dd y \, \dd \tilde y \,
v^\epsilon(x - \tilde x) \, G(x -  \tilde x) \, G_N(z - x) \, v^\epsilon(y - \tilde y) \, G(y -  \tilde y) \, G_N(z - y)
\\
\times \pb{G(x-y) + G(\tilde x - \tilde y) - G(x - \tilde y) - G(\tilde x - y)}\,.
\end{multline*}
The right-hand side is estimated similarly to the proof of Lemma \ref{lem:R_variance}, uniformly in $N$ and $z \in \Lambda$, using that $G_N$ is uniformly bounded in $L^2(\Lambda)$, and the bound
\begin{equation*}
\absb{G(x-y) + G(\tilde x - \tilde y) - G(x - \tilde y) - G(\tilde x - y)} \leq \abs{G(x - y) - G(x - \tilde y)} + \abs{G(\tilde x - \tilde y) - G(\tilde x - y)}\,.
\end{equation*}

For the second order derivative, we obtain
\begin{equation*}
\bar L_{N, \tilde z} L_{N,z} (W^\epsilon - V^\epsilon) =
\int \dd x \, \dd \tilde x \, v^\epsilon(x - \tilde x) \, G(x -  \tilde x) \, G_N(z - x) \, \pb{G_N(\tilde z - \tilde x) - G_N(\tilde z - x)}\,,
\end{equation*}
which can again be estimated as in the proof of Lemma \ref{lem:R_variance}, uniformly in $N$ and $z, \tilde z \in \Lambda$. This concludes the estimate of $\norm{\cal L_{N, \f z} W^\epsilon - \cal L_{N, \f z} V^\epsilon}_{L^2(\P)}$.

To estimate $\norm{\cal L_{N, \f z} V^\epsilon - \cal L_{N, \f z} V}_{L^2(\P)}$, we compute using \eqref{Wick_derivative} and Lemma \ref{lem:L_phi}
\begin{equation*}
L_{N,z} (V^\epsilon - V) = \int \dd x \, \dd \tilde x \, \pb{v^\epsilon(x - \tilde x) - \delta(x - \tilde x)} \, G_N(z - x) \,\wick{\phi(x) \, \abs{\phi(\tilde x)}^2}\,, 
\end{equation*}
so that Lemma \ref{lem:Wick} (see also Example \ref{ex:two_blocks}) yields
\begin{multline*}
\normb{L_{N,z} (V^\epsilon - V)}_{L^2(\P)}^2 = \int \dd x \, \dd \tilde x \, \dd y \, \dd \tilde y \, \pb{v^\epsilon(x - \tilde x) - \delta(x - \tilde x)} \, \pb{v^\epsilon(y - \tilde y) - \delta(y - \tilde y)} \, G_N(z - x) \, G_N(z - y)
\\
\times \pb{G(x-y) G(\tilde x - \tilde y)^2 + G(x - \tilde y) G(\tilde x - y) G(\tilde x - \tilde y)}\,.
\end{multline*}
The right-hand side is estimated  as in the proof of Lemma \ref{lem:conv_V}, uniformly in $N$ and $z \in \Lambda$. The higher order derivatives are estimated analogously. This concludes the proof of \ref{lab:der2}.

The proofs of \ref{lab:der3} and \ref{lab:der4} are similar, and we focus on \ref{lab:der3}. For the first order derivative, we find
\begin{equation*}
(L_{N,z} - L_{M,z}) V = \int \dd x \, (G_N(x - z) - G_M(x - z)) \wick{\phi(x) \abs{\phi(x)}^2}\,,
\end{equation*}
so that Lemma \ref{lem:Wick} (see also Example \ref{ex:two_blocks}) yields
\begin{equation*}
\norm{(L_{N,z} - L_{M,z}) V}^2_{L^2(\P)} = 2 \int \dd x \, \dd y \, (G_N(x - z) - G_M(x - z)) \, (G_N(y - z) - G_M(y - z)) \, G(x - y)^3\,.
\end{equation*}
Telescoping $G_N - G_M = (G_N - G) - (G_M - G)$ and using Lemma \ref{lem:G_smoothness} and Hölder's inequality, we conclude that
\begin{equation*}
\norm{(L_{N,z} - L_{M,z}) V}_{L^2(\P)} \lesssim \pbb{\int \dd x \, \abs{G_N(x) - G(x)}^3}^{1/3} + \pbb{\int \dd x \, \abs{G_M(x) - G(x)}^3}^{1/3}\,.
\end{equation*}
Using \eqref{G_N_est2} and splitting the integration domain into $\abs{x} \leq 1/N$ and $\abs{x} > 1/N$, we easily deduce that the right-hand side vanishes as $N,M \to \infty$.
The higher order derivatives are estimated in exactly the same way. This concludes the proof.
\end{proof}

\section{Proof of Proposition \ref{prop:step1}}
\label{The rate of convergence}

We study the rate of convergence of the relative partition function and the correlation functions in the mean-field limit, while keeping track of the parameter $\epsilon$. This amounts to a quantitative analysis of the infinite-dimensional saddle point argument for the functional integral introduced in \cite{FKSS_2020}. Note that, without the $\tau^\epsilon$ and $E^{\epsilon}$ correction terms in \eqref{Hamiltonian_H}, \eqref{W^epsilon}, and with $\epsilon=1$, this convergence was obtained in a qualitative way in \cite[Section 5]{FKSS_2020}. The main ingredients that enable a quantitative analysis are: (a) the Lipschitz continuity of the interaction potential $v^{\epsilon}$, with Lipschitz constant depending on $\epsilon$ (see Lemma \ref{v^{epsilon}_properties} \ref{itm:heatkernel2} below), and (b) quantitative $L^p$-H\"{o}lder continuity properties of Brownian motion (see Lemma \ref{Heat_kernel_estimates} \ref{itm:heatkernel2} below).
As a result, we can find a suitable choice of $\epsilon$ as a function of $\nu$ such that we get the wanted convergence as $\nu \rightarrow 0$; see \eqref{epsilon_lower_bound} below. Our methods work for $d = 2,3$, and hence all results of this section are stated for both dimensions.

We now state the explicit lower bound on $\epsilon \equiv \epsilon(\nu)$. Namely, throughout the sequel we assume that $\epsilon$ satisfies 
\begin{equation}
\label{epsilon_lower_bound}
\epsilon(\nu) \gtrsim 
\begin{cases}
\exp\Bigl(-(\log \nu^{-1})^{\frac{1-a}{2}}\Bigr) & \text{if } d=2\\
(\log \nu^{-1})^{-\frac{1-a}{2}}& \text{if } d=3\,,
\end{cases}
\end{equation}
for some $a \in (0,1)$.
Let us define $\chi\col \R^+ \rightarrow \R$ by
\begin{equation}
\label{sigma(epsilon)}
\chi(t) \deq 
\begin{cases}
\log t^{-1}  & \text{if } d=2
\\
t^{-1} & \text{if } d=3\,.
\end{cases}
\end{equation}
By \eqref{v_epsilon}, \eqref{tau_epsilon}, \eqref{sigma(epsilon)}, Lemma \ref{lem:G_smoothness} (when $d=2$), and Remark \ref{G_smoothness_3D_remark} (when $d=3$), we note that\footnote{Throughout the sequel, we do not emphasize the dependence of the implied constants on $d=2,3$.}
\begin{equation}
\label{tau(epsilon)_bound}
|\tau^\epsilon| \lesssim_{\kappa,v} \chi(\epsilon)\,, \qquad |E^{\epsilon}| \lesssim_{\kappa} \chi(\epsilon)^2\,.
\end{equation}
Furthermore, by \eqref{sigma(epsilon)} and \eqref{epsilon_lower_bound}, it follows that for all $C,b>0$, we have
\begin{equation}
\label{epsilon_lower_bound_2}
\lim_{\nu \rightarrow 0} \ee^{C \chi(\epsilon)^2}\,\nu^{b}=0\,.
\end{equation}

We now state the main results which, in light of \eqref{epsilon_lower_bound_2}, imply Proposition \ref{prop:step1}.

\begin{proposition}
\label{Partition_function_rate_of_convergence}
There exists $C_1>0$ depending on $\kappa,v$ such that
\begin{equation*}
\Bigl|\cal Z-\zeta^{W^\epsilon}\Bigr| \lesssim_{\kappa,v} 
\begin{cases}
\ee^{C_1 \chi(\epsilon)^2}\, \nu^{1/4}  & \text{if } d=2
\\
\ee^{C_1 \chi(\epsilon)^2}\, \nu^{1/4}\,\log \nu^{-1} & \text{if } d=3\,.
\end{cases}
\end{equation*}
\end{proposition}
In statements of results, we use the notation $b-$ to mean that the statement holds for $b - c$ for any constant $c > 0$. 
\begin{proposition}
\label{Correlation_functions_rate_of_convergence}
For $p \in \N$, we define
\begin{equation}
\label{theta_definition}
\theta(d,p) \deq
\begin{cases}
\frac{1}{4p+4}  & \text{if } d=2
\\
\frac{1}{12p+8}-& \text{if } d=3\,.
\end{cases}
\end{equation}

There exists $C_2>0$ depending on $\kappa,v$ such that
\begin{equation*}
\Bigl\|\nu^p \, \wh{\Gamma}_{p}-\wh{\gamma}_{p}^{W^\epsilon}\Bigr\|_{C}
\lesssim_{p,\kappa,v} \ee^{C_2 \chi(\epsilon)^2}\,\nu^{\theta(d,p)}\,.
\end{equation*}
\end{proposition}

\begin{remark}
\label{v_assumptions}
We note that, in order to obtain Propositions \ref{Partition_function_rate_of_convergence} and \ref{Correlation_functions_rate_of_convergence} above, and hence Proposition \ref{prop:step1}, we only need to use the bounds from Lemma \ref{v^{epsilon}_properties} below. 
In light of this observation, we can consider more general $v$, which are not smooth. We always assume that $v :\R^d \rightarrow \R$ is even, $L^1$ with integral 1, and of positive type. Since $v$ is not smooth, we need to consider suitable regularizations $v_{\eta}$ of $v$. For a motivation and detailed description of the regularization, we refer the reader to \cite[Section 3.1]{FKSS_2020} and \cite[Section 4.1]{FKSS_2020}. A summary is also given in the study of the nonlocal problem in Section \ref{Extension of results to unbounded nonlocal interactions} below.

In the first generalization, we consider $v$ Lipschitz and compactly supported. Then the result of Lemma \ref{v^{epsilon}_properties} holds for $v^{\epsilon}_{\eta}$ (given by \eqref{v_epsilon} with $v$ replaced by $v_{\eta}$), uniformly in $\eta$. This follows immediately from the proof of Lemma \ref{v^{epsilon}_properties} below.

In the second generalization, we assume that $v$ is differentiable and that uniformly in $\epsilon \in (0,1)$, we have 
\begin{equation}
\label{v_epsilon_assumptions_2}
\sup_{x \in \Lambda}\,\sum_{n \in \Z^d}  \Biggl|v \pbb{\frac{x - n}{\epsilon}}\Biggr|+\sup_{x \in \Lambda}\,\sum_{n \in \Z^d}  \Biggl|\nabla v \pbb{\frac{x - n}{\epsilon}}\Biggr| \lesssim 1\,.
\end{equation}
Let us note that \eqref{v_epsilon_assumptions_2} holds if we assume that there exists $a>d$ such that for all $x \in \R^d$
\begin{equation*}
|v(x)|+|\nabla v(x)| \lesssim \frac{1}{(1+|x|)^{a}}\,.
\end{equation*}
In particular, under the latter conditions, we do not need to assume that $v$ is compactly supported. 
\end{remark}

\subsection{The partition function}
In this subsection, we prove Proposition \ref{Partition_function_rate_of_convergence}. Before proceeding with the proof, we make several observations and review the functional integral representation from \cite{FKSS_2020}. See also \cite{ginibre1971some}.
We first recall some basic notions for Brownian paths. Given $0 \leq \tilde{\tau} < \tau$, we denote by $\Omega^{\tau,\tilde{\tau}}$ the space of continuous paths $\omega\col [\tilde{\tau},\tau] \rightarrow \Lambda$. Given $\tilde{x} \in \Lambda$ and $0 \leq \tilde{\tau} <\tau$, $\mathbb{P}^{\tau,\tilde{\tau}}_{\tilde{x}}(\dd \omega)$ denotes the law on $\Omega^{\tau,\tilde{\tau}}$ of standard Brownian motion with periodic boundary conditions on $\Lambda$ that equals $\tilde{x}$ at time $\tilde{\tau}$. Given $x,\tilde{x} \in \Lambda$ and $0 \leq \tilde{\tau} <\tau$, $\mathbb{P}^{\tau,\tilde{\tau}}_{x,\tilde{x}}(\dd \omega)$ denotes the law of the Brownian bridge $\Omega^{\tau,\tilde{\tau}}$ with periodic boundary conditions on $\Lambda$ that equals $\tilde{x}$ at time $\tilde{\tau}$ and $x$ at time $\tau$. For $t>0$, we write the heat kernel on $\Lambda$ as
\begin{equation*}
\psi^t(x)\deq \ee^{t \Delta/2}(x)=\sum_{n \in \Z^d} \frac{1}{(2\pi t)^{d/2}} \,\ee^{-\frac{|x-n|^2}{2t}}\,.
\end{equation*}
For $x,\tilde{x} \in \Lambda$ and $0 \leq \tilde{\tau} <\tau$, we define the positive measure 
\begin{equation}
\label{W_measure}
\mathbb{W}^{\tau,\tilde{\tau}}_{x,\tilde{x}}(\dd \omega) \deq \psi^{\tau-\tilde{\tau}}(x-\tilde{x})\,\mathbb{P}^{\tau,\tilde{\tau}}_{x,\tilde{x}}(\dd \omega)\,.
\end{equation}
Given $n \in \N^*$, $\tilde{\tau}<t_{1}<\cdots<t_{n}<\tau$, and $f\col \Lambda^n \rightarrow \mathbb{R}$ continuous, the measure \eqref{W_measure} satisfies
\begin{align*}
&\int \mathbb{W}^{\tau,\tilde{\tau}}_{x,\tilde{x}}(\dd \omega)\,f(\omega(t_{1}),\ldots,\omega(t_{n}))=
\\
&\int \dd x_{1}\,\cdots\,\dd x_{n}\,\psi^{t_{1}-\tilde{\tau}}(x_{1}-\tilde{x})\,\psi^{t_{2}-t_{1}}(x_{2}-x_{1})\,\cdots\,\psi^{t_{n}-t_{n-1}}(x_{n}-x_{n-1})\,\psi^{\tau-t_{n}}(x-x_{n})\,f(x_{1},\ldots,x_{n})\,.
\end{align*}

We note several useful estimates for the above quantities.

\begin{lemma}
\label{Heat_kernel_estimates}
The following estimates hold.
\begin{enumerate}[label=(\roman*)]
\item \label{itm:heatkernel1}
There exists a constant $C>0$ such that for all $0 \leq \tilde{\tau} <\tau$, we have
\begin{equation*}
\sup_{x,\tilde{x}} \int \mathbb{W}^{\tau,\tilde{\tau}}_{x,\tilde{x}}(\dd \omega)=\sup_{x,\tilde{x}} \psi^{\tau,\tilde{\tau}} (x-\tilde{x}) \leq C\,\biggl(1+\frac{1}{(\tau-\tilde{\tau})^{d/2}}\biggr)\,.
\end{equation*}
\item \label{itm:heatkernel2}
There exists a constant $C>0$ such that for all $\tilde{\tau} \leq s \leq t \leq \tau$, we have
\begin{equation}
\label{Heat_kernel_estimates_ii_1}
\int \mathbb{W}^{\tau,\tilde{\tau}}_{x,\tilde{x}}(\dd \omega)\,|\omega(t)-\omega(s)|^2_{\Lambda} \leq C \,\biggl(1+\frac{1}{(\tau-\tilde{\tau})^{d/2}}\biggr)\,\biggl((t-s)+
|x-\tilde{x}|_{\Lambda}^2\,\frac{(t-s)^2}{(\tau-\tilde{\tau})^2}\biggr)\,.
\end{equation}
and
\begin{equation}
\label{Heat_kernel_estimates_ii_2}
\int \mathbb{W}^{\tau,\tilde{\tau}}_{x,\tilde{x}}(\dd \omega)\,|\omega(t)-\omega(s)|_{\Lambda} \leq C \,\biggl(1+\frac{1}{(\tau-\tilde{\tau})^{d/2}}\biggr)\,\biggl((t-s)^{1/2}+
|x-\tilde{x}|_{\Lambda}\,\frac{t-s}{\tau-\tilde{\tau}}\biggr)\,.
\end{equation}
Here $|x|_{\Lambda} \deq \min_{n\in \mathbb{Z}^d} |x - n|$ denotes the periodic Euclidean norm of $x \in \Lambda$.
\item \label{itm:heatkernel3}
For $0<s\leq t$, we have
\begin{equation*}
\|\psi^t-\psi^s\|_{L^1(\Lambda)} \leq d \log(t/s)\,.
\end{equation*}
\end{enumerate}
\end{lemma}
The results of Lemma \ref{Heat_kernel_estimates} are contained in \cite{FKSS_2020}. Part \ref{itm:heatkernel1} is given in \cite[Lemma 2.2]{FKSS_2020}. Estimate \eqref{Heat_kernel_estimates_ii_1} is proved in \cite[Lemma 2.3]{FKSS_2020}. Estimate \eqref{Heat_kernel_estimates_ii_2} then follows from the Cauchy-Schwarz inequality and part \ref{itm:heatkernel1}. Part \ref{itm:heatkernel3} follows since
\begin{equation*}
\|\psi^{t}-\psi^{s}\|_{L^1} \leq \int_{s}^{t} \dd u\, \|\partial_{u} \psi^{u}\|_{L^1} \leq \int_{s}^{t} \dd u\, \int_{\R^d} \dd x\, \biggl(\frac{d}{2u}+\frac{|x|^2}{2u^2}\biggr)\,\widetilde{\psi}^{u}(x)=d \log(t/s)\,,
\end{equation*}
where $\widetilde{\psi}^{u}(x)=\frac{1}{(2\pi u)^{d/2}}\,\ee^{-\frac{|x|^2}{2u}}$ is the heat kernel on $\R^d$, 
as was noted in the proof of \cite[Lemma 5.18]{FKSS_2020}.

Let us fix a function $\varphi \in C_c^{\infty}(\R)$ which is even, nonnegative, of positive type, and which satisfies $\varphi(0)=1$.
For fixed $L>0$, given $\eta>0$, we define the $\nu$-periodic function
\begin{equation}
\label{approximate_delta_function}
\delta_{\eta}(\tau) \deq \frac{1}{\eta} \sum_{y \in \Z} (\fra F^{-1} \varphi) \biggl(\frac{\tau- \nu y}{\eta}\biggr)\,,
\end{equation}
where $\fra F$ denotes Fourier transform (see \eqref{Fourier}).
Here, \eqref{approximate_delta_function} can be interpreted as an approximate delta function on $[-\nu/2,\nu/2)$.
By construction, we have
\begin{equation}
\label{delta_nu_integral}
\int_{0}^{\nu} \dd \tau\,\delta_{\eta}(\tau)=1\,,\qquad \delta_{\eta} \geq 0.
\end{equation}

For simplicity of notation, we suppress the dependence on $\epsilon$ and $\nu$ in the quantum objects. We only emphasize the $\eta$ dependence through a subscript when appropriate. We write the $\epsilon$ dependence as a superscript in the classical objects.
Let us note several properties of $v^{\epsilon}$ that follow from \eqref{v_epsilon}.
\begin{lemma}
\label{v^{epsilon}_properties}
There exists $C>0$, depending only on $v$, such that the following properties hold.
\begin{enumerate}[label=(\roman*)]
\item \label{itm:veps1}
$\|v^{\epsilon}\|_{L^{\infty}(\Lambda)} \leq \frac{C}{\epsilon^d}$.
\item  \label{itm:veps2}
We have that 
\begin{equation*}
|v^{\epsilon}(x)-v^{\epsilon}(y)|\leq \frac{C}{\epsilon^{d+1}}\,|x-y|_{\Lambda}
\end{equation*}
for all $x,y \in \Lambda$.
\end{enumerate}
\end{lemma}

For $\eta>0$, recalling \eqref{approximate_delta_function}, we let $(\cal{C}_{\eta})^{\tau,\tilde{\tau}}_{x,\tilde{x}} \deq \nu \,\delta_{\eta}(\tau-\tilde{\tau})\,v^{\epsilon}(x-\tilde{x})$ and define $\mu_{\cal C_{\eta}}(\dd \sigma)$ to be the real Gaussian measure with mean zero and covariance
\begin{equation}
\label{Gaussian_measure_quantum}
\int \mu_{\cal C_{\eta}} (\dd \sigma) \, \sigma(\tau,x)\,\sigma(\tilde{\tau},\tilde{x})=(\cal{C}_{\eta})^{\tau,\tilde{\tau}}_{x,\tilde{x}}\,.
\end{equation}
Since $v \in C_c^{\infty}(\R^d)$, under the law $\mu_{\cal C_{\eta}}$, $\sigma$ is almost surely a smooth periodic function on $[0,\nu] \times \Lambda$.
Let us note that we can rewrite \eqref{Hamiltonian_H} as 
\begin{equation}
\label{H^epsilon_2}
H=\bigoplus_{n \in \N} \biggl(H_{n,\varrho^{\epsilon}}-\frac{(\tau^\epsilon)^2}{2}-E^{\epsilon}\biggr)\,,
\end{equation}
where 
\begin{equation}
\label{H^epsilon_3}
H_{n,\varrho} \deq \Biggl[\nu \sum_{i=1}^{n} (\kappa-\Delta/2)_i + \frac{\nu^2}{2} \sum_{i,j=1}^{n}v^{\epsilon}(x_i-x_j)\Biggr]-\varrho \nu^2 n+\frac{\varrho^2\nu^2}{2}
\end{equation}
and 
\begin{equation*}
\varrho^{\epsilon}\deq \frac{\varrho_\nu}{\nu}+\frac{\tau^\epsilon}{\nu}\,,
\end{equation*}
where we recall the definition \eqref{rho_nu_definition} of $\varrho_\nu$.
Namely, from \eqref{Hamiltonian_H}, using notation as in \eqref{H^epsilon_3} as well as $\widehat{v}_{\epsilon}(0)=0$, it follows that $H$ acts on the $n$-th sector of Fock space as
\begin{multline*}
\Biggl[\nu \sum_{i=1}^{n} (\kappa-\Delta/2)_i+\frac{\nu^2}{2} \sum_{i,j=1}^{n} v_{\epsilon}(x_i-x_j)\Biggr]
-\tau_{\epsilon} (n \nu - \varrho_\nu)-\nu \varrho_\nu \widehat{v}_{\epsilon}(0)n+\frac{1}{2} \widehat{v}_\epsilon(0) \varrho_\nu^2-E^{\epsilon}
\\
=H_{n,0}-\biggl(\frac{\varrho_\nu}{\nu}+\frac{\tau^{\epsilon}}{\nu}\biggr)\nu^2 n +\biggl(\frac{\varrho_\nu}{\nu}+\frac{\tau^{\epsilon}}{\nu}\biggr)^2\frac{\nu^2}{2}-\frac{(\tau^{\epsilon})^2}{2}-E^{\epsilon}\,,
\end{multline*}
which gives us \eqref{H^epsilon_2}.
In the sequel, we denote
\begin{equation}
\label{T^epsilon}
T^{\epsilon}:=\frac{(\tau^{\epsilon})^2}{2}+E^{\epsilon}\,.
\end{equation}
In particular, from \eqref{T^epsilon} and \eqref{tau(epsilon)_bound}, it follows that 
\begin{equation}
\label{T^epsilon_bound}
|T^{\epsilon}| \lesssim_{\kappa,v} \chi(\epsilon)^2\,.
\end{equation}


In Lemmas \ref{Partition_function_rate_of_convergence_1} and \ref{Partition_function_rate_of_convergence_2} below, we give the functional integral representation of the (quantum) relative partition function \eqref{Z^epsilon_quantum}
and the (classical) relative partition function $\zeta^{W^{\epsilon}}$ corresponding to the interaction $W^{\epsilon}$ given by \eqref{W^epsilon}. In both cases, the representation is based on the use of the Hubbard-Stratonovich transformation \cite{FKSS_2020}.

\begin{lemma} 
\label{Partition_function_rate_of_convergence_1}
The relative partition function \eqref{Z^epsilon_quantum} can be written as $\cal Z=\lim_{\eta \rightarrow 0} \cal Z_{\eta}$,
where
\begin{equation}
\label{Z^{epsilon}_{eta}}
\cal Z_{\eta} \deq \int \mu_{\cal C_{\eta}}(\dd \sigma)\,\ee^{T^{\epsilon}-\frac{\ii \tau^\epsilon[\sigma]}{\nu}}\,\ee^{F_2(\sigma)}\,,
\end{equation}
for
\begin{multline} 
\label{F_2}
F_2(\sigma) = - \sum_{\mathbf{r} \in (\nu \N)^3} \frac{\ind{\abs{\mathbf{r}} > 0} \, \ee^{-\kappa \abs{\mathbf{r}}}}{\abs{\mathbf{r}}} \int_{[0,\nu]^3} \dd \mathbf{\tau} \int \dd \mathbf{x} \,
\sigma(\tau_2, x_2) \, \sigma(\tau_3, x_3)
\\
\times  \int \bb W_{x_1, x_3}^{\tau_1 + r_3, \tau_3}(\dd \omega_3) \, \bb W_{x_3, x_2}^{\tau_3 + r_2, \tau_2}(\dd \omega_2) \, \bb W_{x_2, x_1}^{\tau_2 + r_1, \tau_1}(\dd \omega_1) \, \ee^{\ii \int \dd s \, \sigma([s]_\nu, \omega_1(s))}\,,
\end{multline}
which satisfies
\begin{equation}
\label{Re_F2}
\re F_2 \leq 0\,.
\end{equation}
In \eqref{Z^{epsilon}_{eta}}, we write 
\begin{equation}
\label{bracket_sigma}
[\sigma] \deq \int_{0}^{\nu} \dd \tau \int \dd x\, \sigma(\tau,x) 
\end{equation}
and in 
\eqref{F_2}, we write 
\begin{equation}
\label{bracket_t}
[t]_{\nu} \deq (t \,\mathrm{mod}\,\nu) \in [0,\nu)\,.
\end{equation}
\end{lemma}

\begin{proof}
The proof follows from that of \cite[Proposition 3.12]{FKSS_2020} by setting
\begin{equation}
\label{Lemma_5.6_proof}
f(\tau,x)=\sum_{i=1}^{n} \delta(x-\omega_i(\tau))-\frac{\varrho_\nu}{\nu}-\frac{\tau^{\epsilon}}{\nu}\,.
\end{equation}
Note that the only difference then is the factor $\ee^{T^{\epsilon}}$ coming from the constant term \eqref{T^epsilon} in \eqref{H^epsilon_2} and the extra term $-\frac{\tau^{\epsilon}}{\nu}$ in \eqref{Lemma_5.6_proof} above. By applying the Hubbard-Stratonovich transformation
\begin{equation*}
\ee^{-\frac{1}{2} \langle f,\mathcal{C}_\eta f \rangle}= \int \mu_{\cal C_{\eta}}(\dd \sigma)\,\ee^{\ii \langle f, \sigma \rangle}\,,
\end{equation*}
this new term adds a factor of 
\begin{equation*}
\ee^{-\frac{\ii \tau^{\epsilon}}{\nu}\int_0^\nu \dd \tau \int \dd x\, \sigma(\tau,x)}=
\ee^{-\frac{\ii \tau^{\epsilon}[\sigma]}{\nu}}
\end{equation*}
to the integrand of the functional integral \eqref{Z^{epsilon}_{eta}}. The formula \eqref{F_2} is given in \cite[Lemma 5.4]{FKSS_2020}.
\end{proof}

We now explain how to obtain the functional integral representation for the classical partition function in our setting.
We define $\mu_{v^{\epsilon}}(\dd \xi)$ to be the real Gaussian measure with covariance 
\begin{equation}
\label{Gaussian_measure_classical}
\int \mu_{v^{\epsilon}}(\dd \xi) \,\xi(x)\,\xi(\tilde x)=v^{\epsilon}(x-\tilde x)\,.
\end{equation}
Since $v \in C_c^{\infty}(\R^d)$, under the law $\mu_{v^{\epsilon}}$, $\xi$ is almost surely a smooth periodic function on $\Lambda$. We have the following representation.

\begin{lemma}
\label{Partition_function_rate_of_convergence_2}
We have 
\begin{equation}
\label{zeta^{epsilon}}
\zeta^{W^\epsilon} = \int \mu_{v^{\epsilon}}(\dd \xi)\,\ee^{T^{\epsilon}-\ii \tau^\epsilon \langle \xi,1\rangle_{L^2}}\,\ee^{f_2(\xi)}\,,
\end{equation}
for
\begin{equation}
\label{f_2}
f_2(\xi)=-\int_{[0,\infty)^{3}} \dd \mathbf{r}\, \frac{\ee^{-\kappa|\mathbf{r}|}}{|\mathbf{r}|}\,\int\dd \mathbf{x}\,\xi(x_2)\,\xi(x_3)
\int \mathbb{W}_{x_{1},x_{3}}^{r_{3},0}(\dd \omega_{3})\,\mathbb{W}_{x_{3},x_{2}}^{r_{2},0}(\dd \omega_{2})\,\mathbb{W}_{x_{2},x_{1}}^{r_{1},0}(\dd \omega_{1})\,\ee^{\ii \int \dd s\,\xi(\omega_{1}(s))}\,,
\end{equation}
which satisfies
\begin{equation}
\label{Re_f2}
\re f_2 \leq 0\,.
\end{equation}
Note that in \eqref{zeta^{epsilon}}, $\langle \xi,1 \rangle_{L^2}=\int \dd x\,\xi(x)$ denotes the $L^2$ inner product. 
\end{lemma}
\begin{proof}
We note that for fixed $0<N<\infty$, we have

\begin{multline}
\label{W^epsilon_N}
W^{\epsilon}_N:=\frac{1}{2} \int \dd x\,\dd \tilde x\,\bigl(\abs{\phi_N(x)}^2- \E \qb{\abs{\phi_N(x)}^2}\bigr)\,v(x-\tilde x)\,
\bigl(\abs{\phi_N(\tilde x)}^2- \E \qb{\abs{\phi_N(\tilde x)}^2}\bigr)
\\
-\tau^{\epsilon} \int \dd x \,\bigl(\abs{\phi_N(x)}^2- \E \qb{\abs{\phi_N(x)}^2}\bigr)-E^{\epsilon}
\\
=\frac{1}{2} \int \dd x\,\dd \tilde x\,\Bigl(\abs{\phi_N(x)}^2- \bigl\{\E \qb{\abs{\phi_N(x)}^2}+\tau^{\epsilon}\bigr\}\Bigr)\,v(x-\tilde x)\,
\Bigl(\abs{\phi_N(\tilde x)}^2- \bigl\{\E \qb{\abs{\phi_N(\tilde x)}^2}+\tau^{\epsilon}\bigr\}\Bigr)-T^{\epsilon}\,.
\end{multline}
Here, we recall \eqref{T^epsilon}.
The claim then follows by using \eqref{W^epsilon_N} in the proof of \cite[Proposition 4.1]{FKSS_2020}. 
In particular, by applying the Hubbard-Stratonovich transformation
\begin{equation*}
\ee^{-\frac{1}{2} \langle f,v^{\epsilon} f \rangle}= \int \mu_{v^{\epsilon}}(\dd \xi)\,\ee^{\ii \langle f, \xi \rangle}\,,
\end{equation*}
for 
\begin{equation*}
f(x)=\abs{\phi_N(x)}^2-\E \qb{\abs{\phi_N(x)}^2}+\tau^{\epsilon}\,,
\end{equation*}
the 
$\tau^{\epsilon}$ terms in \eqref{W^epsilon_N} add a factor of $\ee^{-\ii \tau^\epsilon \langle \xi,1\rangle_{L^2}}$ 
to the integrand of the functional integral \eqref{zeta^{epsilon}}.
Here, we also recall \eqref{T^epsilon}. The formula \eqref{f_2} is given in \cite[Lemma 5.4]{FKSS_2020}.
\end{proof}

Recalling \eqref{delta_nu_integral}, let us note the following useful result.
\begin{lemma}
\label{Lemma 5.2'}
If $\sigma=\sigma(\tau,x)$ has law $\mu_{\cal C_{\eta}}$ with covariance \eqref{Gaussian_measure_quantum}, then its time average
\begin{equation}
\label{sigma_time_average}
\langle \sigma \rangle (x) \deq \frac{1}{\nu}\,\int_{0}^{\nu} \dd \tau\, \sigma(\tau,x)
\end{equation}
has law $\mu_{v^{\epsilon}}$ with covariance \eqref{Gaussian_measure_classical}. 
\end{lemma}
We observe that
\begin{equation}
\label{sigma_time_average_observation}
\big \langle \langle \sigma \rangle,1 \big \rangle_{L^2}=\frac{1}{\nu}\,[\sigma]\,,
\end{equation}
for $[\sigma]$ as in \eqref{bracket_sigma}.
Therefore, we can rewrite \eqref{zeta^{epsilon}} as
\begin{equation}
\label{zeta^{epsilon}_{eta}_2}
\zeta^{W^\epsilon} = \int \mu_{\cal C_{\eta}}(\dd \sigma)\,\ee^{T^{\epsilon}-\frac{\ii \tau^\epsilon [\sigma]}{\nu}}\,\ee^{f_2(\langle \sigma \rangle)}\,.
\end{equation}
We note the following result.

\begin{lemma}
\label{Partition_function_rate_of_convergence_3}
Uniformly in $\eta>0$, we have
\begin{equation*}
\Bigl|\cal Z_{\eta}-\zeta^{W^\epsilon}\Bigr|\leq \ee^{T^{\epsilon}}\, \biggl(\int\mu_{\cal C_{\eta}}(\dd \sigma)\bigl|F_2(\sigma)-f_2(\langle \sigma \rangle)\bigr|^2\biggr)^{1/2}\,.
\end{equation*}
\end{lemma}
\begin{proof}
By \eqref{Z^{epsilon}_{eta}} and \eqref{zeta^{epsilon}_{eta}_2}, we have
\begin{equation}
\label{Partition_function_rate_of_convergence_3_proof}
\cal Z_{\eta}-\zeta^{W^\epsilon}= \int \mu_{\cal C_{\eta}}(\dd \sigma)\,\ee^{T^{\epsilon}-\frac{\ii \tau^\epsilon [\sigma]}{\nu}}\,\Bigl(\ee^{F_2(\sigma)}-\ee^{f_2(\langle \sigma \rangle)}\Bigr)\,.
\end{equation}
By \eqref{Re_F2} and \eqref{Re_f2} and the elementary inequality $|\ee^{a}-\ee^{b}| \leq |a-b|$ for $a,b \in \C$ with $\re a, \re b \leq 0$, we have that for all $\sigma$
\begin{equation}
\label{Partition_function_rate_of_convergence_3_proof_2}
\Bigl|\ee^{F_2(\sigma)}-\ee^{f_2(\langle \sigma \rangle)}\Bigr| \leq \bigl|F_2(\sigma)-f_2(\langle \sigma \rangle)\bigr|\,.
\end{equation}
The claim follows from \eqref{Partition_function_rate_of_convergence_3_proof}, \eqref{Partition_function_rate_of_convergence_3_proof_2}, and the Cauchy-Schwarz inequality.
\end{proof}

In order to simplify notation in the sequel, we define the function $\Theta\col \R^+ \rightarrow \R$ by
\begin{equation}
\label{Theta_function}
\Theta(t) \deq 
\begin{cases}
\sqrt{t} & \text{if } d=2
\\
\sqrt{t}\log t^{-1}& \text{if } d=3\,.
\end{cases}
\end{equation}
Note that the upper bound in Proposition \ref{Partition_function_rate_of_convergence} 
can then be rewritten as $\ee^{C_1 \chi(\epsilon)^2}\Theta^{1/2}(\nu)$, with $\chi$ given by \eqref{sigma(epsilon)}. We prove the following two estimates, which correspond to quantitative versions of \cite[Lemma 5.9]{FKSS_2020} and \cite[Lemma 5.10]{FKSS_2020} respectively, and in turn let us use the bound from Lemma \ref{Partition_function_rate_of_convergence_3} to prove Proposition \ref{Partition_function_rate_of_convergence}.
\begin{lemma}
\label{Partition_function_rate_of_convergence_4}
Uniformly in $\eta>0$, we have
\begin{equation*}
\biggl|\int\mu_{\cal C_{\eta}}(\dd \sigma)\,\overline{F_2(\sigma)}F_2(\sigma)-\int\mu_{v^{\epsilon}}(\dd \xi)\,\overline{f_2(\xi)}f_2(\xi)\biggr|\lesssim_{\kappa,v} \frac{\Theta(\nu)}{\epsilon^{5d+1}}\,.
\end{equation*}
\end{lemma}

\begin{lemma}
\label{Partition_function_rate_of_convergence_5}
Uniformly in $\eta>0$, we have
\begin{equation*}
\biggl|\int\mu_{\cal C_{\eta}}(\dd \sigma)\,\overline{f_2(\langle \sigma \rangle)}F_2(\sigma)-\int\mu_{v^{\epsilon}}(\dd \xi)\,\overline{f_2(\xi)}f_2(\xi)\biggr|\lesssim_{\kappa,v}\frac{\Theta(\nu)}{\epsilon^{5d+1}}\,.
\end{equation*}
\end{lemma}
Assuming Lemmas \ref{Partition_function_rate_of_convergence_4} and \ref{Partition_function_rate_of_convergence_5} for now, we can prove the convergence rate given in Proposition \ref{Partition_function_rate_of_convergence}.
\begin{proof}[Proof of Proposition \ref{Partition_function_rate_of_convergence}]
The claim follows from Lemmas \ref{Partition_function_rate_of_convergence_1}--\ref{Partition_function_rate_of_convergence_5} by recalling \eqref{T^epsilon_bound}.
\end{proof}

The rest of this section is devoted to showing Lemmas \ref{Partition_function_rate_of_convergence_4} and \ref{Partition_function_rate_of_convergence_5}. 
Before proceeding with the proofs, we need to introduce some notation and definitions. Throughout we use the convention that, given a path $\omega \in \Omega^{\tau,\tilde{\tau}}$ and a function $f$, we write
\begin{equation*}
\int \dd s \,f(\omega(s)) \equiv \int_{\tilde{\tau}}^{\tau} \dd s\,f(\omega(s))\,.
\end{equation*}
We define the following quantities that will allow us to rewrite the terms that arise in the sequel.

\begin{definition}[Classical interactions]
\label{V_classical}
Let $x ,\tilde x \in \Lambda$ and $\omega \in \Omega^{\tau_1,\tilde \tau_1}, \tilde \omega \in \Omega^{\tau_2, \tilde \tau_2}
$ be continuous paths. We then define the \emph{point-point interaction}
\begin{equation*}
(\mathbb{V}^{\epsilon})_{x,\tilde x} \deq \int \mu_{v^{\epsilon}}(\dd \xi)\,\xi(x)\,\xi(\tilde x)=v^{\epsilon}(x-\tilde x)\,,
\end{equation*}
the \emph{point-path interaction}
\begin{equation*}
(\mathbb{V}^{\epsilon})_{x}(\omega) \deq \int \mu_{v^{\epsilon}}(\dd \xi)\,\int \dd s\,\xi(\omega(s))=\int \dd s\,v^{\epsilon}(x-\omega(s))\,,
\end{equation*}
and the \emph{path-path interaction}
\begin{equation*}
\mathbb{V}^{\epsilon}(\omega,\tilde \omega) \deq \int \mu_{v^{\epsilon}}(\dd \xi)\,\int \dd s\,\xi(\omega(s))\,\int \dd \tilde{s}\,\xi(\tilde{\omega}(\tilde{s}))=\int \dd s\,\int \dd \tilde{s}\,v^{\epsilon}(\omega(s)-\tilde{\omega}(\tilde{s}))\,.
\end{equation*}
\end{definition}
In what follows, we use the notation $x_{i,0} \equiv x_i, \tilde{x}_{i,1} \equiv \tilde{x}_i$ for $i=1,2$ and we write 
\begin{equation}
\label{set_A}
A \deq \{2,3\} \times \{0,1\}\,. 
\end{equation}
Arguing analogously as for \cite[(5.8)--(5.9)]{FKSS_2020}, we get that
\begin{equation}
\label{5.8'}
\int \mu_{v^{\epsilon}}(\dd \xi)\,\overline{f_2(\xi)}\,f_2(\xi)=\int_{[0,\infty)^3} \dd \mathbf{r} \,\frac{\ee^{-\kappa |\mathbf{r}|}}{|\mathbf{r}|}\,\int_{[0,\infty)^3} \dd \mathbf{\tilde r}\, \frac{\ee^{-\kappa |\mathbf{\tilde r}|}}{|\mathbf{\tilde r}|}\,I^{\epsilon}(\f r, \tilde{\f r})\,,
\end{equation}
where 
\begin{multline}
\label{5.9'}
I^{\epsilon}(\f r, \tilde{\f r}) \deq \int \dd \mathbf{x} \, \int \dd \mathbf{\tilde x}\,\int \mathbb{W}^{r_3,0}_{x_1,x_3}(\dd \omega_3)\,\mathbb{W}^{r_2,0}_{x_3,x_2}(\dd \omega_2)\,\mathbb{W}^{r_1,0}_{x_2,x_1}(\dd \omega_1)\,
\\
\times
\mathbb{W}^{\tilde{r}_3,0}_{\tilde{x}_1,\tilde{x}_3}(\dd \tilde{\omega}_3)\,\mathbb{W}^{\tilde{r}_2,0}_{\tilde{x}_3,\tilde{x}_2}(\dd \tilde{\omega}_2)\,\mathbb{W}^{\tilde{r}_1,0}_{\tilde{x}_2,\tilde{x}_1}(\dd \tilde{\omega}_1)\,
\ee^{-\frac{1}{2}\bigl(\mathbb{V}^{\epsilon}(\omega_1,\omega_1)+\mathbb{V}^{\epsilon}(\tilde \omega_1,\tilde \omega_1)-2\mathbb{V}^{\epsilon}(\omega_1,\tilde \omega_1)\bigr)}\,
\\
\times
\sum_{\Pi \in \fra{M}(A)} \prod_{\{a,b\} \in \Pi} (\mathbb{V}^{\epsilon})_{x_a,x_b}\,\prod_{a \in A \setminus [\Pi]} \ii \Bigl((\mathbb{V}^{\epsilon})_{x_a}(\omega_1)-
(\mathbb{V}^{\epsilon})_{x_a}(\tilde{\omega}_1)\Bigr)\,.
\end{multline}
Here, we write $|\mathbf{r}|=r_1+r_2+r_3, |\mathbf{\tilde r}|=\tilde r_1+\tilde r_2+\tilde r_3$. Moreover, we denote by $\fra{M}(A)$ the set of partial pairings on the set $A$.

\begin{definition}[Quantum interactions]
\label{V_quantum}
Let $(\tau,x),(\tilde{\tau},\tilde{x}) \in [0,\nu] \times \Lambda$ and let $\omega \in \Omega^{\tau_{1},\tilde{\tau}_{1}},\tilde{\omega} \in \Omega^{\tau_{2},\tilde{\tau}_{2}}$ be continuous paths. With $\delta_{\eta}$ given by \eqref{approximate_delta_function}, we define the \emph{point-point interaction}
\begin{equation*}
(\mathbb{V}_{\eta})_{x,\tilde{x}}^{\tau,\tilde{\tau}} \deq \int \mu_{\cal{C}_{\eta}}(\dd \sigma)\,\sigma(\tau,x)\,\sigma(\tilde{\tau},\tilde{x})=\nu\,\delta_{\eta}(\tau-\tilde{\tau})\,v^{\epsilon}(x-\tilde{x})\,,
\end{equation*}
the \emph{point-path interaction}
\begin{equation*}
(\mathbb{V}_{\eta})_{x}^{\tau}(\omega) \deq \int \mu_{\cal{C}_{\eta}}(\dd \sigma)\,\sigma(\tau,x)\,
\int_{0}^{\nu} \dd t\,\int \dd s\,\delta(t-[s]_{\nu})\,\sigma(t,\omega(s))
=\nu\,\int \dd s\, \delta_{\eta}(\tau-[s]_{\nu})\,v^{\epsilon}(x-\omega(s))\,,
\end{equation*}
and the \emph{path-path interaction}
\begin{multline*}
\mathbb{V}_{\eta}(\omega,\tilde{\omega}) \deq \int \mu_{\cal{C}_{\eta}}(\dd \sigma)\,\int_{0}^{\nu} \dd t\,\int \dd s\, \sigma(\tau,x)\,\delta(t-[s]_{\nu})\,\sigma(t,\omega(s))\,\int_{0}^{\nu}\dd \tilde{t}\,\int \dd \tilde{s}\,\delta(\tilde{t}-[\tilde{s}]_{\nu})\,\sigma(\tilde{t},\tilde{\omega}(\tilde{s}))
\\
=\nu\,\int \dd s\,\int \dd \tilde{s}\, \delta_{\eta}([s]_{\nu}-[\tilde{s}]_{\nu})\,v^{\epsilon}(\omega(s)-\tilde{\omega}(\tilde{s}))\,.
\end{multline*}
Here, we recall \eqref{bracket_t}.
\end{definition}

Arguing analogously as for \cite[(5.17)--(5.18), (5.21)]{FKSS_2020}, we have that
\begin{equation} 
\label{5.17'}
\int \mu_{\cal{C}_{\eta}}(\dd \sigma) \, \overline{F_2(\sigma)} F_2(\sigma)
=
\sum_{\mathbf{r} \in (\nu \N)^3} \frac{\ind{|\mathbf{r}|>0}\,\ee^{-\kappa|\mathbf{r}|}}{|\mathbf{r}|}
\sum_{\tilde{\mathbf{r}} \in (\nu \N)^3} \frac{\ind{|\tilde{\mathbf{r}}|>0}\,\ee^{-\kappa|\tilde{\mathbf{r}}|}}{|\tilde{\mathbf{r}}|}\,
J(\mathbf{r},\tilde{\mathbf{r}})\,,
\end{equation}
where 
\begin{multline}
\label{5.21'}
J(\mathbf{r},\tilde{\mathbf{r}}) \deq \int_{[0,\nu]^3} \dd \f \tau \,\int_{[0,\nu]^3} \dd \tilde{\f \tau}\, \int \dd \mathbf{x}\,\int\dd \tilde{\mathbf{x}}\,\int \mathbb{W}^{\tau_1+r_3,\tau_3}_{x_1,x_3}(\dd \omega_3)\,\mathbb{W}^{\tau_3+r_2,\tau_2}_{x_3,x_2}(\dd \omega_2)\,\mathbb{W}^{\tau_2+r_1,\tau_1}_{x_2,x_1}(\dd \omega_1)
\\
\times
\mathbb{W}^{\tilde{\tau}_1+\tilde{r}_3,\tilde{\tau}_3}_{\tilde{x}_1,\tilde{x}_3}(\dd \tilde{\omega}_3)\,\mathbb{W}^{\tilde{\tau}_3+\tilde{r}_2,\tilde{\tau}_2}_{\tilde{x}_3,\tilde{x}_2}(\dd \tilde{\omega}_2)\,\mathbb{W}^{\tilde{\tau}_2+\tilde{r}_1,\tilde{\tau}_1}_{\tilde{x}_2,\tilde{x}_1}(\dd \tilde{\omega}_1)\,
\ee^{-\frac{1}{2}\bigl(\mathbb{V}_{\eta}(\omega_1,\omega_1)+\mathbb{V}_{\eta}(\tilde \omega_1,\tilde \omega_1)-2\mathbb{V}_{\eta}(\omega_1,\tilde \omega_1)\bigr)}\,
\\
\times
\sum_{\Pi \in \fra{M}(A)} \prod_{\{a,b\} \in \Pi} (\mathbb{V}_{\eta})_{x_a,x_b}^{\tau_a,\tau_b}\,\prod_{a \in A \setminus [\Pi]} \ii \Bigl((\mathbb{V}_{\eta})_{x_a}^{\tau_a}(\omega_1)-
(\mathbb{V}_{\eta})_{x_a}^{\tau_a}(\tilde{\omega}_1)\Bigr)\,.
\end{multline}

The first step in the proof of Lemma \ref{Partition_function_rate_of_convergence_4} is to compare \eqref{5.21'} with $\nu^6 I^{\epsilon}(\mathbf{r},\tilde{\mathbf{r}})$, which appears in a Riemann sum of mesh size $\nu$ for \eqref{5.8'}. We show the following quantitative estimate.

\begin{lemma}[Approximation of $J(\mathbf{r},\tilde{\mathbf{r}})$]
\label{J_approximation}
For all $\f r,\tilde{\f r} \in (\nu \N)^3$ with $|\f r|, |\tilde{\f r}|>0$, we have that uniformly in $\eta>0$
\begin{equation}
\label{J_approximation_claim}
J(\mathbf{r},\tilde{\mathbf{r}})=\nu^6 I^{\epsilon}(\mathbf{r},\tilde{\mathbf{r}})+O\Biggl(\frac{\nu^{13/2}}{\epsilon^{5d+1}}\,\biggl(1+\frac{1}{|\f r|^{d/2}}\biggr)\,\biggl(1+\frac{1}{|\tilde{\f r}|^{d/2}}\biggr)\,(1+|\f r|+|\tilde{\f r}|)^6\Biggr)\,.
\end{equation}
\end{lemma}

The second step in the proof of Lemma \ref{Partition_function_rate_of_convergence_4} consists in giving a quantitative estimate on the error obtained by approximating the integral \eqref{5.8'} with the above Riemann sum. To this end, we prove the following estimate.
\begin{lemma}[Quantitative Riemann sum approximation for \eqref{5.8'}]
\label{Riemann_sum_I}
Recalling \eqref{Theta_function}, we have that uniformly in $\eta>0$
\begin{multline*}
\Biggl|\int_{[0,\infty)^3} \dd \mathbf{r} \,\frac{\ee^{-\kappa |\mathbf{r}|}}{|\mathbf{r}|}\,\int_{[0,\infty)^3} \dd \mathbf{\tilde r}\, \frac{\ee^{-\kappa |\mathbf{\tilde r}|}}{|\mathbf{\tilde r}|}\,I^{\epsilon}(\f r, \tilde{\f r})-\nu^6 \sum_{\mathbf{r} \in (\nu \N)^3}\,\frac{\ind{|\mathbf{r}|>0}\,\ee^{-\kappa |\mathbf{r}|}}{|\mathbf{r}|}\, \sum_{\mathbf{\tilde r} \in (\nu \N)^3}\, \frac{\ind{|\mathbf{\tilde r}|>0}\,\ee^{-\kappa |\mathbf{\tilde r}|}}{|\mathbf{\tilde r}|}\,I^{\epsilon}(\f r, \tilde{\f r})\Biggr| 
\\
\lesssim_{\kappa,v} \frac{\Theta(\nu)}{\epsilon^{5d+1}}\,.
\end{multline*}
\end{lemma}

With the above two results, we have all of the necessary tools to prove Lemma \ref{Partition_function_rate_of_convergence_4}.
\begin{proof}[Proof of Lemma \ref{Partition_function_rate_of_convergence_4}]
The claim follows from Lemmas \ref{J_approximation} and \ref{Riemann_sum_I} by using \eqref{5.8'} and \eqref{5.17'}.
\end{proof}

We now give the proofs of Lemmas \ref{J_approximation} and \ref{Riemann_sum_I}. In the proof of Lemma \ref{J_approximation}, we use the following estimates that are obtained from Definition \ref{V_quantum} and Lemma \ref{v^{epsilon}_properties} \ref{itm:veps1}, and \eqref{delta_nu_integral}.
\begin{equation}
\label{V_quantum_estimates}
|(\mathbb{V}_{\eta})_{x,\tilde x}^{\tau,\tilde \tau}| \leq \frac{C\nu}{\epsilon^d}\,\delta_{\eta}(\tau-\tilde \tau)\,,\quad |(\mathbb{V}_{\eta})_{x}^{\tau}(\omega')| \leq \frac{C}{\epsilon^d}\,\bigl[(\tau'-\tilde{\tau}')+\nu\bigr]\,.
\end{equation}
In the second estimate in \eqref{V_quantum_estimates}, we take $\omega' \in \Omega^{\tau',\tilde{\tau}'}$.
We note that for paths $\omega',\tilde \omega'$ with 
\begin{equation*}
f(\tau,x)=\int \dd s\, \delta(\tau-[s]_{\nu})\,\delta(x-\omega_1(s))-
\int \dd \tilde s\,\delta(\tau-[\tilde s]_{\nu})\,\delta(x-\tilde{\omega}_1(s))\,
\end{equation*}
we have that
\begin{equation}
\label{V_quantum_positive}
\mathbb{V}_{\eta}(\omega,\omega)+\mathbb{V}_{\eta}(\tilde \omega,\tilde \omega)-2\mathbb{V}_{\eta}(\omega,\tilde \omega) =\langle f, \cal C_{\eta} f \rangle \geq 0\,,
\end{equation}
since $\delta_{\eta}$ and $v^{\epsilon}$ are of positive type.

\begin{proof}[Proof of Lemma \ref{J_approximation}.]
We recall the definition \eqref{set_A} of the set $A$.
For fixed $\f \tau,\f \tau \in [0,\nu]^3$, arguing analogously as in  \cite[(5.21)]{FKSS_2020}, we apply a time translation and rewrite \eqref{5.21'} as
\begin{align}
\notag
&\int_{[0,\nu]^3} \dd \f \tau \,\int_{[0,\nu]^3} \dd \tilde{\f \tau}\, \int \dd \mathbf{x}\,\int\dd \tilde{\mathbf{x}}\,
\int \mathbb{W}^{\tau_1+|\mathbf{r}|,\tau_3+r_1+r_2}_{x_1,x_3}(\dd \omega_3)\,\mathbb{W}^{\tau_3+r_1+r_2,\tau_2+r_1}_{x_3,x_2}(\dd \omega_2)\,\mathbb{W}^{\tau_2+r_1,\tau_1}_{x_2,x_1}(\dd \omega_1)
\\
\notag
&\mathbb{W}^{\tilde{\tau}_1+|\mathbf{\tilde{r}}|,\tilde{\tau}_3+\tilde{r}_1+\tilde{r}_2}_{\tilde{x}_1,\tilde{x}_3}(\dd \tilde{\omega}_3)\,\mathbb{W}^{\tilde{\tau}_3+\tilde{r}_1+\tilde{r}_2,\tilde{\tau}_2+\tilde{r}_1}_{\tilde{x}_3,\tilde{x}_2}(\dd \tilde{\omega}_2)\,\mathbb{W}^{\tilde{\tau}_2+\tilde{r}_1,\tilde{\tau}_1}_{\tilde{x}_2,\tilde{x}_1}(\dd \tilde{\omega}_1)\,
\ee^{-\frac{1}{2}\bigl(\mathbb{V}_{\eta}(\omega_1,\omega_1)+\mathbb{V}_{\eta}(\tilde \omega_1,\tilde \omega_1)-2\mathbb{V}_{\eta}(\omega_1,\tilde \omega_1)\bigr)}\,
\\
\label{5.21*}
&\sum_{\Pi \in \fra{M}(A)} \prod_{\{a,b\} \in \Pi} (\mathbb{V}_{\eta})_{x_a,x_b}^{\tau_a,\tau_b}\,
\prod_{a \in A \setminus [\Pi]} \ii \Bigl((\mathbb{V}_{\eta})_{x_a}^{\tau_a}(\omega_1)-
(\mathbb{V}_{\eta})_{x_a}^{\tau_a}(\tilde{\omega}_1)\Bigr)=J(\mathbf{r},\tilde{\mathbf{r}})\,.
\end{align}

In the sequel, we denote by $\hat k$ either the quantity $k$ or $\tilde k$. We hence consider paths $\hat{\omega}\col[\hat{\tau}_1, \hat{\tau_1}+|\hat{\f r}|] \rightarrow \Lambda$ obtained by concatenating $\hat \omega_1,\hat \omega_2, \hat \omega_3$ occurring in \eqref{5.21*}. In particular, 
\begin{equation}
\label{omega_hat}
\hat \omega_1=\hat \omega |_{[\hat{\tau}_{1}, \hat{\tau}_2+\hat r_1]}\,,\quad  \hat \omega_2=\hat \omega|_{[\hat{\tau}_{2}+\hat r_{1},\hat{\tau}_{3}+\hat {r}_{1}+\hat{r}_{2}]}\,,\quad \hat \omega_3=\hat \omega|_{[\hat{\tau}_{3}+\hat{r}_{1}+\hat{r}_{2},\hat{\tau}_{1}+|\hat{\f r}|]}\,.
\end{equation}
We define the following parameters for $a \in A$.
\begin{equation}
\label{parameter_choice_quantum}
\begin{cases}
{t}_{(2,0)}\deq {\tau}_{2}+{r}_{1}\,,\;\; {t}_{(3,0)}\deq {\tau}_{3}+{r}_{1}+{r}_{2}\,,\;\; {t}_{(2,1)} \deq \tilde{\tau}_{2}+\tilde{r}_{1}\,,\;\; {t}_{(3,1)}\deq \tilde{\tau}_{3}+\tilde{r}_{1}+\tilde{r}_{2}
\\
{s}_{(2,0)}\deq {\tau}_{1}+{r}_{1}\,,\;\;{s}_{(3,0)}\deq {\tau}_{1}+{r}_{1}+{r}_{2}\,, \;\;{s}_{(2,1)}\deq \tilde{\tau}_{1}+ \tilde{r}_{1} \,, \;\;{s}_{(3,1)}\deq  \tilde{\tau}_{1}+ \tilde{r}_{1}  +\tilde{r}_{2}
\\
{\omega}_{(2,0)}={\omega}_{(3,0)}  \deq  \omega\,, \;\;{\omega}_{(2,1)}={\omega}_{(3,1)}  \deq \tilde{\omega}  \,.
\end{cases}
\end{equation}
Note that the times $t_a$ in \eqref{parameter_choice_quantum} were chosen as initial and final times of the paths $\hat \omega_j$ in \eqref{omega_hat}. The $s_a$ satisfy $|s_a-t_a|\leq \nu$ and are used as approximations of the $t_a$ which help us remove the $\hat{\tau}_2, \hat{\tau}_3$ dependence. 

Using \eqref{V_quantum_positive}, \eqref{omega_hat} and Lemma \ref{5.12'} \ref{itm:512_1} below for the factor 
$
\ee^{-\frac{1}{2}\bigl(\mathbb{V}_{\eta}(\omega_1,\omega_1)+\mathbb{V}_{\eta}(\tilde \omega_1,\tilde \omega_1)-2\mathbb{V}_{\eta}(\omega_1,\tilde \omega_1)\bigr)}
$
occurring in the integrand of \eqref{5.21*}, recalling \eqref{parameter_choice_quantum} and integrating in the $\hat x_2,\hat x_3$ variables, we obtain from \eqref{5.21*} that

\begin{multline}
\label{5.25'}
J(\mathbf{r},\tilde{\mathbf{r}}) 
= \int_{[0,\nu]^3} \dd \f \tau \,\int_{[0,\nu]^3} \dd \tilde{\f \tau}\, \int \dd x_1\,\int\dd \tilde x_1\,\int \mathbb{W}^{\tau_1+|\f r|,\tau_1}_{x_1,x_1}(\dd \omega)\,\mathbb{W}^{\tilde \tau_1+|\tilde{\f r}|,\tilde{\tau}_1}_{\tilde x_1,\tilde x_1}(\dd \tilde \omega)\,
\\
\biggl[\ee^{-\frac{1}{2}\bigl(\mathbb{V}^{\epsilon}(\omega_1,\omega_1)+\mathbb{V}^{\epsilon}(\tilde \omega_1,\tilde \omega_1)-2\mathbb{V}^{\epsilon}(\omega_1,\tilde \omega_1)\bigr)}+O\Bigl(\bigl|\mathbb{\hat{V}}_{\eta}(\omega_1,\omega_1)\bigr|+\bigl|\mathbb{\hat{V}}_{\eta}(\tilde{\omega}_1,\tilde{\omega}_1)\bigr|+\bigl|\mathbb{\hat{V}}_{\eta}(\omega_1,\tilde{\omega}_1)\bigr|\Bigr)
\\
+O\biggl(\frac{\nu (r_1+ \tilde r_1 +\nu)}{\epsilon^d}\biggr)\biggr]
\sum_{\Pi \in \fra{M}(A)} \prod_{\{a,b\} \in \Pi} (\mathbb{V}_{\eta})_{\omega_a(t_a),\omega_b(t_b)}^{\tau_a,\tau_b}\,\prod_{a \in A \setminus [\Pi]} \ii \Bigl((\mathbb{V}_{\eta})_{\omega_a(t_a)}^{\tau_a}(\omega_1)-
(\mathbb{V}_{\eta})_{\omega_a(t_a)}^{\tau_a}(\tilde{\omega}_1)\Bigr)\,.
\end{multline}
We first estimate the contributions to \eqref{5.25'} coming from the two error terms occurring in the square brackets. We recall \eqref{V_hat_intermediate_estimate}, \eqref{I_i}, \eqref{omega_hat}, and apply
Fubini's theorem, Lemma \ref{Heat_kernel_estimates} \ref{itm:heatkernel1}--\ref{itm:heatkernel2}, and \eqref{delta_nu_integral} to deduce that
\begin{multline}
\label{V_hat_bound}
\int \mathbb{W}^{\tau_1+|\f r|,\tau_1}_{x_1,x_1}(\dd \omega)\,\mathbb{W}^{\tilde \tau_1+|\tilde{\f r}|,\tilde{\tau}_1}_{\tilde x_1,\tilde x_1}(\dd \tilde \omega)\,\Bigl(\bigl|\mathbb{\hat{V}}_{\eta}(\omega_1,\omega_1)\bigr|+\bigl|\mathbb{\hat{V}}_{\eta}(\tilde{\omega}_1,\tilde{\omega}_1)\bigr|+\bigl|\mathbb{\hat{V}}_{\eta}(\omega_1,\tilde{\omega}_1)\bigr|\Bigr)
\\
\leq \frac{C \sqrt{\nu}}{\epsilon^{d+1}}\,\biggl(1+\frac{1}{|\f r|^{d/2}}\biggr)\,\biggl(1+\frac{1}{|\tilde{\f r}|^{d/2}}\biggr)\,(1+|\f r|+|\tilde{\f r}|)^2\,.
\end{multline}
Note that here we used estimate \eqref{Heat_kernel_estimates_ii_2} when applying Lemma \ref{Heat_kernel_estimates} \ref{itm:heatkernel2}.
Combining \eqref{V_hat_bound} with \eqref{V_quantum_estimates}, and using \eqref{delta_nu_integral} for the $\dd \hat{\f \tau}$ integration, we deduce that the first error term in the square brackets in \eqref{5.25'} gives a contribution which is 
\begin{equation}
\label{error_term_quantum}
O\Biggl(\frac{\nu^{13/2}}{\epsilon^{5d+1}}\,\biggl(1+\frac{1}{|\f r|^{d/2}}\biggr)\,\biggl(1+\frac{1}{|\tilde{\f r}|^{d/2}}\biggr)\,(1+|\f r|+|\tilde{\f r}|)^6\Biggr)\,.
\end{equation}
Here, we note that there are at most $|A|=4$ factors of $\frac{1}{\epsilon^d}$ coming from \eqref{V_quantum_estimates}. Likewise, there are at most 4 additional factors of $1+|\f r|+|\tilde{\f r}|$.
Similarly, the second error term also gives a contribution which is bounded from above by \eqref{error_term_quantum} (note that now, we do not need to use \eqref{V_hat_bound}). Therefore, we can write \eqref{5.25'} as 
\begin{equation}
\label{5.25'_error_bound}
J(\mathbf{r},\tilde{\mathbf{r}}) =\bar{J}(\mathbf{r},\tilde{\mathbf{r}}) + O\Biggl(\frac{\nu^{13/2}}{\epsilon^{5d+1}}\,\biggl(1+\frac{1}{|\f r|^{d/2}}\biggr)\,\biggl(1+\frac{1}{|\tilde{\f r}|^{d/2}}\biggr)\,(1+|\f r|+|\tilde{\f r}|)^6\Biggr)\,,
\end{equation}
where
\begin{multline}
\label{J_bar}
\bar{J}(\mathbf{r},\tilde{\mathbf{r}}) \deq
\\
\int_{[0,\nu]^3} \dd \f \tau \,\int_{[0,\nu]^3} \dd \tilde{\f \tau}\, \int \dd x_1\,\int\dd \tilde x_1\,\int \mathbb{W}^{\tau_1+|\f r|,\tau_1}_{x_1,x_1}(\dd \omega)\,\mathbb{W}^{\tilde \tau_1+|\tilde{\f r}|,\tilde{\tau}_1}_{\tilde x_1,\tilde x_1}(\dd \tilde \omega)\,
\ee^{-\frac{1}{2}\bigl(\mathbb{V}^{\epsilon}(\omega_1,\omega_1)+\mathbb{V}^{\epsilon}(\tilde \omega_1,\tilde \omega_1)-2\mathbb{V}^{\epsilon}(\omega_1,\tilde \omega_1)\bigr)}
\\
\times
\sum_{\Pi \in \fra{M}(A)} \prod_{\{a,b\} \in \Pi} (\mathbb{V}_{\eta})_{\omega_a(t_a),\omega_b(t_b)}^{\tau_a,\tau_b}\,\prod_{a \in A \setminus [\Pi]} \ii \Bigl((\mathbb{V}_{\eta})_{\omega_a(t_a)}^{\tau_a}(\omega_1)-
(\mathbb{V}_{\eta})_{\omega_a(t_a)}^{\tau_a}(\tilde{\omega}_1)\Bigr)\,.
\end{multline}
In order to analyse \eqref{J_bar},  we note the following three estimates.
\begin{enumerate}[label=(\roman*)]
\item \label{itm:jbar1}
We have
\begin{multline}
\label{claim_(i)}
\ee^{-\frac{1}{2}\bigl(\mathbb{V}^{\epsilon}(\omega_1,\omega_1)+\mathbb{V}^{\epsilon}(\tilde \omega_1,\tilde \omega_1)-2\mathbb{V}^{\epsilon}(\omega_1,\tilde \omega_1)\bigr)}
\\
=\ee^{-\frac{1}{2}\bigl(\mathbb{V}^{\epsilon}\bigl(\omega|_{[\tau_1,\tau_1+r_1]},\omega|_{[\tau_1,\tau_1+r_1]}\bigr)+\mathbb{V}^{\epsilon}\bigl(\tilde \omega|_{[\tilde \tau_1,\tilde \tau_1+\tilde r_1]},\tilde \omega|_{[\tilde \tau_1,\tilde \tau_1+\tilde r_1]}\bigr)-2\mathbb{V}^{\epsilon}\bigl(\omega|_{[\tau_1,\tau_1+r_1]},\tilde \omega|_{[\tilde \tau_1,\tilde \tau_1+\tilde r_1]}\bigr)\bigr)}
\\
+O\biggl(\frac{\nu}{\epsilon^d}(1+r_1+\tilde r_1)\biggr)\,.
\end{multline}
\item \label{itm:jbar2}
For all $a,b \in A$, we have
\begin{multline}
\label{claim_(ii)}
\int_0^{\nu}\dd \tau_a\,\int_0^{\nu} \dd \tau_b\,(\mathbb{V}_{\eta})^{\tau_a,\tau_b}_{\omega_a(t_a),\omega_b(t_b)}
=\nu^2\,(\mathbb{V}^{\epsilon})_{\omega_a(s_a),\omega_b(s_b)}
\\
+O\biggl(\frac{\nu}{\epsilon^{d+1}}\,\int_0^{\nu}\dd \tau_a\,\int_0^{\nu}\dd \tau_b\,\delta_{\eta}(\tau_a-\tau_b)\,\bigl(|\omega_a(t_a)-\omega_a(s_a)|_{\Lambda}+|\omega_b(t_b)-\omega_b(s_b)|_{\Lambda}\bigr)\biggr)\,.
\end{multline}
\item \label{itm:jbar3}
For all $a \in A$, we have
\begin{multline}
\label{claim_(iii)}
\int_0^{\nu} \dd \tau_a\,(\mathbb{V}_{\eta})^{\tau_a}_{\omega_a(t_a)}(\hat \omega_1)=
\nu\,(\mathbb{V}^{\epsilon})_{\omega_a(s_a)}(\hat{\omega}|_{[\hat \tau_1,\hat \tau_1+\hat r_1]})
\\
+O\biggl(\frac{\hat r_1}{\epsilon^{d+1}}\,\int_0^{\nu} \dd \tau_a\,|\omega_a(t_a)-\omega_a(s_a)|_{\Lambda}\biggr)
+O\biggl(\frac{\nu^2}{\epsilon^d}\biggr)\,.
\end{multline}
\end{enumerate}
In order to show \ref{itm:jbar1}, we note that for paths $\omega',\tilde \omega'$ and
\begin{equation*}
f(x)=\int \dd s\, \delta(x-\omega_1(s))-\int \dd \tilde s\,\delta(x-\tilde{\omega}_1(s))\,,
\end{equation*}
we obtain
\begin{equation}
\label{V_classical_positive}
\mathbb{V}^{\epsilon}(\omega',\omega')+\mathbb{V}^{\epsilon}(\tilde \omega',\tilde \omega')-2\mathbb{V}^{\epsilon}(\omega',\tilde \omega')
=\int \dd x\,\dd y\, f(x) \,v^{\epsilon}(x-y)\,f(y)\geq 0\,,
\end{equation}
since $v^{\epsilon}$ is of positive type.

Claim \ref{itm:jbar1} then follows by using Definition \ref{V_classical}, \eqref{V_classical_positive}, and Lemma \ref{5.12'} \ref{itm:512_2}.  In order to prove claim \ref{itm:jbar2}, we use Definition \ref{V_quantum}, \eqref{parameter_choice_quantum} and \eqref{delta_nu_integral} to rewrite the left-hand side of \eqref{claim_(ii)} as
\begin{equation*}
\nu^2\,v^{\epsilon}\bigl(\omega_a(s_a)-\omega_b(s_b)\bigr)+\nu\,\int_0^{\nu}\dd \tau_a\,\int_0^{\nu} \dd \tau_b \,\delta_{\eta}(\tau_a-\tau_b)\,\Bigl[v^{\epsilon}\bigl(\omega_a(t_a)-\omega_b(t_b)\bigr)-v^{\epsilon}\bigl(\omega_a(s_a)-\omega_b(s_b)\bigr)\Bigr]\,,
\end{equation*}
which by Lemma \ref{v^{epsilon}_properties} \ref{itm:veps2} is of the form given by the right-hand side of \eqref{claim_(ii)}. In order to prove claim \ref{itm:jbar3}, we use first recall \eqref{omega_hat} and use $|\hat \tau_1-\hat \tau_2|\leq \nu$, Definition \ref{V_quantum}, \eqref{delta_nu_integral}, and Lemma \ref{v^{epsilon}_properties} \ref{itm:veps1} to rewrite the left-hand side of \eqref{claim_(iii)} as
\begin{multline}
\label{claim_(iii)_LHS}
\nu\,\int_0^{\nu} \dd \tau_a\,\int_{\hat \tau_1}^{\hat \tau_1+\hat r_1} \dd s\,\delta_{\eta}(\tau_a-[s]_{\nu})\,v^{\epsilon}\bigl(\omega_a(s_a)-\hat \omega(s)\bigr)
\\
+\nu\,\int_0^{\nu} \dd \tau_a\,\int_{\hat \tau_1}^{\hat \tau_1+\hat r_1} \dd s\,\delta_{\eta}(\tau_a-[s]_{\nu})\,\Bigl[v^{\epsilon}\bigl(\omega_a(t_a)-\hat \omega(s)\bigr)-v^{\epsilon}\bigl(\omega_a(s_a)-\hat \omega(s)\bigr)\Bigr]+O\biggl(\frac{\nu^2}{\epsilon^d}\biggr)\,.
\end{multline}
The last error term in \eqref{claim_(iii)_LHS} comes from replacing $\hat \tau_2+ \hat r_2$ by $\hat \tau_1+\hat r_1$ in the upper limit of the $s$ integral.
For the first term, we use \eqref{delta_nu_integral} to integrate in $\tau_a$. For the second term, we use Lemma \ref{v^{epsilon}_properties} \ref{itm:veps2} followed by \eqref{delta_nu_integral} to integrate in $s$. It follows that \eqref{claim_(iii)_LHS} is of the form given by the right-hand side of \eqref{claim_(iii)}.

We recall \eqref{5.9'}, \eqref{omega_hat}, \eqref{parameter_choice_quantum}. Then, we use the estimates \ref{itm:jbar1}--\ref{itm:jbar3} above together with a telescoping argument in \eqref{J_bar}, and argue analogously as in the proof of \eqref{5.25'_error_bound} to deduce that 
\begin{equation}
\label{bar_J_approximation}
\bar{J}(\mathbf{r},\tilde{\mathbf{r}})=\nu^6 I^{\epsilon}(\mathbf{r},\tilde{\mathbf{r}})+O\Biggl(\frac{\nu^{13/2}}{\epsilon^{5d+1}}\,\biggl(1+\frac{1}{|\f r|^{d/2}}\biggr)\,\biggl(1+\frac{1}{|\tilde{\f r}|^{d/2}}\biggr)\,(1+|\f r|+|\tilde{\f r}|)^6\Biggr)\,.
\end{equation}
We hence deduce \eqref{J_approximation_claim} from \eqref{5.25'_error_bound} and \eqref{bar_J_approximation}.
\end{proof}

In the proof of Lemma \ref{Riemann_sum_I}, we use the following estimates that follow from Definition \ref{V_classical} and Lemma \ref{v^{epsilon}_properties} \ref{itm:veps1}. 
\begin{equation}
\label{V_classical_estimates}
\bigl|(\mathbb{V}^{\epsilon})_{x,\tilde x}\bigr| \leq \frac{C}{\epsilon^d}\,,\qquad \bigl|(\mathbb{V}^{\epsilon})_{x}(\omega)\bigr| \leq \frac{C(\tau_1-\tilde \tau_1)}{\epsilon^d}\,.
\end{equation}

\begin{proof}[Proof of Lemma \ref{Riemann_sum_I}]
We recall \eqref{5.9'} and note that for $\f r, \tilde{\f r} \in [0,\infty)^3 \times [0,\infty)^3$, we have
\begin{equation}
\label{I_bound}
|I^{\epsilon}(\f r,\tilde{\f r})| \leq \frac{C}{\epsilon^{4d}} \,\biggl(1+\frac{1}{|\f r|^{d/2}}\biggr)\,\biggl(1+\frac{1}{|\tilde{\f r}|^{d/2}}\biggr)\,(1+|\f r|+|\tilde{\f r}|)^4\,,
\end{equation}
by using \eqref{V_classical_estimates}, \eqref{V_classical_positive}, and Lemma \ref{Heat_kernel_estimates} \ref{itm:heatkernel1}.

Let $C_0>0$ be given. Recalling \eqref{Theta_function}, we first prove the following two estimates.
\begin{equation}
\label{I_estimate_1}
\Biggl|\int_{[0,\infty)^3} \dd \mathbf{r} \,\frac{\ee^{-\kappa |\mathbf{r}|}}{|\mathbf{r}|}\,\int_{[0,\infty)^3} \dd \mathbf{\tilde r}\, \frac{\ee^{-\kappa |\mathbf{\tilde r}|}}{|\mathbf{\tilde r}|}\,I^{\epsilon}(\f r, \tilde{\f r})\,\ind{\min(r_j,\tilde r_j) \leq C_0 \nu}\Biggr| \lesssim_{C_0,\kappa,v} \frac{\sqrt{\nu}}{\epsilon^{4d}}\,.
\end{equation}
\begin{equation}
\label{I_estimate_2}
\Biggl|\nu^6 \sum_{\mathbf{r} \in (\nu \N)^3}\,\frac{\ind{|\mathbf{r}|>0}\,\ee^{-\kappa |\mathbf{r}|}}{|\mathbf{r}|}\, \sum_{\mathbf{\tilde r} \in (\nu \N)^3}\, \frac{\ind{|\mathbf{\tilde r}|>0}\,\ee^{-\kappa |\mathbf{\tilde r}|}}{|\mathbf{\tilde r}|}\,I^{\epsilon}(\f r, \tilde{\f r})\,\ind{\min(r_j,\tilde r_j) \leq C_0 \nu}\Biggr| \lesssim_{C_0,\kappa,v} \frac{\Theta(\nu)}{\epsilon^{4d}}\,.
\end{equation}

By using \eqref{I_bound} and symmetry, we note that \eqref{I_estimate_1} follows from \footnote{We observe that this is a slightly stronger bound than what we need when $d=2$. The bound \eqref{I_estimate_1} is sufficient for the rest of the argument. An analogous observation holds for the bound \eqref{I_estimate_2}.}
\begin{equation}
\label{I_estimate_1_A}
\int_{[0,\infty)^3} \dd \mathbf{r}\,\ind{r_3 \leq C_0\nu}\,\frac{\ee^{-\kappa |\mathbf{r}|/2}}{|\mathbf{r}|^{1+d/2}} \lesssim_{C_0,\kappa} \nu^{2-d/2}\,.
\end{equation}
The estimate \eqref{I_estimate_1_A} follows by using spherical coordinates and considering the contributions $\frac{C_0 \nu}{r} \geq 1$ and $\frac{C_0 \nu}{r} <1$ separately. We omit the details.

By using \eqref{I_bound} and symmetry, we deduce that \eqref{I_estimate_2} follows from 
\begin{equation}
\label{I_estimate_2_A}
\nu^3 \sum_{r_3 \in \nu \N}\,\ind{r_3>0}\,\frac{\ee^{-\kappa r_3}}{r_3^{1+d/2}}
+
\nu^3 \sum_{\mathbf{r} \in (\nu \N)^3}\,\ind{r_3 \leq C_0 \nu}\,\frac{\ind{|(r_1,r_2)|>0}\,\ee^{-\kappa |\mathbf{r}|}}{|\mathbf{r}|^{1+d/2}} \lesssim_{C_0,\kappa} \Theta(\nu)\,.
\end{equation}
The first term on the left-hand side of \eqref{I_estimate_2_A} is
\begin{equation*}
\leq \nu^{3-d/2} \sum_{r_3 \in \nu \N}\,\ind{r_3>0}\,\frac{\ee^{-\kappa r_3}}{r_3} \lesssim_{\kappa} \nu^{2-d/2}\,\log \nu^{-1} \lesssim \Theta(\nu)
\end{equation*}
by considering Riemann sums in one dimension.
The second term on the left-hand side of \eqref{I_estimate_2_A} is
\begin{equation*}
\lesssim_{C_0} \nu^3 \sum_{\mathbf{r} \in (\nu \N)^2}\,\frac{\ind{|\mathbf{r}|>0}\,\ee^{-\kappa |\mathbf{r}|}}{|\mathbf{r}|^{1+d/2}}\lesssim_{\kappa} \sqrt{\nu}\,,
\end{equation*}
where the latter inequality follows by considering Riemann sums in two dimensions. Here, we use the convention that for $\mathbf{r}=(r_1,r_2) \in (\nu \N)^2$, we take $|\mathbf{r}|=r_1+r_2$. We deduce \eqref{I_estimate_2_A} and \eqref{I_estimate_2} then follows.

We henceforth fix $C_0>0$ large and consider $\f r,\tilde{\f r} \in [0,\infty)^3$, $\f \alpha,\tilde{\f \alpha} \in \R^3$ with 
\begin{equation}
\label{assumption_r_alpha}
|r_i|,|\tilde r_i|\geq C_0 \nu\,,\quad |\alpha_i|,|\tilde\alpha_i|\leq \nu\,.
\end{equation}
Under the assumption \eqref{assumption_r_alpha}, we show that
\begin{multline}
\label{I_estimate_3}
|I^{\epsilon}(\f r+\f \alpha,\tilde{\f r}+\tilde{\f \alpha})-I^{\epsilon}(\f r,\tilde{\f r})|=O\Biggl(\frac{\sqrt{\nu}}{\epsilon^{5d+1}} 
\biggl(1+\frac{1}{|\f r|^{d/2}}\biggr)\,\biggl(1+\frac{1}{|\tilde{\f r}|^{d/2}}\biggr)\,(1+|\f r|+|\tilde{\f r}|)^5\Biggr)
\\
+O\Biggl(\frac{\nu}{\epsilon^{4d}}\,\frac{1}{r_3}\,\biggl(1+\frac{1}{(r_1+r_2)^{d/2}}\biggr)\,\biggl(1+\frac{1}{|\tilde{\f r}|^{d/2}}\biggr)\,(1+|\f r|+|\tilde{\f r}|)^4\Biggr)
\\
+O\Biggl(\frac{\nu}{\epsilon^{4d}}\,\biggl(1+\frac{1}{|\f r|^{d/2}}\biggr)\,\frac{1}{{\tilde r_3}}\,\biggl(1+\frac{1}{(\tilde r_1+\tilde r_2)^{d/2}}\biggr)\,(1+|\f r|+|\tilde{\f r}|)^4\Biggr) \deq \cal E_1(\f r,\tilde{\f r})\,.
\end{multline}
By using \eqref{I_bound}, \eqref{assumption_r_alpha}, and the notation in \eqref{I_estimate_3} we have by a direct calculation that
\begin{align}
\notag
\frac{\ee^{-\kappa |\f r+\f \alpha|}}{|\f r+ \f \alpha|}\,\frac{\ee^{-\kappa |\tilde{\f r}+\tilde{\f \alpha}|}}{|\tilde{\f r}+ \tilde{\f \alpha}|}\,&I^{\epsilon}(\f r+\f \alpha,\tilde{\f r}+\tilde{\f \alpha})-\frac{\ee^{-\kappa |\f r|}}{|\f r|}\,\frac{\ee^{-\kappa |\tilde{\f r}|}}{|\tilde{\f r}|}\,I^{\epsilon}(\f r,\tilde{\f r})
\\
\notag
&=O\Biggl(\frac{\ee^{-\kappa |\f r|}}{|\f r|}\,\frac{\ee^{-\kappa |\tilde{\f r}|}}{|\tilde{\f r}|}\,
\cal E_1(\f r,\tilde{\f r})\Biggr)
\\
\notag
&+O_{\kappa}\Biggl(\frac{\nu}{\epsilon^{4d}} \biggl(1+\frac{1}{|\f r|^{(d+3)/2}}\biggr)\,\biggl(1+\frac{1}{|\tilde{\f r}|^{(d+3)/2}}\biggr)\,(1+|\f r|+|\tilde{\f r}|)^4\,\,\ee^{-\kappa (|\f r| +|\tilde{\f r}|)}\Biggr)
\\
\label{I_estimate_3B}
& \deq \cal E_2(\f r,\tilde{\f r})\,.
\end{align}
In order to obtain the second error term in \eqref{I_estimate_3B}, we note that, by \eqref{assumption_r_alpha}, we have\footnote{In order to obtain the second error term in \eqref{I_estimate_3B*}, we interpolate between the estimates
$|\frac{1}{|\hat{\f r}+ \hat{\f \alpha}|}-\frac{1}{|\hat{\f r}|}| \lesssim \frac{1}{|\hat{\f r}|}$ and $|\frac{1}{|\hat{\f r}+ \hat{\f \alpha}|}-\frac{1}{|\hat{\f r}|}| \lesssim \frac{\nu}{|\hat{\f r}|^2}$.
}
\begin{equation}
\label{I_estimate_3B*}
\frac{\ee^{-\kappa |\hat{\f r}+\hat{\f \alpha}|}}{|\hat{\f r}+ \hat{\f \alpha}|}-
\frac{\ee^{-\kappa |\hat{\f r}|}}{|\hat{\f r}|}=O_{\kappa} \Biggl(\frac{\ee^{-\kappa |\hat{\f r}|}\,\nu}{|\hat{\f r}|}+\frac{\ee^{-\kappa |\hat{\f r}|}\,\nu^{1/2}}{|\hat{\f r}|^{3/2}}\Biggr)\,.
\end{equation}

We then deduce the lemma by using \eqref{I_estimate_1}, \eqref{I_estimate_2}, \eqref{assumption_r_alpha}, \eqref{I_estimate_3}, \eqref{I_estimate_3B}, and considering Riemann sums for \eqref{5.8'}. Indeed, for all $(\f s,\tilde{\f s}) \in [0,\infty)^3 \times [0,\infty)^3$, we take $(\f r,\tilde{\f r}) \in (\nu \N)^3 \times (\nu \N)^3$ such that $\hat r_j = \lfloor \hat s_j \rfloor_{\nu}$, where
\begin{equation}
\label{floor_nu}
\lfloor s \rfloor_{\nu} \deq \mathrm{max}\,\{u \in \nu \N\,,\,u \leq s\}\,.
\end{equation}
Then, we automatically have $|\hat r_j-\hat s_j| \leq \nu$ for all $j=1,2,3$. 
We then use \eqref{I_estimate_3}, \eqref{I_estimate_3B} with $\hat{\f r}+\hat{\f \alpha} \equiv \hat{\f s}$ and we reduce the claim to showing
\begin{equation*}
\nu^6 \sum_{\f r,\tilde{\f r} \in (\nu \N)^3 \cap [C_0\nu,\infty)^3} \cal E_2(\f r,\tilde{\f r}) \lesssim \Theta(\nu)\,,
\end{equation*}
which follows from \eqref{I_estimate_3}--\eqref{I_estimate_3B}.
Let us note that when estimating the contributions from the last two error terms in \eqref{I_estimate_3}, we use
\begin{align}
\label{I_estimate_3'_A}
&\sum_{r_3 \in (\nu \N) \cap [C_0\nu,\infty)} \frac{\nu}{r_3}\,\ee^{-\kappa r_3/2} \lesssim_{\kappa} \log \nu^{-1}
\\
\label{I_estimate_3'}
&\sum_{\f r \in (\nu \N)^2 \cap [C_0\nu,\infty)^2} \nu^2 \,\biggl(1+\frac{1}{|\f r|^{1+d/2}}\biggr) \,\ee^{-\kappa |\f r|/2}\lesssim_{C_0,\kappa} \log \nu^{-1}\, \ind{d=2}+ \nu^{-1/2}\, \ind{d=3}\,,
\end{align}
which follow by considering Riemann sums in one and two dimensions respectively. Here, the term on the left-hand side of \eqref{I_estimate_3'} comes from estimating
$\frac{1}{|(r_1,r_2,r_3)|}(1+\frac{1}{(r_1+r_2)^{d/2}}) \lesssim 1+\frac{1}{(r_1+r_2)^{1+d/2}}$.
Finally, let us note that when $d=3$, we have $(d+3)/2=3$ and therefore the second error term in \eqref{I_estimate_3B} yields the logarithmic factor in the error term. Similarly, in light of \eqref{I_estimate_3'_A}--\eqref{I_estimate_3'}, the same is true for the last two error terms in \eqref{I_estimate_3}. This requires the necessary modification in \eqref{Theta_function} when $d=3$.

The rest of the proof is devoted to showing \eqref{I_estimate_3}. Let $\f r, \tilde{\f r}, \f \alpha, \tilde{\f \alpha}$ be as in \eqref{assumption_r_alpha}.
Recalling \eqref{5.9'} and using an appropriate time translation of the paths\footnote{The time-translation is analogous to that used to in order to rewrite \eqref{5.21'} as \eqref{5.21*} above.}, we can write

\begin{align}
\notag
I^{\epsilon}&(\f r+\f\alpha, \tilde{\f r}+\tilde{\f \alpha}) 
\\
\notag
&= \int \dd \mathbf{x} \, \int \dd \mathbf{\tilde x}\,\int \mathbb{W}^{|\f r|+|\f \alpha|,r_1+r_2+\alpha_1+\alpha_2}_{x_1,x_3}(\dd \omega_3)\,\mathbb{W}^{r_1+r_2+\alpha_1+\alpha_2,r_1+\alpha_1}_{x_3,x_2}(\dd \omega_2)\,\mathbb{W}^{r_1+\alpha_1,0}_{x_2,x_1}(\dd \omega_1)\,
\\
\notag
&\times
\mathbb{W}^{|\tilde{\f r}|+|\tilde{\f \alpha}|,\tilde r_1+\tilde r_2+\tilde \alpha_1+\tilde \alpha_2}_{\tilde x_1,\tilde x_3}(\dd \tilde{\omega}_3)\,\mathbb{W}^{\tilde r_1+\tilde r_2+\tilde \alpha_1+\tilde \alpha_2,\tilde r_1+\tilde \alpha_1}_{\tilde x_3,\tilde x_2}(\dd \tilde{\omega}_2)\,\mathbb{W}^{\tilde r_1+\tilde \alpha_1,0}_{\tilde x_2,\tilde x_1}(\dd \tilde{\omega}_1)\,
\\
\label{I^{epsilon}_translated}
&\times \ee^{-\frac{1}{2}\bigl(\mathbb{V}^{\epsilon}(\omega_1,\omega_1)+\mathbb{V}^{\epsilon}(\tilde \omega_1,\tilde \omega_1)-2\mathbb{V}^{\epsilon}(\omega_1,\tilde \omega_1)\bigr)}\,\sum_{\Pi \in \fra{M}(A)} \prod_{\{a,b\} \in \Pi} (\mathbb{V}^{\epsilon})_{x_a,x_b}\,\prod_{a \in A \setminus [\Pi]} \ii \Bigl((\mathbb{V}^{\epsilon})_{x_a}(\omega_1)-
(\mathbb{V}^{\epsilon})_{x_a}(\tilde{\omega}_1)\Bigr)\,.
\end{align}
Note that, by \eqref{assumption_r_alpha}, we indeed have that $|\tilde{\f r}|+|\tilde{\f \alpha}|>\tilde r_1+\tilde r_2+\tilde \alpha_1+\tilde \alpha_2$ and $|\tilde{\f r}|+|\tilde{\f \alpha}|>\tilde r_1+\tilde r_2+\tilde \alpha_1+\tilde \alpha_2$, hence the above expression is well-defined.

We now show that 
\begin{multline}
\label{I^{epsilon}_translated_2}
\eqref{I^{epsilon}_translated}= 
\bar{I}^{\epsilon}(\f r+\f\alpha, \tilde{\f r}+\tilde{\f \alpha})+O\Biggl(\frac{\nu}{\epsilon^{4d}}\,\frac{1}{r_3}\,\biggl(1+\frac{1}{(r_1+r_2)^{d/2}}\biggr)\,\biggl(1+\frac{1}{|\tilde {\f r}|^{d/2}}\biggr)\,(1+|\f r|+|\tilde{\f r}|)^4\Biggr)
\\
+O\Biggl(\frac{\nu}{\epsilon^{4d}}\,\biggl(1+\frac{1}{|\f r|^{d/2}}\biggr)\,\frac{1}{\tilde r_3}\,\biggl(1+\frac{1}{(\tilde r_1+\tilde r_2)^{d/2}}\biggr)\,(1+|\f r|+|\tilde{\f r}|)^4\Biggr)\,,
\end{multline}
where
\begin{multline}
\label{I^{epsilon}_translated_bar}
\bar{I}^{\epsilon}(\f r+\f\alpha, \tilde{\f r}+\tilde{\f \alpha})\deq
\int \dd \mathbf{x} \, \int \dd \mathbf{\tilde x}\,\int \mathbb{W}^{|\f r|,r_1+r_2+\alpha_1+\alpha_2}_{x_1,x_3}(\dd \omega_3)\,\mathbb{W}^{r_1+r_2+\alpha_1+\alpha_2,r_1+\alpha_1}_{x_3,x_2}(\dd \omega_2)\,\mathbb{W}^{r_1+\alpha_1,0}_{x_2,x_1}(\dd \omega_1)\,
\\
\times
\mathbb{W}^{|\tilde{\f r}|,\tilde r_1+\tilde r_2+\tilde \alpha_1+\tilde \alpha_2}_{\tilde x_1,\tilde x_3}(\dd \tilde{\omega}_3)\,\mathbb{W}^{\tilde r_1+\tilde r_2+\tilde \alpha_1+\tilde \alpha_2,\tilde r_1+\tilde \alpha_1}_{\tilde x_3,\tilde x_2}(\dd \tilde{\omega}_2)\,\mathbb{W}^{\tilde r_1+\tilde \alpha_1,0}_{\tilde x_2,\tilde x_1}(\dd \tilde{\omega}_1)\,
\\
\times \ee^{-\frac{1}{2}\bigl(\mathbb{V}^{\epsilon}(\omega_1,\omega_1)+\mathbb{V}^{\epsilon}(\tilde \omega_1,\tilde \omega_1)-2\mathbb{V}^{\epsilon}(\omega_1,\tilde \omega_1)\bigr)}\,\sum_{\Pi \in \fra{M}(A)} \prod_{\{a,b\} \in \Pi} (\mathbb{V}^{\epsilon})_{x_a,x_b}\,\prod_{a \in A \setminus [\Pi]} \ii \Bigl((\mathbb{V}^{\epsilon})_{x_a}(\omega_1)-
(\mathbb{V}^{\epsilon})_{x_a}(\tilde{\omega}_1)\Bigr)\,
\end{multline}
is obtained by modifying \eqref{I^{epsilon}_translated} to replace the final time of $\omega_3$ and $\tilde \omega_3$ by $|\f r|$ and $|\tilde{\f r}|$ respectively\footnote{$\bar{I}^{\epsilon}(\f r+\f\alpha, \tilde{\f r}+\tilde{\f \alpha})$ can easily be written as a function of $(\f r,\tilde{\f r})$. This notation is more convenient for the purposes of our argument.}.

In order to deduce the estimate \eqref{I^{epsilon}_translated_2}, let us first consider the case $\tilde{\f \alpha}=\f 0$. We note that neither of the integrands in \eqref{I^{epsilon}_translated}, \eqref{I^{epsilon}_translated_bar} depend on $\omega_3$.
Using \eqref{V_classical_positive}, \eqref{V_classical_estimates}, integrating in $x_2,\tilde x_1,\tilde x_2, \tilde x_3$ and recalling Lemma \ref{Heat_kernel_estimates} \ref{itm:heatkernel1} as well as \eqref{assumption_r_alpha}, we deduce that 
\begin{align*}
|\bar{I}^{\epsilon}(\f r+\f\alpha, \tilde{\f r})-I^{\epsilon}(\f r+\f\alpha, \tilde{\f r})| \lesssim_{d,v} \frac{1}{\epsilon^{4d}}\,&\int \dd x_1\,\dd x_3 \,\bigl|\psi^{r_3+\alpha_3-\alpha_1-\alpha_2}(x_1-x_3)-\psi^{r_3-\alpha_1-\alpha_2}(x_1-x_3)\bigr|\,
\\
&\times\biggl(1+\frac{1}{(r_1+r_2)^{d/2}}\biggr)\,\biggl(1+\frac{1}{|\tilde {\f r}|^{d/2}}\biggr)\,(1+|\f r|+|\tilde {\f r}|)^4
\\
&\lesssim_{d,v} \frac{|\alpha_3|}{\epsilon^{4d}}\,\frac{1}{r_3}\,\biggl(1+\frac{1}{(r_1+r_2)^{d/2}}\biggr)\,\biggl(1+\frac{1}{|\tilde {\f r}|^{d/2}}\biggr)\,(1+|\f r|+|\tilde{\f r}|)^4\,.
\end{align*}
For the last inequality, we used Lemma \ref{Heat_kernel_estimates} \ref{itm:heatkernel3} and \eqref{assumption_r_alpha}. We hence deduce \eqref{I^{epsilon}_translated_2} when $\tilde{\f \alpha}=\f 0$. The general claim of \eqref{I^{epsilon}_translated_2} follows by a telescoping argument.

We now show that $|\bar{I}^{\epsilon}(\f r+\f \alpha,\tilde{\f r}+\tilde{\f \alpha})-I^{\epsilon}(\f r,\tilde{\f r})|$ is bounded by the first error term on the right-hand side of \eqref{I_estimate_3}. In order to do this, we rewrite $\bar{I}^{\epsilon}(\f r+\f \alpha,\tilde{\f r}+\tilde{\f \alpha})$
by arguing as in the proof of Lemma \ref{J_approximation}.
Similarly as in \eqref{omega_hat}, consider loops 
$\hat \omega \col [0,|\f {\hat{r}}|] \rightarrow \Lambda$ which are obtained by concatenating $\hat \omega_1, \hat \omega_2, \hat \omega_3$. In particular, 
\begin{equation}
\label{hat_omega_2}
\hat \omega_1=\hat \omega \big|_{[0,\hat r_1+\hat \alpha_1]}\,,\quad \hat \omega_2=\hat \omega \big|_{[\hat r_1+\hat \alpha_1,\hat r_1+\hat r_2+\hat \alpha_1+\hat \alpha_2]}\,,\quad \hat \omega_3 =\hat \omega \big|_{[\hat r_1+\hat r_2+\hat \alpha_1+\hat \alpha_2,|\hat{\f r}|+|\hat{\f \alpha}|]}\,.
\end{equation}
Given $a \in A$, we define the following quantities, similarly to \eqref{parameter_choice_quantum}.
\begin{equation}
\label{parameter_choice}
\begin{cases}
&t_{(2,0)} \deq r_1+\alpha_1\,,\;\; t_{(3,0)}\deq r_1+r_2+\alpha_1+\alpha_2\,,\;\; t_{(2,1)} \deq \tilde r_1+\tilde \alpha_1\,,\;\;
t_{(3,1)} \deq \tilde r_1 +\tilde r_2 + \tilde \alpha_1+ \tilde \alpha_2\,,
\\
&s_{(2,0)} \deq r_1,\,\;\; s_{(3,0)} \deq r_1+r_2\,,\;\; s_{(2,1)} \deq \tilde r_1\,,\;\; s_{(3,1)} \deq \tilde r_1+\tilde r_2
\\
&\omega_{(2,0)}=\omega_{(3,0)}  \deq \omega \,,\;\; \omega_{(2,1)}=\omega_{(3,1)}  \deq \tilde{\omega}\,.
\end{cases}
\end{equation}
Note that, in \eqref{parameter_choice}, the $s_a$ are the approximations of the $t_a$ which have been decoupled from the variables $\hat \alpha_j$.
By construction, we have that $\omega_a(t_a)=x_a$ and $|s_a-t_a|\lesssim \nu$ for all $a \in A$.

Therefore, integrating in $\hat x_2, \hat x_3$, we get
\begin{multline}
\label{I_bar}
\bar{I}^{\epsilon}(\f r+\f \alpha,\tilde{\f r}+\tilde{\f \alpha})= \int \dd x_1\,\int \dd \tilde x_1\,\int \mathbb{W}^{|\mathbf{r}|,0}_{x_1,x_1}(\dd \omega)\,\mathbb{W}^{|\mathbf{\tilde r}|,0}_{\tilde x_1,\tilde x_1}(\dd \tilde \omega)\,\ee^{-\frac{1}{2}\bigl(\mathbb{V}^{\epsilon}(\omega_1,\omega_1)+\mathbb{V}^{\epsilon}(\tilde \omega_1,\tilde \omega_1)-2\mathbb{V}^{\epsilon}(\omega_1,\tilde \omega_1)\bigr)}\,
\\
\times
\sum_{\Pi \in \fra{M}(A)} \prod_{\{a,b\} \in \Pi} (\mathbb{V}^{\epsilon})_{\omega_a(t_a),\omega_b(t_b)}\,\prod_{a \in A \setminus [\Pi]} \ii \Bigl((\mathbb{V}^{\epsilon})_{\omega_a(t_a)}(\omega_1)-
(\mathbb{V}^{\epsilon})_{\omega_a(t_a)}(\tilde{\omega}_1)\Bigr)\,.
\end{multline}
We set $\f \alpha=\tilde{\f \alpha}=\f 0$ in \eqref{I^{epsilon}_translated} and recall \eqref{hat_omega_2}--\eqref{parameter_choice} to write
\begin{multline}
\label{I_comparison}
I^{\epsilon}(\f r,\tilde{\f r})= \int \dd x_1\,\int \dd \tilde x_1\,\int \mathbb{W}^{|\mathbf{r}|,0}_{x_1,x_1}(\dd \omega)\,\mathbb{W}^{|\mathbf{\tilde r}|,0}_{\tilde x_1,\tilde x_1}(\dd \tilde \omega)\,\ee^{-\frac{1}{2}\bigl(\mathbb{V}^{\epsilon}(\omega_1,\omega_1)+\mathbb{V}^{\epsilon}(\tilde \omega_1,\tilde \omega_1)-2\mathbb{V}^{\epsilon}(\omega_1,\tilde \omega_1)\bigr)}\,
\\
\times
\sum_{\Pi \in \fra{M}(A)} \prod_{\{a,b\} \in \Pi} (\mathbb{V}^{\epsilon})_{\omega_a(s_a),\omega_b(s_b)}\,\prod_{a \in A \setminus [\Pi]} \ii \Bigl((\mathbb{V}^{\epsilon})_{\omega_a(s_a)}(\omega_1)-
(\mathbb{V}^{\epsilon})_{\omega_a(s_a)}(\tilde{\omega}_1)\Bigr)\,.
\end{multline}
In order to compare \eqref{I_bar} and \eqref{I_comparison}, we note the following three estimates, which hold uniformly in $\eta>0$.
\begin{enumerate}[label=(\roman*)]
\item \label{itm:V1}
We have
\begin{multline*}
\ee^{-\frac{1}{2}\bigl(\mathbb{V}^{\epsilon}(\omega_1,\omega_1)+\mathbb{V}^{\epsilon}(\tilde \omega_1,\tilde \omega_1)-2\mathbb{V}^{\epsilon}(\omega_1,\tilde \omega_1)\bigr)}
\\
=\ee^{-\frac{1}{2}\bigl(\mathbb{V}^{\epsilon}\bigl(\omega|_{[0,r_1]},\omega_{[0,r_1]}\bigr)+\mathbb{V}^{\epsilon}\bigl(\tilde \omega_{[0,\tilde r_1]},\tilde \omega_{[0,\tilde r_1]}\bigr)-2\mathbb{V}^{\epsilon}\bigl(\omega_{[0,r_1]},\tilde \omega_{[0,\tilde r_1]}\bigr)\bigr)}+O\biggl(\frac{\nu}{\epsilon^d}(1+r_1+\tilde r_1)\biggr)\,.
\end{multline*}
\item \label{itm:V2}
For all $a,b \in A$, we have
\begin{equation*}
\bigl|(\mathbb{V}^{\epsilon})_{\omega_a(t_a),\omega_b(t_b)}-(\mathbb{V}^{\epsilon})_{\omega_a(s_a),\omega_b(s_b)}\bigr| \leq \frac{C}{\epsilon^{d+1}} \Bigl(|\omega_a(t_a)-\omega_a(s_a)|_{\Lambda}+|\omega_b(t_b)-\omega_b(s_b)|_{\Lambda}\Bigr)\,.
\end{equation*}
\item \label{itm:V3}
For all $a \in A$, we have
\begin{equation*}
\bigl|(\mathbb{V}^{\epsilon})_{\omega_a(t_a)}(\hat \omega_1)-(\mathbb{V}^{\epsilon})_{\omega_a(s_a)}(\hat \omega_1)\bigr| \leq \frac{C}{\epsilon^{d+1}}\,\hat r_1\,|\omega_a(t_a)-\omega_a(s_a)|_{\Lambda}\,.
\end{equation*}
\end{enumerate}
Claim \ref{itm:V1} follows by arguing as for \eqref{claim_(i)}.
Claim \ref{itm:V2} follows by using Definition \ref{V_classical} and Lemma \ref{v^{epsilon}_properties} \ref{itm:veps2}. Claim \ref{itm:V3} is shown analogously.

Starting from the identities \eqref{I_bar}, \eqref{I_comparison}, using \eqref{V_classical_estimates}, estimates \ref{itm:V1}--\ref{itm:V3} above, and recalling Lemma \ref{Heat_kernel_estimates} \ref{itm:heatkernel1}--\ref{itm:heatkernel2}, we deduce that
\begin{equation}
\label{I_estimate_4}
|\bar{I}^{\epsilon}(\f r+\f \alpha,\tilde{\f r}+\tilde{\f \alpha})-I^{\epsilon}(\f r,\tilde{\f r})|=O\Biggl(\frac{\sqrt{\nu}}{\epsilon^{5d+1}} 
\biggl(1+\frac{1}{|\f r|^{d/2}}\biggr)\,\biggl(1+\frac{1}{|\tilde{\f r}|^{d/2}}\biggr)\,(1+|\f r|+|\tilde{\f r}|)^5\Biggr)\,.
\end{equation}
Note that here we again used estimate \eqref{Heat_kernel_estimates_ii_2} when applying Lemma \ref{Heat_kernel_estimates} \ref{itm:heatkernel2}.
Combining \eqref{I^{epsilon}_translated}--\eqref{I^{epsilon}_translated_2}, and \eqref{I_estimate_4}, we obtain \eqref{I_estimate_3} and the lemma follows.
\end{proof}

In order to prove Lemma \ref{Partition_function_rate_of_convergence_5}, we need to make some minor modifications.
Arguing analogously as for \cite[(5.31)--(5.33)]{FKSS_2020}, we have that
\begin{equation} 
\label{5.31'}
\int \mu_{\cal{C}_{\eta}}(\dd \sigma) \, \overline{f_2(\langle \sigma \rangle)} F_2(\sigma)
=\frac{1}{\nu^3}\,
\sum_{\mathbf{r} \in (\nu \N)^3} \frac{\ind{|\mathbf{r}|>0}\,\ee^{-\kappa|\mathbf{r}|}}{|\mathbf{r}|}
\int_{[0,\infty)^3} \dd \tilde{\mathbf{r}} \,\frac{\ee^{-\kappa|\tilde{\mathbf{r}}|}}{|\tilde{\mathbf{r}}|}\,
\tilde{J}^{\epsilon}(\mathbf{r},\tilde{\mathbf{r}})\,,
\end{equation}
where 
\begin{multline}
\label{5.33'}
\tilde{J}^{\epsilon}(\mathbf{r},\tilde{\mathbf{r}}) \deq \int_{[0,\nu]^3} \dd \f \tau \,\int_{[0,\nu]^3} \dd \tilde{\f \tau}\, \int \dd \mathbf{x}\,\int\dd \tilde{\mathbf{x}}\,\int \mathbb{W}^{\tau_1+r_3,\tau_3}_{x_1,x_3}(\dd \omega_3)\,\mathbb{W}^{\tau_3+r_2,\tau_2}_{x_3,x_2}(\dd \omega_2)\,\mathbb{W}^{\tau_2+r_1,\tau_1}_{x_2,x_1}(\dd \omega_1)
\\
\times
\mathbb{W}^{\tilde{r}_3,0}_{\tilde{x}_1,\tilde{x}_3}(\dd \tilde{\omega}_3)\,\mathbb{W}^{\tilde{r}_2,0}_{\tilde{x}_3,\tilde{x}_2}(\dd \tilde{\omega}_2)\,\mathbb{W}^{\tilde{r}_1,0}_{\tilde{x}_2,\tilde{x}_1}(\dd \tilde{\omega}_1)\,
\ee^{-\frac{1}{2}\bigl(\mathbb{V}_{\eta}(\omega_1,\omega_1)+\mathbb{V}^{\epsilon}(\tilde \omega_1,\tilde \omega_1)-2\mathbb{V}^{\epsilon}(\omega_1,\tilde \omega_1)\bigr)}\,
\\
\times
\sum_{\Pi \in \fra{M}(A)} \prod_{\{a,b\} \in \Pi} (\mathbb{V}_{\eta})_{x_a,x_b}^{\tau_a,\tau_b}\,\prod_{a \in A \setminus [\Pi]} \ii \Bigl((\mathbb{V}_{\eta})_{x_a}^{\tau_a}(\omega_1)-
(\mathbb{V}^{\epsilon})_{x_a}(\tilde{\omega}_1)\Bigr)\,.
\end{multline}

The following analogues of Lemmas \ref{J_approximation} and \ref{Riemann_sum_I} hold.
\begin{lemma}
\label{J_approximation2}
For all $\f r \in (\nu \N)^3$ and $\tilde{\mathbf{r}} \in [0,\infty)^3$ with $|\f r|, |\tilde{\mathbf{r}}|>0$, we have that uniformly in $\eta>0$
\begin{equation*}
J(\mathbf{r},\tilde{\mathbf{r}})=\nu^6 I^{\epsilon}(\mathbf{r},\tilde{\mathbf{r}})+O\Biggl(\frac{\nu^{13/2}}{\epsilon^{5d+1}}\,\biggl(1+\frac{1}{|\f r|^{d/2}}\biggr)\,\biggl(1+\frac{1}{|\tilde{\f r}|^{d/2}}\biggr)\,(1+|\f r|+|\tilde{\f r}|)^6\Biggr)\,.
\end{equation*}
\end{lemma}

\begin{lemma}
\label{Riemann_sum_I2}
Recalling \eqref{Theta_function}, we have that uniformly in $\eta>0$
\begin{multline*}
\Biggl|\int_{[0,\infty)^3} \dd \mathbf{r} \,\frac{\ee^{-\kappa |\mathbf{r}|}}{|\mathbf{r}|}\,\int_{[0,\infty)^3} \dd \mathbf{\tilde r}\, \frac{\ee^{-\kappa |\mathbf{\tilde r}|}}{|\mathbf{\tilde r}|}\,I^{\epsilon}(\f r, \tilde{\f r})-\nu^3 \sum_{\mathbf{r} \in (\nu \N)^3}\,\frac{\ind{|\mathbf{r}|>0}\,\ee^{-\kappa |\mathbf{r}|}}{|\mathbf{r}|}\, \int_{[0,\infty)^3}\dd \mathbf{\tilde r}\, \frac{\ee^{-\kappa |\mathbf{\tilde r}|}}{|\mathbf{\tilde r}|}\,I^{\epsilon}(\f r, \tilde{\f r})\Biggr| 
\\
\lesssim_{\kappa,v} \frac{\Theta(\nu)}{\epsilon^{5d+1}}\,.
\end{multline*}
\end{lemma}
The proofs of Lemmas \ref{J_approximation2} and \ref{Riemann_sum_I2} are very similar to those of Lemmas \ref{J_approximation} and \ref{Riemann_sum_I}; see Appendix \ref{Appendix C}.
\begin{proof}[Proof of Lemma \ref{J_approximation2}]
The proof is analogous to that of Lemma \ref{J_approximation} given above. The only difference is that for $\hat \omega_1=\tilde \omega_1$, we replace \eqref{claim_(iii)} by the estimate
\begin{equation*}
(\mathbb{V}^{\epsilon})_{\omega_a(t_a)}(\tilde{\omega}_1)=(\mathbb{V}^{\epsilon})_{\omega_a(s_a)}(\tilde{\omega}|_{[\tilde \tau_1,\tilde \tau_1+\tilde r_1]})+O\biggl(\frac{\tilde{r}_1}{\epsilon^{d+1}}\,\,|\omega_a(t_a)-\omega_a(s_a)|_{\Lambda}\biggr)+O\biggl(\frac{\nu}{\epsilon^d}\biggr)\,,
\end{equation*}
which follows from Lemma \ref{v^{epsilon}_properties} \ref{itm:veps1}--\ref{itm:veps2} and \eqref{omega_hat}--\eqref{parameter_choice_quantum}.
\end{proof}

\begin{proof}[Proof of Lemma \ref{Riemann_sum_I2}]
The result follows directly from the proof of Lemma \ref{Riemann_sum_I}.
\end{proof}

We can now prove Lemma \ref{Partition_function_rate_of_convergence_5}.
\begin{proof}[Proof of Lemma \ref{Partition_function_rate_of_convergence_5}]
The claim follows from Lemmas \ref{J_approximation2} and \ref{Riemann_sum_I2} by using \eqref{5.8'} and \eqref{5.31'}.
\end{proof}

\subsection{Correlation functions}

We note the following analogues of Lemma \ref{Partition_function_rate_of_convergence_1} and \ref{Partition_function_rate_of_convergence_2} that give us functional integral representations for \eqref{hat_Gamma_p_2} and \eqref{hat_gamma_p_2} respectively.

\begin{lemma}
\label{hat_Gamma_fr}
For all $p \in \N$, we have that as $\eta \rightarrow 0$, $\wh{\Gamma}_{p,\eta}\overset{C}{\longrightarrow} \wh{\Gamma}_{p}$, 
where $\wh{\Gamma}_{p}$ satisfies
\begin{equation}
\label{Gamma_hat_quantum}
\nu^p \, \wh{\Gamma}_{p,\eta}=\frac{p!}{\cal Z_{\eta}^{\epsilon}}P_p \,Q_{p,\eta}
\end{equation}
for 
\begin{equation}
\label{Q_quantum}
(Q_{p,\eta})_{\f x,\tilde{\f x}} \deq \int \mu_{\cal C_{\eta}} (\dd \sigma)\, \ee^{T^{\epsilon}-\frac{\ii \tau^\epsilon [\sigma]}{\nu}}\,\ee^{F_2(\sigma)}\,
\prod_{i=1}^{p} \Biggl[\nu\sum_{r_i \in \nu \N^*} \ee^{-\kappa r_i} \, \int \mathbb{W}^{r_i,0}_{x_i,\tilde{x}_i}(\dd \omega_i)\, \Bigl(\ee^{\ii \int_{0}^{r_i} 
\dd t\, \sigma([t]_{\nu},\omega_i(t))}-1\Bigr)\Biggr]\,.
\end{equation}
Here, $F_2$ is given as in \eqref{F_2}. Moreover, we recall that $P_p$ denotes the projection given by \eqref{def_Pn} and 
$\overset{C}{\longrightarrow}$ convergence in the space of continuous functions on $\Lambda^p \times \Lambda^p$ with respect to the supremum norm.
\end{lemma}
\begin{proof}
We obtain \eqref{Gamma_hat_quantum}--\eqref{Q_quantum} by arguing analogously as for 
\cite[(5.36)--(5.37)]{FKSS_2020}. The only difference is that, in the $\sigma$, integral, we have to add the extra factor of 
$\ee^{T^{\epsilon}-\frac{\ii \tau^\epsilon [\sigma]}{\nu}}$ due to the $\epsilon$-dependent corrections in \eqref{Hamiltonian_H}. This change is justified by arguing analogously as in the proof of Proposition \ref{Partition_function_rate_of_convergence_1} above.
\end{proof}

\begin{lemma}
\label{hat_gamma_fr}
For all $p \in \N$, we have that 
\begin{equation}
\label{gamma_hat_classical}
\wh{\gamma}_{p}^{W^\epsilon}=\frac{p!}{\zeta^{W^\epsilon}}P_p \,q_{p}^{\epsilon}
\end{equation}
for 
\begin{equation}
\label{q_classical}
(q_{p}^{\epsilon})_{\f x,\tilde{\f x}} \deq \int \mu_{v^{\epsilon}} (\dd \xi)\, \ee^{T^{\epsilon}-\ii \tau^\epsilon \langle \xi,1 \rangle_{L^2}}\,\ee^{f_2(\xi)}\,
\prod_{i=1}^{p} \Biggl[\int_{0}^{\infty}  \dd r_i\,\ee^{-\kappa r_i} \, \int \mathbb{W}^{r_i,0}_{x_i,\tilde{x}_i}(\dd \omega_i)\, \Bigl(\ee^{\ii \int_{0}^{r_i} 
\dd t\, \xi(\omega_i(t))}-1\Bigr)\Biggr]\,.
\end{equation}
Here, $f_2$ is given as in \eqref{f_2}.
\end{lemma}
\begin{proof}
We obtain \eqref{gamma_hat_classical}--\eqref{q_classical} analogously as for \cite[(5.38)]{FKSS_2020}. 
The only difference is that, in the $\xi$ integral, we have to add the extra factor of $\ee^{T^{\epsilon}-\ii \tau^\epsilon \langle \xi,1 \rangle_{L^2}}$ due to the 
second and third term in \eqref{W^epsilon}.
This change is justified by arguing analogously as in the proof of Proposition \ref{Partition_function_rate_of_convergence_2} above. Moreover, since $v^{\epsilon}$ is already assumed to be smooth, there is no need to regularize it with the parameter $\eta$, which was used in \cite{FKSS_2020}.
\end{proof}


Recalling \eqref{sigma(epsilon)}, we now compare the quantities \eqref{Q_quantum} and \eqref{q_classical}.
\begin{lemma}
\label{Q_and_q_estimates}
There exists $c_3>0$ depending on $\kappa,v$ such that the following results hold for all $p \in \N$.
\begin{enumerate}[label=(\roman*)]
\item  \label{itm:Q1}
We have
\begin{equation*}
\|Q_{p,\eta}\|_{C} \leq C_{\kappa,v}^p \, \ee^{c_3 \chi(\epsilon)^2}\, \biggl(\frac{1}{\epsilon^d}\biggr)^{p/2}\,.
\end{equation*}
\item \label{itm:Q2}
For $\theta(d,p)$ as in \eqref{theta_definition}, we have uniformly in $\eta>0$
\begin{equation}
\label{Q_and_q_estimates_ii}
\|Q_{p,\eta}-q^{\epsilon}_{p}\|_{C} \leq C_{\kappa,v}^p \, \ee^{c_3 \chi(\epsilon)^2}\, \Biggl(\biggl(\frac{1}{\epsilon^d}\biggr)^{p/2}+\biggl(\frac{1}{\epsilon^{5d+1}}\biggr)^{1/2}\Biggr)\,\nu^{\theta(d,p)}\,.
\end{equation}
\end{enumerate}
\end{lemma}
Furthermore, we show the following lower bound on the classical and quantum relative partition functions.
\begin{lemma}[Lower bound on the relative partition function]
\label{lower_bound}
The following estimates hold for some constant $c_4>0$ depending on $\kappa$.
\begin{enumerate}[label=(\roman*)]
\item \label{itm:lb1}
$\zeta^{W^{\epsilon}} \geq \exp[-c_4 \chi(\epsilon)^2]$.
\item \label{itm:lb2}
$\cal Z \geq \exp[-c_4 \chi(\epsilon)^2]$.
\end{enumerate}
\end{lemma}
We prove Lemmas \ref{Q_and_q_estimates} and \ref{lower_bound} in Appendix \ref{Appendix C}.
Using Lemmas \ref{Q_and_q_estimates} and \ref{lower_bound}, we now prove Proposition \ref{Correlation_functions_rate_of_convergence}.

\begin{proof}[Proof of Proposition \ref{Correlation_functions_rate_of_convergence}]

By Lemmas \ref{Partition_function_rate_of_convergence_1}, \ref{hat_Gamma_fr}, and \ref{hat_gamma_fr} it suffices to estimate
\begin{equation}
\label{ratio_estimate}
\lim_{\eta \rightarrow 0}\, \biggl\|\frac{Q_{p,\eta}}{\cal Z_{\eta}}-\frac{q^{\epsilon}_{p}}{\zeta^{W^\epsilon}}\biggr\|_{C} 
\leq \frac{\bigl|\cal Z-\zeta^{W^\epsilon}\bigr|}{\cal Z\,\zeta^{W^\epsilon}}\,\lim_{\eta \rightarrow 0}\,\|Q_{p,\eta}\|_{C}+\frac{1}{\zeta^{W^\epsilon}}\,\lim_{\eta \rightarrow 0}\,\|Q_{p,\eta}-q^{\epsilon}_{p}\|_{C}\,.
\end{equation}
The claim now follows from \eqref{ratio_estimate} by using Proposition \ref{Partition_function_rate_of_convergence}, Lemma \ref{Q_and_q_estimates} \ref{itm:Q1} and Lemma \ref{lower_bound} to estimate the first term and Lemma \ref{Q_and_q_estimates} \ref{itm:Q2} combined with Lemma \ref{lower_bound} \ref{itm:lb1} to estimate the second term. Throughout we recall \eqref{epsilon_lower_bound}, and we obtain the claim if we take $C_2>C_1+c_3+2c_4$.
\end{proof} 

\subsection{The mean-field limit for unbounded nonlocal interactions in dimensions $d=2,3$}
\label{Extension of results to unbounded nonlocal interactions}

We conclude this section by using the techniques developed above to extend the mean-field limit of \cite{LNR3, FKSS_2020}, with nonlocal interaction, from bounded interaction potentials to unbounded interaction potentials. Our assumptions on the potential are the same as in the seminal work \cite{bourgain1997invariant}. Previously, the mean-field limit with unbounded interaction potentials was considered in \cite{sohinger2019microscopic}, however with a modified, regularized, quantum many-body state instead of the grand canonical state \eqref{gc_density}. We remark that the results in \cite{bourgain1997invariant} are originally stated in a setting that does not assume any positivity of the interaction, and hence require a truncation in the Wick-ordered mass of the field. However, when restricted to positive (defocusing) interactions, the truncation can be removed.

\begin{assumption}
\label{v_assumption_L^q}
In the classical setting, we consider $v \in L^q(\Lambda)$ which is even, real-valued, and of positive type, such that 
\begin{equation}
\label{v_assumption_L^q_eq}
\begin{cases}
q>1 & \text{if } d=2
\\
q>3 & \text{if } d=3\,.
\end{cases}
\end{equation}
\end{assumption}
Note that, in terms of $L^q$ integrability, \eqref{v_assumption_L^q_eq} is the optimal range for $q$. We refer the reader to \cite{bourgain1997invariant} and \cite[Section 1.4]{sohinger2019microscopic} for a further discussion. In particular, for $d = 2$ we can take $v$ to be the Coulomb potential.
For $q$ as in \eqref{v_assumption_L^q_eq}, one verifies that $G \in L^{2q'}(\Lambda)$ (where $q'$ denotes the H\"{o}lder conjugate of $q$) by writing $G$ as a Fourier series and using Sobolev embedding.

With $v$ as above, we study the interacting field theory \eqref{field theory}, where now
\begin{equation}
\label{zeta_nonlocal}
V=\frac{1}{2} \int \dd x\,\dd y\, \wick{|\phi(x)|^2}\,v(x-y)\,\wick{|\phi(y)|^2}\,.
\end{equation}
The interaction $V$ is rigorously defined by using a frequency truncation, as in Section \ref{sec:classical_field}. We refer the reader to \cite[Lemma 1.4]{sohinger2019microscopic} for a precise summary.
We also make the appropriate modifications in the definition of the correlation functions.

Similarly as in \eqref{approximate_delta_function}, we consider $\psi \in C_c^{\infty}(\R^d)$ even, nonnegative and satisfying $\psi(0)=1$ and for $\epsilon>0$, we define $\delta_{\epsilon}: \Lambda \rightarrow \C$ by
\begin{equation}
\label{delta_{epsilon,lambda}}
\delta_{\epsilon,\Lambda}(x) \deq \frac{1}{\epsilon^d} \sum_{y \in \Z^d} (\fra F^{-1} \psi) \biggl(\frac{x-y}{\epsilon}\biggr)\,,
\end{equation}
where $\fra F$ denotes Fourier transform (see \eqref{Fourier}).
With notation as in \eqref{delta_{epsilon,lambda}} and $v$ as in Assumption \ref{v_assumption_L^q}, we now write
\begin{equation}
\label{v_regularized}
v^{\epsilon} \deq v*\delta_{\epsilon,\Lambda} \in C^{\infty}(\Lambda)\,.
\end{equation}
From \eqref{v_regularized}, we deduce that 
\begin{equation}
\label{v_regularized_b}
\|v^{\epsilon}\|_{L^q}=\|v\|_{L^q}\,,\qquad v^{\epsilon} \stackrel{L^q}{\longrightarrow} v\,.
\end{equation}

When working in the quantum setting, we hence set $\tau^\epsilon = 0$ and $E^{\epsilon}=0$ in \eqref{Hamiltonian_H}.
In this section, instead of \eqref{epsilon_lower_bound}, we take 
\begin{equation}
\label{epsilon_bound_nonlocal}
\epsilon(\nu) \gtrsim \frac{1}{\log \nu^{-1}}\,.
\end{equation}

\begin{theorem} 
\label{theorem_non_local}
Let $d \leq 3$.
With $V$ as in \eqref{zeta_nonlocal}, $v$ as in Assumption \ref{v_assumption_L^q},  and  $\epsilon>0$ as in \eqref{epsilon_bound_nonlocal}, we have the following results as $\epsilon, \nu \to 0$.
\begin{enumerate}[label=(\roman*)]
\item \label{itm:nonlocal1}
$\cal Z \to \zeta^V$.
\item \label{itm:nonlocal2}
$\nu^p \, \wh \Gamma_p \overset{C}{\longrightarrow} \wh \gamma_p^{V}$.
\end{enumerate}
\end{theorem}
\begin{proof}
We first prove \ref{itm:nonlocal1}.
If $v \in L^q(\Lambda)$, we use \eqref{v_regularized} and Young's inequality, to deduce that for $\epsilon>0$ sufficiently small, we have
\begin{equation}
\label{v_regularized_bound}
\|v^{\epsilon}\|_{L^{\infty}(\Lambda)} \lesssim_{\psi,q,v} \epsilon^{-\frac{d}{q}}\,,\qquad \|\nabla v^{\epsilon}\|_{L^{\infty}(\Lambda)} \lesssim_{\psi,q,v} \epsilon^{-\frac{d}{q}-1}\,.
\end{equation}
Furthermore, $v^{\epsilon}$ is of positive type and has compactly supported Fourier transform.
We use the functional integral setup as before and start from appropriate analogues of Lemmas \ref{Partition_function_rate_of_convergence_1} and \ref{Partition_function_rate_of_convergence_2}. Using \eqref{v_regularized_bound} instead of Lemma \ref{v^{epsilon}_properties}, and arguing as for  Lemmas \ref{Partition_function_rate_of_convergence_3}--\ref{Partition_function_rate_of_convergence_5}, we deduce that for $\epsilon>0$ sufficiently small and all $\nu>0$, we have
\begin{equation}
\label{nonlocal_bound_1}
\Bigl|\cal Z-\cal \zeta^{V^\epsilon}\Bigr| \lesssim \Biggl(\frac{\Theta(\nu)}{\epsilon^{\frac{5d}{q}+1}}\Biggr)^{1/2}\,,
\end{equation}
where $V^{\epsilon}$ denotes the interaction as in \eqref{zeta_nonlocal} with $v$ replaced by $v^{\epsilon}$, and $\Theta(\nu)$ is as in \eqref{Theta_function}. By \eqref{epsilon_bound_nonlocal}, this is an acceptable upper bound and we reduce the claim to showing that 
\begin{equation}
\label{zeta_epsilon_convergence}
\lim_{\epsilon \rightarrow 0} \zeta^{V^\epsilon}=\zeta^V\,. 
\end{equation}
Using the assumption that $v,v^{\epsilon}$ are of positive type and the Cauchy-Schwarz inequality, we obtain
\begin{equation}
\label{zeta_difference}
|\zeta^{\epsilon}-\zeta| \leq \|V^{\epsilon}-V\|_{L^2(\P)}\,.
\end{equation}

We now show that 
\begin{equation}
\label{nonlocal_bound_2}
\|V^{\epsilon}-V\|_{L^2(\P)} \lesssim 
\|v^{\epsilon}-v\|_{L^q}\,. 
\end{equation}
By \eqref{v_regularized_b}, we have that \eqref{nonlocal_bound_2} indeed implies \eqref{zeta_epsilon_convergence}.

By using Fubini's theorem followed by Wick's theorem, we can rewrite the right-hand side of \eqref{zeta_difference} as 
\begin{equation}
\label{nonlocal_bound_3}
\frac{1}{2} \, \biggl(\int \dd x\, \dd 
\tilde x\, \dd y\, \dd \tilde y \,F^{\epsilon}(x,\tilde x,y, \tilde y) \biggr)^{1/2}\,,
\end{equation}
where 
\begin{equation}
\label{F^epsilon}
F^{\epsilon}(x,\tilde x,y, \tilde y) \deq \sum_{\Pi \in \cal M_c^{\mathrm{Wick}}(B)} \prod_{\{i,j\} \in \Pi} G(x_i-x_j)\, \pb{v^{\epsilon}(x-y)-v(x-y)}\,\pb{v^{\epsilon}(\tilde x- \tilde y)-v(\tilde x- \tilde y)}\,.
\end{equation}
In \eqref{F^epsilon}, we let $B \deq \{1,2\} \times \{1,2\} \times \{+,-\}$. Furthermore, we use the variables 
$x_{(1,1,\pm)} \equiv x,\, x_{(2,1,\pm)} \equiv 
\tilde x,\, x_{(1,2,\pm)} \equiv y,\, x_{(2,2,\pm)} \equiv
\tilde y$ and denote by $\cal M_c^{\mathrm{Wick}}(B)$ the set of all complete pairings $\Pi$ of $B$ where 
$\{(a,b,+),(a,b,-)\} \notin \Pi$ for all $a,b \in \{1,2\}$. 

We note that each integration variable in \eqref{nonlocal_bound_3} appears exactly once as part of the argument of $v^{\epsilon}-v$ and exactly twice as part of an argument of a Green function.
We can therefore apply H\"{o}lder's inequality and deduce \eqref{nonlocal_bound_2}. Here,
we use that, uniformly in $c,d \in \Lambda$
\begin{equation}
\label{Holder_estimate}
\int \dd x\, |w(x)|\,\pb{1+G(x-c)}\,\pb{1+G(x-d)} \leq \|w\|_{L^q} \,\pb{1+\|G\|_{L^{2q'}}^2} \lesssim \|w\|_{L^q}\,,
\end{equation}
with $w=v^{\epsilon}-v$, since $G \in L^{2q'}(\Lambda)$. 
When applying \eqref{Holder_estimate}, we recall \eqref{v_regularized_b}.

We now prove \ref{itm:nonlocal2}. With notation defined analogously as in \eqref{Gamma_hat_quantum}--\eqref{q_classical}, we use \eqref{v_regularized_bound} and argue as in the proof of Lemma \ref{Q_and_q_estimates} to deduce that for all $\epsilon>0$ and uniformly in $\eta>0$
\begin{equation}
\label{nonlocal_bound_ii_1}
\|Q_{p,\eta}\|_{C} \leq C_{\kappa,v}^p \,\biggl(\frac{1}{\epsilon^{d/q}}\biggr)^{p/2}\,, \quad 
\|Q_{p,\eta}-q^{\epsilon}_{p}\|_{C} \leq C_{\kappa,v}^p \, \Biggl(\biggl(\frac{1}{\epsilon^{d/q}}\biggr)^{p/2}+\biggl(\frac{1}{\epsilon^{5d/q+1}}\biggr)^{1/2}\Biggr)\,\nu^{\theta(d,p)}\,,
\end{equation}
with $\theta(d,p)$ given as in \eqref{theta_definition}.
Similarly, arguing as in the proof of Lemma \ref{lower_bound}, using Wick's theorem we have that for all $\epsilon>0$
\begin{equation}
\label{nonlocal_bound_ii_2}
\zeta^{V^{\epsilon}} \gtrsim 1\,,\qquad \cal Z \geq \exp \biggl[-c\biggl(1+\frac{\nu \chi(\sqrt{\nu})}{\epsilon^{d/q}}\biggr)\biggr]\,,
\end{equation}
with $\chi$ as in \eqref{sigma(epsilon)}. In order to obtain \eqref{nonlocal_bound_ii_2} we used \eqref{Holder_estimate}. Using \eqref{nonlocal_bound_ii_1}--\eqref{nonlocal_bound_ii_2}, arguing as in the proof of Proposition \ref{Correlation_functions_rate_of_convergence}, and recalling \eqref{epsilon_bound_nonlocal}, we deduce that 
\begin{equation}
\label{nonlocal_bound_ii_3_A}
\Bigl\|\nu^p \, \wh{\Gamma}_{p}-\wh{\gamma}_{p}^{V^{\epsilon}}\Bigr\|_{C} \lesssim \nu^{\alpha}\,,
\end{equation}
for $0<\alpha<\theta(d,p)$.
Hence, we reduce to showing 
\begin{equation}
\label{nonlocal_bound_ii_3}
\lim_{\epsilon \rightarrow 0} \Bigl\|\wh{\gamma}_{p}^{V^{\epsilon}}-\wh{\gamma}_{p}^V\Bigr\|_{C}=0\,.
\end{equation}
By \eqref{nonlocal_bound_2}, we have that $\normb{\ee^{-V^{\epsilon}} - \ee^{-V}}_{L^2(\P)} \to 0$ as $\epsilon \to 0$. Therefore, we obtain \eqref{nonlocal_bound_ii_3} if we prove bounds analogous to those in Lemma  \ref{lem:derivatives-convergence} (with $W^{\epsilon}$ replaced by $V^{\epsilon}$). More precisely, we note that the following claims hold, with $V$ as in \eqref{zeta_nonlocal} and $V^{\epsilon}$ as in \eqref{zeta_nonlocal} with $v$ replaced by $v^{\epsilon}$ and notation as in Lemma \ref{lem:derivatives-convergence}.

\begin{enumerate}[label=(\roman*)]
\item \label{lab:der1_q}
$\sup_N \sup_{\f z} \norm{\cal L_{N, \f z}V}_{L^r(\P)} < \infty$.
\item \label{lab:der2_q}
As $\epsilon \to 0$ we have $\sup_N \sup_{\f z} \norm{\cal L_{N, \f z} V^\epsilon - \cal L_{N, \f z} V}_{L^r(\P)} \to 0$.
\item \label{lab:der3_q}
As $M,N \to \infty$ we have $\sup_{\f z} \norm{(\cal L_{N, \f z} - \cal L_{M, \f z}) V}_{L^r(\P)} \to 0$.
\item \label{lab:der4_q}
For any $\epsilon > 0$, as $M,N \to \infty$ we have $\sup_{\f z} \norm{(\cal L_{N, \f z} - \cal L_{M, \f z}) V^\epsilon}_{L^r(\P)} \to 0$.
\end{enumerate}

We first show \ref{lab:der2_q}. For fixed $0<N<\infty$ and $z \in \Lambda$, we note that by Lemma \ref{lem:L_phi} 
\begin{equation*}
L_{N,z} (V^\epsilon - V) = \int \dd x \, \dd \tilde x \, \pb{v^\epsilon(x - \tilde x) - v(x - \tilde x)} \,
\, G_N(z - x) \,\phi(x)\,\wick{|\phi(\tilde x)|^2}
\end{equation*}
and hence by Lemma \ref{lem:Wick}, we deduce that 
\begin{multline}
\label{(ii)_q}
\normb{L_{N,z} (V^\epsilon - V)}_{L^2(\P)}^2 = \int \dd x \, \dd \tilde x \, \dd y \, \dd \tilde y \,\bigl[v^\epsilon(x - \tilde x) - v(x - \tilde x)\bigr] \,\bigl[v^\epsilon(y - \tilde y) - v(y - \tilde y)\bigr] \,G_N(z-x)\,G_N(z-y)
\\
\times \bigl[G(x-y)\,G^2(\tilde x - \tilde y) + G(x-\tilde x)\,G(\tilde x - \tilde y)\,G(y-\tilde y) + G(x - \tilde y)\, G(\tilde x-\tilde y)\, G(y-\tilde x)\bigr]\,.
\end{multline}
Analogously as in \eqref{nonlocal_bound_3}, each integration variable in \eqref{(ii)_q} appears exactly once as part of the argument of $v^{\epsilon}-v$ and exactly twice as part of the argument of $G$ or $G_N$. We then use H\"{o}lder's inequality as in \eqref{Holder_estimate} where some of the factors of $G$ can be replaced by $G_N$ and note that $\|G_N\|_{L^{2q'}} \leq \|G\|_{L^{2q'}}<\infty$, which holds by \eqref{def_G_N} and Young's inequality. Putting everything together, we get that 
\begin{equation*}
\eqref{(ii)_q} \lesssim \|v^{\epsilon}-v\|_{L^q}^2\,,
\end{equation*}
which is an acceptable bound by \eqref{v_regularized_b}. The other terms for \ref{lab:der2_q} are treated similarly.
We omit the details. 
The proof of \ref{lab:der1_q} is analogous, except that now we apply H\"{o}lder's inequality similarly as in \eqref{Holder_estimate} with $w=v$.

We now show  \ref{lab:der3_q}. We first compute 
\begin{equation*}
(L_{N,z} - L_{M,z}) V = \int \dd x \,\dd \tilde x\,v(x-\tilde x)\, (G_N(x - z) - G_M(x - z)) \phi(x) \wick{\abs{\phi(\tilde x)}^2}\,.
\end{equation*}
Hence, by Lemma \ref{lem:Wick}, we have
\begin{multline}
\label{(iii)_q}
\normb{(L_{N,z}-L_{M,z})V}_{L^2(\P)}^2=
\int \dd x\,\dd \tilde x\,\dd y\,\dd \tilde y\,
v(x-\tilde x)\,v(y-\tilde y)\,
\bigl[G_N(z-x)-G_M(z-x)\bigr]\,
\\
\times \bigl[G_N(z-y)-G_M(z-y)\bigr]\,\bigl[G(x-y)\,G^2(\tilde x-\tilde y)+G(x-\tilde x)\,G(\tilde x-\tilde y)\,G(y-\tilde y)
\\
+G(x-\tilde y)\,G(\tilde x-\tilde y)\,G(y-\tilde x)\bigr]
\end{multline}
Since each integration variable occurs exactly once as part of an argument of $v$ and exactly twice as a part of an argument of $G$ or $G_N-G_M$, we can use H\"{o}lder's inequality as earlier to deduce that 
\begin{equation*}
\eqref{(iii)_q} \lesssim \|v\|_{L^q}^2\,\|G_N-G_M\|_{L^{2q'}}^2\,(1+\|G\|_{L^{2q'}}^3)\,,
\end{equation*}
which is an acceptable upper bound by using \eqref{G_N_est2} as in the proof of Lemma \ref{lem:derivatives-convergence} \ref{lab:der3} .The higher order derivatives are estimated in the same way.
The proof of \ref{lab:der4_q} is analogous, except that we now apply apply H\"{o}lder's inequality similarly as in \eqref{Holder_estimate} with $w=v^{\epsilon}$ and recall \eqref{v_regularized_b}. We hence obtain \eqref{nonlocal_bound_ii_3}. 
\end{proof}

\begin{remark}
\label{weak_operator_topology}
We note that if we relax the topology of convergence in \eqref{nonlocal_bound_ii_3} (and hence in Theorem \ref{theorem_non_local} \ref{itm:nonlocal2}) to the weak operator topology, we can obtain the result by using the first bound in \eqref{nonlocal_bound_ii_2}, and the Cauchy-Schwarz inequality, similarly as in \eqref{zeta_difference}. This applies to the case $(d,q)=1$ as well. We refer the reader to the proof of \cite[Proposition 4.4]{FKSS_2020} for details.
\end{remark}

\begin{remark}
\label{endpoint_admissible}
One can also consider $v \in L^1(\Lambda)$ which is even, real-valued, and of positive type with suitable decay on its Fourier coefficients, see \cite[(16)-(17)]{bourgain1997invariant}. For $d=3$, the assumption in \cite{bourgain1997invariant} is that $\hat{v}(k) \leq \frac{C}{\langle k \rangle^{2+\delta}}$ for some $\delta>0$, which is covered by Theorem \ref{theorem_non_local} above by the Hausdorff-Young inequality (in the classical setting, the decay assumption was recently relaxed in \cite{deng2021invariant}). For $d=2$, the assumption in \cite{bourgain1997invariant} is that $\hat{v}(k) \leq \frac{C}{\langle k \rangle^{\delta}}$ for some $\delta>0$. Note that this corresponds the \emph{endpoint admissible} regime in the terminology of \cite[Definition 1.2 and Section 4]{sohinger2019microscopic}, except that we do not assume pointwise nonnegativity of $v$. Here, it is possible to prove convergence of the partition function (and consequently the convergence of the correlation functions in the weak operator topology as in Remark \ref{weak_operator_topology} above). We present the details in Appendix \ref{Appendix C}.
\end{remark}

\appendix

\section{Wick ordering and hypercontractive moment bounds} \label{sec:wick}

In this appendix we recall some standard facts about Wick ordering and hypercontractive estimates. We refer e.g.\ to \cite{nualart2006malliavin} for a comprehensive account. For the convenience of the reader, we keep this appendix self-contained.

\subsection{Wick ordering} \label{sec:Wick}
Let $\xi = (\xi_1, \dots, \xi_n)$ be a real\footnote{When dealing with complex vectors, we split them into their real and imaginary parts.} Gaussian vector with mean zero. We define the \emph{Wick ordering of the monomial} $\xi_1 \cdots \xi_n$ through
\begin{equation} \label{def_Wick}
\wick{\xi_1 \cdots \xi_n} = \frac{\partial^n}{\partial \lambda_1 \cdots \partial \lambda_n} \frac{\ee^{\lambda \cdot \xi}}{\E [\ee^{\lambda \cdot \xi}]} \biggr \vert_{\lambda = 0}\,.
\end{equation}
The expectation is equal to
\begin{equation*}
\E[ \ee^{\lambda \cdot \xi}] = \ee^{\frac{1}{2} \lambda \cdot \cal C \lambda}\,,
\end{equation*}
where
\begin{equation*}
\cal C_{ij} \deq \E [\xi_i \xi_j]
\end{equation*}
is the covariance matrix of $\xi$. 
Computing the derivatives explicitly, we hence find
\begin{equation} \label{wick_expanded}
\wick{\xi_1 \cdots \xi_n} = \sum_{\Pi \in \fra M([n])} \prod_{i \in [n] \setminus [\Pi]} \xi_i \prod_{\{i,j\} \in \Pi} \p{- \E [\xi_i \xi_j]}\,,
\end{equation}
where we defined $[n] \deq \{1, \dots, n\}$, and $\fra M([n])$ is the set of partial pairings of the set $[n]$ (i.e.\ a set of disjoint unordered pairs of elements of $[n]$) with $[\Pi] \deq \bigcup_{\{i,j\} \in \Pi} \{i,j\}$. Since both sides of \eqref{wick_expanded} are linear in $\xi_1, \dots, \xi_n$, we can extend \eqref{wick_expanded} to a complex Gaussian vector $(\xi_1, \dots, \xi_n)$.

By splitting the summation in \eqref{wick_expanded} over $\Pi$ satisfying $n \notin [\Pi]$ and $n \in [\Pi]$, we obtain the recursion
\begin{equation} \label{Wick_recursion}
\wick{\xi_1 \cdots \xi_n} = \wick{\xi_1 \cdots \xi_{n-1}} \, \xi_n - \sum_{i = 1}^{n-1} \E[\xi_i \xi_n] \, \wick{\xi_1 \cdots \xi_{i-1} \xi_{i+1} \cdots \xi_{n-1}}\,.
\end{equation}
Moreover, from the definition \eqref{def_Wick} we find that Wick ordering commutes with differentiation, since
\begin{equation} \label{Wick_derivative}
\frac{\partial}{\partial \xi_1} \wick{\xi_1 \cdots \xi_n} =  \frac{\partial^n}{\partial \lambda_1 \cdots \partial \lambda_n}  \frac{\lambda_1 \ee^{\lambda \cdot \xi}}{\E [\ee^{\lambda \cdot \xi}]} \biggr \vert_{\lambda = 0}
=  \frac{\partial^{n-1}}{\partial \lambda_2 \cdots \partial \lambda_n}  \frac{\ee^{\sum_{i = 2}^n \lambda_i \xi_i}}{\E [\ee^{\sum_{i = 2}^n \lambda_i \xi_i}]} \biggr \vert_{\lambda = 0} = \wick{\xi_2 \cdots \xi_n}\,.
\end{equation}
Note that the variables $\xi_1, \dots, \xi_n$ are treated as independent for the differentiation, although they need not be stochastically independent. For instance, \eqref{Wick_derivative} implies that $\frac{\dd}{\dd \xi} \,\wick{\xi^n} =  \wick{ n \xi^{n-1}}\,$.

The following is a generalization of Wick's rule to moments of Wick-ordered monomials.
\begin{lemma} \label{lem:Wick}
Let $\xi = (\xi_1, \dots, \xi_n)$ be a complex Gaussian vector with mean zero. Let $Q$ be a partition of $[n]$. Then
\begin{equation*}
\E \qa{\prod_{q \in Q} \wick{\prod_{i \in q} \xi_i}} = \sum_{\Pi \in \fra M_c([n],Q)} \prod_{\{i,j\} \in \Pi} \E [\xi_i \xi_j]\,,
\end{equation*}
where $\fra M_c([n],Q)$ is the set of complete pairings $\Pi$ of the set $[n]$ such that no pair $\{i,j\} \in \Pi$ satisfies $i,j \in q$ for some $q \in Q$.
\end{lemma}
\begin{proof}
By linearity, we may assume that $(\xi_1, \dots, \xi_n)$ is real.
From the definition \eqref{def_Wick} we find
\begin{equation*}
\prod_{q \in Q} \wick{\prod_{i \in q} \xi_i} = \frac{\partial^n}{\partial \lambda_1 \cdots \partial \lambda_n} \ee^{\lambda \cdot \xi} \prod_{q \in Q} \exp \pbb{-\frac{1}{2} \sum_{i,j \in q} \cal C_{ij} \lambda_i \lambda_j} \biggr \vert_{\lambda = 0}\,,
\end{equation*}
so that taking the expectation yields
\begin{equation*}
\E \qa{\prod_{q \in Q} \wick{\prod_{i \in q} \xi_i}} = \frac{\partial^n}{\partial \lambda_1 \cdots \partial \lambda_n} \exp \pbb{\frac{1}{2} \sum_{i,j \in [n]} \cal C^Q_{ij} \lambda_i \lambda_j} \biggr \vert_{\lambda = 0}\,,
\end{equation*}
where $\cal C^Q_{ij} \deq \pb{1 - \sum_{q \in Q} \ind{i,j \in q}} \cal C_{ij}$. The claim now follows by differentiation.
\end{proof}

\begin{example} \label{ex:two_blocks}
If $Q$  has two blocks, Lemma \ref{lem:Wick} takes on the following form. If $(\xi_1, \dots, \xi_n, \zeta_1, \dots, \zeta_n)$ is a complex Gaussian vector with mean zero, then
\begin{equation*}
\E[\wick{\xi_1 \cdots \xi_n} \, \wick{\zeta_1 \cdots \zeta_n}] = \sum_{\sigma \in S_n} \prod_{i = 1}^n \E[\xi_i \zeta_{\sigma(i)}]\,.
\end{equation*}
\end{example}

\subsection{Hypercontractive moment bounds} \label{sec:hypercontr}

Let $(\cal H, \scalar{\cdot}{\cdot})$ be a separable real Hilbert space, and let $(\phi(f))_{f \in \cal H}$ the abstract Gaussian process indexed by $\cal H$. For an explicit definition, let $(e_k)_{k \in \N}$ be an orthonormal basis of $\cal H$, and let $\Omega = \R^\N$ be equipped with the product sigma-algebra $\cal F$ and the probability measure $\P$, which is an infinite product of standard Gaussians. With the notation $\omega = (\omega_k)_{k \in \N} \in \Omega$, we define
\begin{equation*}
\phi(f) \deq \sum_{k \in \N} \omega_k \scalar{e_k}{f}\,,
\end{equation*}
which converges in $L^2(\Omega, \cal F, \P)$.
Moreover, since $\E[\phi(f) \phi(g)] = \scalar{f}{g}$, the map $\phi \col \cal H \mapsto L^2(\Omega, \cal F, \P)$ is an isometry.

For $n \in \N$ we define the $n$th \emph{polynomial chaos}, denoted by $\cal B_n$, as the closure of the subspace of $L^2(\Omega, \cal F, \P)$ spanned by random variables of the form
\begin{equation*}
\wick{\phi(f_1) \cdots \phi(f_n)} \,, \qquad f_1, \dots, f_n \in \cal H\,.
\end{equation*}

\begin{lemma} \label{lem:chaos_decomp}
We have
\begin{equation*}
L^2(\Omega, \cal F, \P) = \bigoplus_{n \in \N} \cal B_n\,.
\end{equation*}
\end{lemma}

\begin{proof}
The orthogonality of the spaces $(\cal B_n)_{n \in \N}$ is easy to deduce from the definition of Wick ordering. It remains to show that $L^2(\Omega, \cal F, \P) \subset \bigoplus_{n \in \N} \cal B_n$.

For $K \in \N$, let $\cal H_K  = \Span (e_k \col k \leq K) \subset \cal H$. Let $\cal B_n^{(K)}$ be the subspace of $L^2(\Omega, \cal F, \P)$ spanned by random variables of the form $\wick{\phi(f_1) \cdots \phi(f_n)}$ with $f_1, \dots, f_n \in \cal H_K$. Let $\cal F_K$ be the sigma-algebra generated by $(\phi(f) \col f \in \cal H_K)$. Note that $\cal B_n^{(K)} \subset L^2(\Omega, \cal F_K, \P)$. We now claim that
\begin{equation} \label{L2_K}
L^2(\Omega, \cal F_K, \P) = \bigoplus_{n \in \N} \cal B_n^{(K)}\,.
\end{equation}
To see this, it suffices to show that if $\xi \in L^2(\Omega, \cal F_K, \P)$ satisfies $\E [\xi \zeta] = 0$ for every $\zeta \in \cal B_n^{(K)}$ and every $n \in \N$ then $\xi = 0$. We can write $\xi = f(\phi(e_1), \dots, \phi(e_K))$ for some measurable function $f \col \R^K \to \R$, which is square integrable with respect to the the standard Gaussian measure on $\R^K$, which we denote by $\mu_K$. By assumption, $\int \mu_K(\dd x) \,f(x) \, P(x)=0$ for all polynomials $P$ on $\R^K$. Using that $\mu_K$ has Gaussian tails we easily deduce that $f = 0$.

To deduce the claim from \eqref{L2_K}, we choose $\xi \in L^2(\Omega, \cal F, \P)$ and set $\xi_K \deq \E[\xi | \cal F_K]$ so that $(\xi_K)_{K \in \N}$ is a martingale. By Doob's martingale convergence theorem (see e.g.\ \cite[Chapter 12]{williams1991probability}), we have $\xi_K \to \xi$ in $L^2$, which concludes the proof.
\end{proof}

Next, we state and prove a hypercontractive moment bound. Such an estimate is usually derived as a consequence of the hypercontractive property of the Ornstein-Uhlenbeck semigroup associated with $\phi$. Here we give an elementary and very simple argument, relying only on Lemma \ref{lem:Wick} and the Cauchy-Schwarz inequality.

\begin{lemma} \label{lem:hyper}
Let $n \in \N$ and $\xi \in \cal B_n$. Then for any $p \in 2\N^*$ and some universal constant $C$, we have
\begin{equation*}
\E [\xi^{p}] \leq C p^{np/2} \, (\E [\xi^2])^{p/2}\,.
\end{equation*}
\end{lemma}
\begin{proof}
Let
\begin{equation*}
\xi = \sum_{k_1, \dots, k_n \in \N} a_{k_1 \cdots k_n} \wick{\phi(e_{k_1}) \cdots \phi(e_{k_n})}\,,
\end{equation*}
where $a_{k_1 \cdots k_n}$ is symmetric under permutations. By Lemma \ref{lem:Wick} (see also Example \ref{ex:two_blocks}) we have
\begin{equation} \label{Z_variance}
\E [\xi^2] = n! \sum_{k_1, \dots, k_n} a_{k_1 \cdots k_n}^2\,.
\end{equation}
To estimate the $p$th moment, it is convenient to introduce the index sets $I = [p] \times [n]$ and $I_i = \{i\} \times [n]$ for $i \in [p]$. For any $A \subset I$ we use the notation $\f k_A = (k_{ij} \col (i,j) \in A) \in \N^{A}$ for the summation variables indexed by the set $A$. In this notation we can write $\xi = \sum_{\f k_{I_i}} a_{\f k_{I_i}} \wick{\prod_{j = 1}^n \phi(e_{k_{ij}})}$ for each $i \in [n]$, and hence we get, using Lemma \ref{lem:Wick},
\begin{equation} \label{E_xi_n}
\E [\xi^p] = \sum_{\f k_I} \prod_{i = 1}^p a_{\f k_{I_i}} \E \qBB{\prod_{i = 1}^p \wick{\prod_{j = 1}^n \phi(e_{k_{ij}})}} = \sum_{\Pi} \sum_{\f k_I} \prod_{i = 1}^p a_{\f k_{I_i}} \prod_{\{(i,j), (i',j')\} \in \Pi} \, \ind{k_{ij} = k_{i'j'}}\,,
\end{equation}
where the summation ranges over all complete pairings $\Pi$ of $[p] \times [n]$ such that for all $\{(i,j), (i',j')\} \in \Pi$ we have $i \neq i'$.
 
The idea is to fix $\Pi$ and to sum over the variables $k_{ij}$ in pairs connected by the delta function on the right-hand side of \eqref{E_xi_n}, by using a simple repeated application of Cauchy-Schwarz. To that end, we introduce an inductive summation of the edges of $\Pi$ one by one; the order of summation is immaterial. After summing out a number of edges, we obtain a partial pairing $\Sigma \subset \Pi$, whose blocks contain those summation variables that have not yet been summed out. We refer to Figure \ref{fig:hypercontractivity} for an illustration. For the example $\Pi$ of Figure \ref{fig:hypercontractivity}, the complete estimate can be  explicitly written out:
\begin{align*}
\sum_{k,l,m,u,v,w} a_{klm} \, a_{klu} \, a_{vwu} \, a_{vwm}
&\leq \sum_{l,m,u,v,w} \pBB{\sum_k a_{klm}^2}^{1/2} \pBB{\sum_k a_{klu}^2}^{1/2} a_{vwu} \, a_{vwm}
\\
&\leq \sum_{l,m,u,v} \pBB{\sum_k a_{klm}^2}^{1/2} \pBB{\sum_k a_{klu}^2}^{1/2} \pBB{\sum_{w} a_{vwu}^2}^{1/2} \pBB{\sum_w a_{vwm}^2}^{1/2}
\\
&\leq \sum_{l,u,v} \pBB{\sum_{k,m} a_{klm}^2}^{1/2} \pBB{\sum_{k} a_{klu}^2}^{1/2} \pBB{\sum_{w} a_{vwu}^2}^{1/2} \pBB{\sum_{m,w} a_{vwm}^2}^{1/2}
\\
&\leq \sum_{u,v} \pBB{\sum_{k,l,m} a_{klm}^2}^{1/2} \pBB{\sum_{k,l} a_{klu}^2}^{1/2} \pBB{\sum_{w} a_{vwu}^2}^{1/2} \pBB{\sum_{m,w} a_{vwm}^2}^{1/2}
\\
&\leq \sum_{v} \pBB{\sum_{k,l,m} a_{klm}^2}^{1/2} \pBB{\sum_{k,l,u} a_{klu}^2}^{1/2} \pBB{\sum_{u,w} a_{vwu}^2}^{1/2} \pBB{\sum_{m,w} a_{vwm}^2}^{1/2}
\\
&\leq \pBB{\sum_{k,l,m} a_{klm}^2}^{1/2} \pBB{\sum_{k,l,u} a_{klu}^2}^{1/2} \pBB{\sum_{u,v,w} a_{vwu}^2}^{1/2} \pBB{\sum_{m,v,w} a_{vwm}^2}^{1/2}\,.
\end{align*}
The estimate has six steps, corresponding to the six edges to be summed out. The third step corresponds to the step illustrated in the right half of Figure \ref{fig:hypercontractivity}.

\begin{figure}[!ht]
\begin{center}
{\small 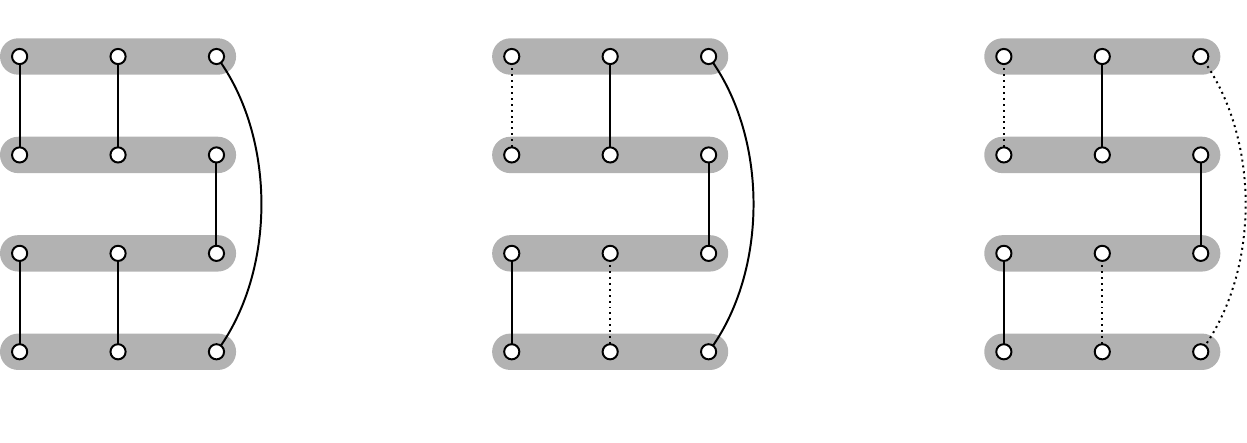}
\end{center}
\caption{An illustration of the inductive algorithm for successively summing out edges of the pairing $\Pi$. Here $p = 4$ and $n = 3$. The pairing $\Pi$ is illustrated in the figure on the left-hand side. The four grey blocks correspond to the four factors $a_{\f k_{I_1}}, \dots, a_{\f k_{I_4}}$, and each vertex $(i,j) \in [p] \times [n]$ corresponds to a summation variable $k_{ij}$. Note that each edge connects vertices from different grey blocks, as is required by Lemma \ref{lem:Wick}. A partial pairing $\Sigma \subset \Pi$ is represented in the figure on the right-hand side, where the edges of $\Sigma$ are drawn using solid lines. The dashed lines, incident to vertices corresponding to the summation variables $\f k_{I \setminus [\Sigma]}$, have been summed out at this point. This summation contributed a factor $\prod_{i = 1}^p\pb{\sum_{\f k_{I_i \setminus [\Sigma]}} a_{\f k_{I_i}}^2}^{1/2}$, which depends on the remaining summation variables $\f k_{[\Sigma]}$ that correspond to the vertices incident to the edges of $\Pi$. The middle figure corresponds to the partial pairing $\Sigma \cup \{e\}$ with an edge $e \in \Pi \setminus \Sigma$ given by $e = \{(i,j), (i',j')\}$. The induction step underlying the argument, going from $\Sigma \cup \{e\}$ to $\Sigma$, is the summation of a single edge $e$, which amounts to summing over the variables $k_{ij} = k_{i',j',}$ and using Cauchy-Schwarz. For this, it is crucial that $i \neq i'$, i.e.\ $e$ connects vertices in different grey blocks.
\label{fig:hypercontractivity}}
\end{figure}

To describe the general procedure more formally, we denote, for any partial pairing $\Sigma \subset \Pi$,
\begin{equation} \label{def_V_Pi}
V_\Pi(\Sigma) \deq \sum_{\f k_{[\Sigma]}} \prod_{i = 1}^p \pBB{\sum_{\f k_{I_i \setminus [\Sigma]}} a_{\f k_{I_i}}^2}^{1/2} \prod_{\{(i,j), (i',j')\} \in \Sigma} \, \ind{k_{ij} = k_{i'j'}}\,,
\end{equation}
where $[\Sigma] \deq \bigcup_{\sigma \in \Sigma} \sigma$. This expression has three crucial properties. First, by \eqref{E_xi_n} we have
\begin{equation} \label{hyp1}
\E [\xi^p] = \sum_\Pi V_\Pi(\Pi)
\end{equation}
Second, by the definition \eqref{def_V_Pi} we have
\begin{equation} \label{hyp2}
V_\Pi(\emptyset) = \pBB{\sum_{k_1, \dots, k_n} a_{k_1 \cdots k_n}^2}^{p/2}\,.
\end{equation}
Third, for any partial pairing $\Sigma \subset \Pi$ and any edge $e \in \Pi \setminus \Sigma$ we have
\begin{equation}\label{hyp3}
V_\Pi(\Sigma \cup \{e\}) \leq V_\Pi(\Sigma)\,.
\end{equation}
To show \eqref{hyp3}, we suppose that $e = \{(i,j), (i',j')\}$ with $i \neq i'$ (see Figure \ref{fig:hypercontractivity}), and estimate the sum over $\f k_e = (k_{ij}, k_{i'j'})$ in the expression \eqref{def_V_Pi} for $V_\Pi(\Sigma \cup \{e\})$ as
\begin{equation*}
\sum_{\f k_e} \pBB{\sum_{\f k_{I_i \setminus [\Sigma \cup \{e\}]}} a^2_{\f k_{I_i}}}^{1/2} \pBB{\sum_{\f k_{I_{i'} \setminus [\Sigma \cup \{e\}]}} a^2_{\f k_{I_{i'}}}}^{1/2} \ind{k_{ij} = k_{i'j'}}
\leq  \pBB{\sum_{\f k_{I_i \setminus [\Sigma]}} a^2_{\f k_{I_i}}}^{1/2} \pBB{\sum_{\f k_{I_{i'} \setminus [\Sigma]}} a^2_{\f k_{I_{i'}}}}^{1/2}\,,
\end{equation*}
by Cauchy-Schwarz and the fact that the summation variable $k_{ij} = k_{i'j'}$ appears exactly once in each factor.

From \eqref{hyp1}--\eqref{hyp3}, we conclude that
\begin{equation} \label{hyp4}
\E[\xi^p] \leq \sum_{\Pi} \pBB{\sum_{k_1, \dots, k_n} a_{k_1 \cdots k_n}^2}^{p/2}\,.
\end{equation}
Using that the number of pairings $\Pi$ is bounded by $\frac{(np)!}{(np/2)! 2^{np/2}}$, we conclude from \eqref{hyp4} and \eqref{Z_variance} that 
\begin{equation*}
\E [\xi^p] \leq \frac{(np)!}{(np/2)! 2^{np/2}} \frac{1}{(n!)^{p/2}} (\E \xi^2)^{p/2} \leq C p^{np/2} \, (\E \xi^2)^{p/2}\,,
\end{equation*}
by Stirling's approximation. This concludes the proof.
\end{proof}

\begin{remark} \label{rem:hypercontr_complex}
If $\xi$ is complex-valued such that $\re \xi$ and $\im \xi$ both belong to $\cal B_n$, then by Minkowski's inequality we deduce from Lemma \ref{lem:hyper} that for any $p \in 2 \N^*$ we have
\begin{equation*}
\E \abs{\xi}^{p} \leq C p^{np/2} \, (\E \abs{\xi}^2)^{p/2}\,.
\end{equation*}
\end{remark}

\section{Basic estimates for the Green function} \label{sec:G}
In this appendix we prove some basic analytic properties of the Green function $G$ on the torus. All of these results are well known, and we collect and prove them here for the reader's convenience. The following lemma is similar to e.g.\ \cite[Lemma 5.4]{rougerie2016higher} or \cite[Section III.3]{aftalion2001pinning}.

\begin{lemma} \label{lem:G_smoothness}
Let $d = 2$.
There exists $\widetilde{G} \in C^\infty (\Lambda \setminus \{0\})$ such that
\[ G (x) = - \frac{1}{\pi} \log |x| + \widetilde{G} (x)\,, \] 
as well as $\abs{\wt G(x)} + \abs{\nabla \wt G(x)} \lesssim 1$ and $\abs{\nabla^2 \wt G (x)} \lesssim 1 + \abs{\log \abs{x}}$.
\end{lemma} 

\begin{proof}
Let $g_\kappa$ denote the Green function of the operator $\kappa - \Delta / 2$ on $\R^2$, which for $\kappa > 0$ has Fourier transform
\begin{equation*}
\hat g_\kappa(\xi) = \int \dd x \, g_\kappa(x) \, \ee^{- 2 \pi \ii \xi \cdot x} = \frac{1}{\kappa + 2 \pi^2 \abs{\xi}^2}\,.
\end{equation*}
Moreover, $g_0(x) = \frac{1}{\pi} \log \abs{x}^{-1}$.
Let $0 < r \leq 1/4$, and choose $\varphi \in [0,1]$ to be a smooth compactly supported function equal to $1$ in the ball of radius $r$ around the origin and zero outside the ball of radius $2r$ around the origin.

First we note that $(1 - \varphi) g_\kappa \in C^\infty(\Lambda)$. This is a manifestation of elliptic regularity (see for instance \cite[Theorem 8.6.1]{friedlander1998introduction}), which can be seen directly by applying the operator $(\kappa - \Delta/2)^k$ to $(1 - \varphi) g_\kappa$, using the Leibniz rule, and then estimating the decay of $\hat g_\kappa$ for large enough $k$ using the a priori bound $\abs{\hat g_\kappa}(\xi) \leq 1/\kappa$. 

Using $h$ to denote some generic smooth function which may change from one expression to the next, we therefore find that
\begin{equation*}
(\kappa - \Delta/2) \varphi (g_\kappa - g_0) = - \kappa \varphi g_0 + h\,.
\end{equation*}
We conclude that
\begin{equation*}
\varphi g_\kappa = \varphi g_0 - \kappa g_\kappa * (\varphi g_0) + h = \varphi g_0 - \kappa (\varphi g_\kappa) * (\varphi g_0) + h\,.
\end{equation*}
By plugging this equation into itself, we conclude that
\begin{equation*}
\varphi g_\kappa - \varphi g_0 + \kappa (\varphi g_0) * (\varphi g_0) \in C^2(\Lambda)\,,
\end{equation*}
since $(\varphi g_\kappa) * (\varphi g_0) * (\varphi g_0) \in C^2(\Lambda)$ because its Fourier transform decays as $\abs{\xi}^{-6}$ for $\abs{\xi} \to \infty$.

It is elementary that $\wt g \deq \kappa (\varphi g_0) * (\varphi g_0)$ satisfies $\abs{\wt g(x)} +  \abs{\nabla \wt g(x)} \lesssim 1$ and $\abs{\nabla^2 \wt g (x)} \lesssim 1 + \abs{\log \abs{x}}$ for $\abs{x} \leq 1$.
The claim now follows from the observation that
\begin{equation*}
G(x) = \sum_{n \in \Z^d} g_\kappa(x + n)\,,
\end{equation*}
by Poisson summation.
\end{proof}

\begin{remark}
\label{G_smoothness_3D_remark}
An analogue of the result of Lemma \ref{lem:G_smoothness} holds in three dimensions, with the same proof. Namely, if $G$ denotes the Green function of the operator $\kappa - \Delta / 2$ on $\Lambda$ for $d = 3$, then we have
\begin{equation*}
G(x)=\frac{1}{4 \pi \abs{x}} + \wt G(x)\,,
\end{equation*}
where $\abs{\wt G(x)} \lesssim 1$.
\end{remark}

\begin{lemma} \label{lem:prop_G_N}
Let $d = 2$.
The truncated Green function from \eqref{def_G_N} satisfies the bounds
\begin{equation} \label{G_N_est1}
\abs{G_N(x)} \lesssim (1 + \abs{\log \abs{x}}) \wedge\log N
\end{equation}
and
\begin{equation} 
\label{G_N_est2}
\abs{G_N(x) - G(x)} \lesssim(1 + \abs{\log (N \abs{x})}) \wedge \frac{1}{N^2 \abs{x}^2}\,.
\end{equation}
\end{lemma}

\begin{proof}
The bound \eqref{G_N_est1} follows easily from Lemma \ref{lem:G_smoothness} and the definition \eqref{def_G_N}, by considering the cases $\abs{x} > 2/N$ and $\abs{x} \leq 2/N$ separately.

To prove \eqref{G_N_est2}, we write
\begin{equation} \label{G_N-G}
G_N(x) - G(x) = \int \dd y \, \pb{G(x + y) - G(x)} \, \rho_N(y)
\end{equation}
For $\abs{x} \leq 1/8$, we get from Lemma \ref{lem:G_smoothness} that
\begin{equation} \label{G_N-G2}
\abs{G_N(x) - G(x)} \lesssim \absBB{\int \dd y \, \pb{\wt G(x + y) - \wt G(x)} \, \rho_N(y)} + \absBB{\int_{\R^2} \dd y \, \pbb{\log \absbb{x + \frac{y}{N}}^2 - \log \abs{x}^2} \, \rho(y)}\,.
\end{equation}
If $\abs{x} \leq 8/N$ then we estimate the right-hand side of \eqref{G_N-G2} by
\begin{equation*}
C + 2 \int_{\R^2} \dd y \, \absBB{\log \absbb{\frac{x}{\abs{x}} + \frac{y}{N\abs{x}}}} \, \rho(y) \lesssim 1 + \biggl|\log \frac{1}{N \abs{x}}\biggr|\,.
\end{equation*}
If $8/N < \abs{x} \leq 1/8$ we estimate the right-hand side of \eqref{G_N-G2} as
\begin{equation*}
\frac{1 + \abs{\log \abs{x}}}{N^2} + \absBB{ \int_{\R^2} \dd y \, \log \pbb{1+ 2 \frac{x \cdot y}{N\abs{x}^2 } + \frac{\abs{y}^2}{N^2 \abs{x}^2 }} \, \rho(y)} \lesssim \frac{1}{N^2 \abs{x}^2}\,,
\end{equation*}
by Taylor expansion and the fact that $\rho$ is an even function. Similarly, for the case $\abs{x} > 1/8$ we easily get from Lemma \ref{lem:G_smoothness}, a Taylor expansion and the evenness of $\rho$ that the right-hand side of \eqref{G_N-G} is bounded by $1/N^2$. This concludes the proof.
\end{proof}

\section{Proofs of auxiliary claims from Section \ref{The rate of convergence}}
\label{Appendix C}

Throughout this appendix, we recall that $\epsilon$ satisfies \eqref{epsilon_lower_bound}.
Before proceeding to the proof of Lemma \ref{J_approximation}, we note a quantitative analogue of \cite[Lemma 5.12]{FKSS_2020} concerning path-path interactions.

\begin{lemma}
\label{5.12'}
For continuous paths $\omega \in \Omega^{\tau_1,\tilde{\tau}_1}, \tilde{\omega} \in \Omega^{\tau_2,\tilde{\tau}_2}$, the following properties hold.
\begin{enumerate}[label=(\roman*)]
\item \label{itm:512_1}
We have, uniformly in $\eta>0$
\begin{equation}
\label{5.12'_i}
\mathbb{V}_{\eta}(\omega,\tilde{\omega})=\mathbb{V}^{\epsilon}(\omega,\tilde{\omega})+
\mathbb{\hat{V}}_{\eta}(\omega,\tilde{\omega})+O\biggl(\frac{\nu}{\epsilon^d}(\tau_1-\tilde \tau_1)\biggr)\,,
\end{equation}
where $\mathbb{\hat{V}}_{\eta}(\omega,\tilde{\omega})$ satisfies
\begin{equation}
\label{V_hat_intermediate_estimate}
\bigl|\mathbb{\hat{V}}_{\eta}(\omega,\tilde{\omega})\bigr| \leq \frac{C}{\epsilon^{d+1}}\,\int \dd s\,\sum_{i=1}^{N} \int_{I_i} \dd \tilde{s}\,\int_{I_i} \dd \hat{s}\,|\tilde{\omega}(\tilde{s})-\tilde{\omega}(\hat{s})|_{\Lambda}\,\delta_{\eta}([s]_{\nu}-[\hat{s}]_{\nu})\,,
\end{equation}
for $N \deq \lceil \frac{\tau_2-\tilde{\tau}_2}{\nu}\rceil$ and 
\begin{equation}
\label{I_i}
I_{i}\deq
\begin{cases}
[\tilde{\tau}_{2}+(i-1) \nu,\tilde{\tau}_{2}+ i\nu] &\mbox{for }1 \leq i \leq N-1\,,
\\
[\tilde{\tau}_{2} +(N-1)\nu,{\tau}_{2}]  &\mbox{for }i=N\,.
\end{cases}
\end{equation}
\item \label{itm:512_2}
Consider continuous paths $\omega' \in \Omega^{\tau'_{1},\tilde{\tau}'_{1}}, \tilde{\omega}' \in \Omega^{\tau'_{2},\tilde{\tau}'_{2}}$ that agree with the paths $\omega, \tilde{\omega}$ respectively on the intersection of their domains. Suppose that 
$|\tau_j-\tau_j'|\leq \nu$ and $|\tilde{\tau}_j-\tilde{\tau}'_j|\leq \nu$ for $j=1,2$. Then, we have
\begin{equation*}
\mathbb{V}^{\epsilon}(\omega,\tilde{\omega})=\mathbb{V}^{\epsilon}(\omega',\tilde{\omega}')+O\biggl(\frac{\nu}{\epsilon^d}\,(\tau_1-\tilde\tau_1)+\frac{\nu}{\epsilon^d}\,(\tau_2-\tilde\tau_2)+\frac{\nu^2}{\epsilon^d}\biggr)\,.
\end{equation*}
\end{enumerate}
\end{lemma}

\begin{proof}[Proof of Lemma \ref{5.12'}]
Let us first show \ref{itm:512_1}.
By using Definition \ref{V_classical}, \eqref{delta_nu_integral}, and \eqref{I_i} we write
\begin{equation}
\label{Lemma_C1_1}
\mathbb{V}_{\eta}(\omega,\tilde{\omega})= \int \dd s\,\sum_{i=1}^{N} \int_{I_i} \dd \tilde{s}\,\int_{\hat{I}_i} \dd \hat{s}\,v^{\epsilon}\bigl(\omega(s)-\tilde{\omega}(\hat{s})\bigr)\,\delta_{\eta}([s]_{\nu}-[\hat{s}]_{\nu})\,,
\end{equation}
where $\hat{I}_i:=[\tilde{\tau}_2+(i-1)\nu, \tilde{\tau}_2+i \nu]$ for $1 \leq i \leq N$.
By Lemma \ref{v^{epsilon}_properties} \ref{itm:veps1}, it follows that we have
\begin{equation}
\label{Lemma_C1_2}
\int \dd s\,\int_{I_N} \dd \tilde s\, \int_{I_N \setminus \hat{I}_N} \dd \hat{s}\,
v^{\epsilon}\bigl(\omega(s)-\tilde{\omega}(\hat{s})\bigr)\,\delta_{\eta}([s]_{\nu}-[\hat{s}]_{\nu})
=O\biggl(\frac{\nu}{\epsilon^d}(\tau_1-\tilde \tau_1)\biggr)\,.
\end{equation}
For \eqref{Lemma_C1_2}, we also used \eqref{delta_nu_integral} to deduce that 
\begin{equation*}
0 \leq \int_{\hat{I}_N \setminus I_N} \dd \hat{s}\,\delta_{\eta,\nu}([s]_{\nu}-[\hat{s}]_{\nu}) \leq 1\,.
\end{equation*}
Combining \eqref{Lemma_C1_1} and \eqref{Lemma_C1_1}, it follows that 
\begin{equation}
\label{Lemma_C1_3}
\mathbb{V}_{\eta}(\omega,\tilde{\omega})= \int \dd s\,\sum_{i=1}^{N} \int_{I_i} \dd \tilde{s}\,\int_{I_i} \dd \hat{s}\,v^{\epsilon}\bigl(\omega(s)-\tilde{\omega}(\hat{s})\bigr)\,\delta_{\eta}([s]_{\nu}-[\hat{s}]_{\nu})+O\biggl(\frac{\nu}{\epsilon^d}(\tau_1-\tilde \tau_1)\biggr)\,.
\end{equation}
From \eqref{Lemma_C1_3} and 
Definition \ref{V_quantum}, we obtain \eqref{5.12'_i} with 
\begin{equation}
\label{V_hat_definition}
\mathbb{\hat{V}}_{\eta}(\omega,\tilde{\omega}) \deq \int \dd s\,\sum_{i=1}^{N} \int_{I_i} \dd \tilde{s}\,\int_{I_i} \dd \hat{s}\,\Bigl[v^{\epsilon}\bigl(\omega(s)-\tilde{\omega}(\hat{s})\bigr)-v^{\epsilon}\bigl(\omega(s)-\tilde{\omega}(\tilde{s})\bigr)\Bigr]\,\delta_{\eta}([s]_{\nu}-[\hat{s}]_{\nu})\,.
\end{equation}
By using Lemma \ref{v^{epsilon}_properties} \ref{itm:veps2}, we get that \eqref{V_hat_definition} satisfies \eqref{V_hat_intermediate_estimate}.
Claim \ref{itm:512_2} follows from Definition \ref{V_classical} by using Lemma \ref{v^{epsilon}_properties} \ref{itm:veps1}.
\end{proof}

We now give the proofs of Lemmas \ref{Q_and_q_estimates} and \ref{lower_bound}, which are used in the analysis of the rate of convergence for correlation functions.

\begin{proof}[Proof of Lemma \ref{Q_and_q_estimates}]
In order to prove \ref{itm:Q1}, we rewrite \eqref{Q_quantum} as
\begin{equation}
\label{Y_definition}
(Q_{p,\eta})_{\f x,\tilde{\f x}}=\nu^p\,\sum_{\f r \in (\nu \N^*)^p}\,\ee^{-\kappa |\f r|}\,Y_{\f x,\tilde{\f x}}(\f r)\,,
\end{equation}
where
\begin{equation}
\label{Y_definition2}
Y_{\f x,\tilde{\f x}} (\f r) \deq \int \mu_{\cal C_{\eta}} (\dd \sigma)\, \ee^{T^{\epsilon}-\frac{\ii \tau^\epsilon [\sigma]}{\nu}}\,\ee^{F_2(\sigma)}\,
\prod_{i=1}^{p} \Biggl[\int \mathbb{W}^{r_i,0}_{x_i,\tilde{x}_i}(\dd \omega_i)\, \Bigl(\ee^{\ii \int_{0}^{r_i} 
\dd t\, \sigma([t]_{\nu},\omega_i(t))}-1\Bigr)\Biggr]\,.
\end{equation}
Here we recall the definition of $T^{\epsilon}$ in \eqref{T^epsilon} and the definition of $F_2$ in \eqref{F_2}.
By arguing analogously as for \cite[(5.41)]{FKSS_2020}, we deduce that
\begin{equation}
\label{Y_bound}
|Y_{\f x,\tilde{\f x}} (\f r)| \leq C_d^p\, \ee^{T^{\epsilon}}\,\|v^{\epsilon}\|_{L^{\infty}}^{p/2}\, \prod_{i=1}^{p} \bigl(r_i+r_i^{1-d/2}\bigr) \leq  C_{d,v}^p\,\ee^{T^{\epsilon}}\, \biggl(\frac{1}{\epsilon^d}\biggr)^{p/2}\,\prod_{i=1}^{p} \bigl(r_i+r_i^{1-d/2}\bigr)\,,
\end{equation}
where in the second inequality we used Lemma \ref{v^{epsilon}_properties} \ref{itm:veps1}.
Substituting \eqref{Y_bound} into \eqref{Y_definition}, considering Riemann sums, and recalling \eqref{T^epsilon_bound}, we deduce claim \ref{itm:Q1}.

We now prove claim \ref{itm:Q2}. 
Analogously to \eqref{Y_definition}, we rewrite \eqref{q_classical} as 
\begin{equation}
\label{y_definition}
(q_{p}^{\epsilon})_{\f x,\tilde{\f x}}=\int_{[0,\infty)^p}\,\dd \f r\,\ee^{-\kappa |\f r|}\,y^{\epsilon}_{\f x,\tilde{\f x}}(\f r,\f r)\,,
\end{equation}
where for $\f r, \f s \in [0,\infty)^p$ with $\f s \leq \f r$ (meaning that $s_i \leq r_i$ for all $i$), we define
\begin{equation}
\label{y_definition2}
y^{\epsilon}_{\f x,\tilde{\f x}} (\f r,\f s) \deq  \int \mu_{v^{\epsilon}} (\dd \xi)\, \ee^{T^{\epsilon}-\ii \tau^\epsilon \langle \xi,1 \rangle_{L^2}}\,\ee^{f_2(\xi)}\,
\prod_{i=1}^{p} \Biggl[\int \mathbb{W}^{r_i,0}_{x_i,\tilde{x}_i}(\dd \omega_i)\, \Bigl(\ee^{\ii \int_{0}^{s_i} 
\dd t\, \xi(\omega_i(t))}-1\Bigr)\Biggr]\,.
\end{equation}
Here, we recall the definition of $f_2$ in \eqref{f_2}.
Similarly as in \cite[Section 5.2]{FKSS_2020}, we now divide the proof into the following three steps.
\paragraph{Step 1} We truncate the variables $r_i$ occurring in \eqref{Y_definition}, \eqref{y_definition} to lie in some interval $[\delta,1/\delta]$ for $\delta>0$ small.
\paragraph{Step 2} In the expression obtained by using the above truncation in $q^{\epsilon}_p$ given by \eqref{q_classical}, we compare the integral $\int_{\delta}^{1/\delta}\,\dd r_i\, (\cdots)$ with the corresponding Riemann sum $\nu\,\sum_{r_i \in \nu \N^{*}}\,\ind{\delta \leq r_{i} \leq 1/\delta}\, (\cdots)$.
\paragraph{Step 3} We replace $\xi$ by $\langle \sigma \rangle$ given by \eqref{sigma_time_average} in the functional integral and compare the resulting approximations of $q^{\epsilon}_p$ and $Q^{\epsilon}_{p}$.

We now carry out the details of each of the steps above. At every stage, we keep explicit track of the error terms.
\paragraph{Step 1}
Denoting $[a,b]_{\nu} \deq [a,b] \cap \nu \Z$, we observe the following estimate.
\begin{align}
\notag
\biggl|(Q_{p,\eta})_{\f x,\tilde{\f x}}-\nu^p\,\sum_{\f r \in [\delta,1/\delta]_{\nu}^p}\,\ee^{-\kappa |\f r|}\,Y_{\f x,\tilde{\f x}}(\f r)\biggr| +
\biggl|(q_{p}^{\epsilon})_{\f x,\tilde{\f x}}-&\int_{[\delta,1/\delta]^p}\dd \f r\,\ee^{-\kappa |\f r|}\,y^{\epsilon}_{\f x,\tilde{\f x}}(\f r,\f r)\biggr| 
\\
\label{Lemma 5.15'-5.16'}
&\leq C_{\kappa,v}^p\,\ee^{T^{\epsilon}}\,\epsilon^{-\frac{dp}{2}}\,\delta^{2-\frac{d}{2}}\,.
\end{align}
The estimate for the first term on the left-hand side of \eqref{Lemma 5.15'-5.16'} follows from \eqref{Y_definition} and \eqref{Y_bound} by considering Riemann sums. Similarly, the estimate for the second term on the left-hand side of \eqref{Lemma 5.15'-5.16'} follows from \eqref{y_definition} and from the observation that $y^{\epsilon}_{\f x,\tilde{\f x}}(\f r)$ satisfies the same bound as in \eqref{Y_bound} by an analogous proof. We omit the details.

\paragraph{Step 2}
We show that for $\delta=\delta(\nu)$ sufficiently small (for the precise bound, see \eqref{delta_condition} below), the following estimate holds.
\begin{multline}
\label{Lemma 5.19'}
\biggl|\int_{[\delta,1/\delta]^p}\,\dd \f r\, \ee^{-\kappa |\f r|}\,y^{\epsilon}_{\f x,\tilde{\f x}} (\f r, \f r)-
\nu^p\,\sum_{\f r \in [\delta,1/\delta]_{\nu}^p} \ee^{-\kappa |\f r|}\,y^{\epsilon}_{\f x,\tilde{\f x}} (\f r, \f r)\biggr|
\\
\leq C_{\kappa,v}^p\,\ee^{T^{\epsilon}}\,\epsilon^{-\frac{d}{2}}\,\delta^{-\frac{dp}{2}}\,\sqrt{\nu}+ C_p^d \,\ee^{T^{\epsilon}}\,\delta^{-\frac{dp}{2}-(d-1)}\,\nu\,.
\end{multline}
By arguing analogously as in the proof of \cite[Lemma 5.17]{FKSS_2020}, we obtain that for all $\f r \in [\delta,1/\delta]_{\nu}^p$ and for all $\f s \in (0,\infty)^p$ with $\f s \leq \f r$ we have
\begin{equation}
\label{Lemma 5.17'}
\bigl|y^{\epsilon}_{\f x,\tilde{\f x}}(\f r, \f r)-y^{\epsilon}_{\f x, \tilde{\f x}}(\f r, \f s)\bigr| \leq C_{d,v}^p \,\ee^{T^{\epsilon}}\,\epsilon^{-\frac{d}{2}}\, \delta^{-\frac{dp}{2}}\,\max_{i}\,(r_i-s_i)\,.
\end{equation}
We note that the $\delta$ dependence in \eqref{Lemma 5.17'} is obtained from Lemma \ref{Heat_kernel_estimates} \ref{itm:heatkernel1}.
Furthermore, the $\epsilon$ dependence is obtained from Lemma \ref{v^{epsilon}_properties} \ref{itm:veps1}.

By arguing analogously as in the the proof of \cite[Lemma 5.18]{FKSS_2020}, we deduce that
\begin{equation}
\label{Lemma 5.18'}
\Bigl|\ee^{-\kappa |\f r|}\,y^{\epsilon}_{\f x,\tilde{\f x}}(\f r, \f s)-\ee^{-\kappa |\tilde{\f r}|}\,y^{\epsilon}_{\f x, \tilde{\f x}}(\tilde{\f r}, \f s)\Bigr| \leq C_{\kappa}^p \,\ee^{T^{\epsilon}}\,\delta^{-\frac{dp}{2}}\,\frac{\nu}{\min_{i}\,(r_i-s_i)}\,,
\end{equation}
provided that $\f r,\tilde{\f r} \in [\delta,1/\delta]^p$ are such that $|\f r-\tilde{\f r}| \leq \nu$, $\f s \in (0,\infty)^p$ satisfies 
\begin{equation}
\label{s_conditions}
\f s \leq \f r\,, \quad \f s \leq \tilde {\f r}\,, \quad \mathrm{min}_i s_i \geq \delta/2\,, \quad 1 \geq \mathrm{min}_{i} (r_i-s_i) \geq 2\nu\,.
\end{equation}

We now explain how \eqref{Lemma 5.17'} and \eqref{Lemma 5.18'} imply \eqref{Lemma 5.19'} for a suitable choice of $\delta$.
We recall \eqref{floor_nu} and write $\lfloor \f r \rfloor_{\nu} \deq (\lfloor r_i \rfloor_{\nu})_{i=1}^{p}$. By using \eqref{Re_f2} and Lemma \ref{Heat_kernel_estimates} \ref{itm:heatkernel1} in \eqref{y_definition2}, we have that for all $\f r \in [\delta,1/\delta]^p$ and $\f s \in (0,\infty)^p$ with $\f s \leq \f r$
\begin{equation}
\label{bound_y}
\bigl|y^{\epsilon}_{\f x,\tilde{\f x}}(\f r, \f s)\bigr| \leq C_d^p\,\ee^{T^{\epsilon}}\,\delta^{-\frac{dp}{2}}\,.
\end{equation}
Using \eqref{bound_y}, we get that 
\begin{equation}
\label{(5.47')}
\Biggl|\nu^p\,\sum_{\f r \in [\delta,1/\delta]_{\nu}^p} \ee^{-\kappa |\f r|}\,y^{\epsilon}_{\f x,\tilde{\f x}} (\f r, \f r)-\int_{[\delta,1/\delta]^p} \dd \f r\,\ee^{-\kappa |\lfloor \f r \rfloor_{\nu}|}\,y^{\epsilon}_{\f x, \tilde{\f x}}(\lfloor \f r \rfloor_{\nu}, \lfloor \f r \rfloor_{\nu})
\Biggr| \leq C_d^p\,\ee^{T^{\epsilon}}\,\delta^{-\frac{dp}{2}-(d-1)}\,\nu\,.
\end{equation}
We henceforth assume
\begin{equation}
\label{delta_condition}
\delta \sim \nu^a\,,\quad a \in (0,1/2)\,.
\end{equation}
For fixed $\f r \in [\delta,1/\delta]^p$, and $\nu$ sufficiently small, we define $\f s \equiv \f s(\f r) \in (0,\infty)^p$ by $s_i \deq r_i -\sqrt{\nu}$. Note that $s_i>0$ for small $\nu$ by \eqref{delta_condition}.
By \eqref{Lemma 5.17'}, we have that
\begin{equation}
\label{(5.48')}
\Biggl|\int_{[\delta,1/\delta]^p} \dd \f r\, \ee^{-\kappa |\f r|}\, y^{\epsilon}_{\f x,\tilde{\f x}}(\f r,\f r)-
\int_{[\delta,1/\delta]^p} \dd \f r\, \ee^{-\kappa |\f r|}\, y^{\epsilon}_{\f x,\tilde{\f x}}(\f r,\f s(\f r))\Biggr| \leq C_{\kappa,v}^p\,\ee^{T^{\epsilon}}\,\epsilon^{-\frac{d}{2}}\, \delta^{-\frac{dp}{2}}\,\sqrt{\nu}
\end{equation}
and 
\begin{multline}
\label{(5.49')}
\Biggl|\int_{[\delta,1/\delta]^p} \dd \f r\, \ee^{-\kappa |\lfloor \f r \rfloor_{\nu}|}\, y^{\epsilon}_{\f x,\tilde{\f x}}(\lfloor \f r \rfloor_{\nu},\lfloor \f r \rfloor_{\nu})-
\int_{[\delta,1/\delta]^p} \dd \f r\, \ee^{-\kappa |\lfloor \f r \rfloor_{\nu}|}\, y^{\epsilon}_{\f x,\tilde{\f x}}(\lfloor \f r \rfloor_{\nu},\f s(\f r))\Biggr| 
\\
\leq C_{\kappa,v}^p\,\ee^{T^{\epsilon}}\,\epsilon^{-\frac{d}{2}}\, \delta^{-\frac{dp}{2}}\,\sqrt{\nu}\,.
\end{multline}
In \eqref{(5.49')}, we used $|\lfloor r_i \rfloor_{\nu}-s_i| \lesssim \sqrt{\nu}$.
By \eqref{Lemma 5.18'}, we have that
\begin{equation}
\label{(5.50')}
\Biggl|\int_{[\delta,1/\delta]^p} \dd \f r\, \ee^{-\kappa |\f r|}\, y^{\epsilon}_{\f x,\tilde{\f x}}(\f r,\f s(\f r))-
\int_{[\delta,1/\delta]^p} \dd \f r\, \ee^{-\kappa |\lfloor \f r \rfloor_{\nu}|}\, y^{\epsilon}_{\f x,\tilde{\f x}}(\lfloor \f r \rfloor_{\nu},\f s(\f r))\Biggr| \leq C_{\kappa}^p\,\ee^{T^{\epsilon}}\,\epsilon^{-\frac{d}{2}}\, \delta^{-\frac{dp}{2}}\,\sqrt{\nu}\,.
\end{equation}
Here, we used that for small $\nu$, $\f s(\f r)$ satisfies \eqref{s_conditions} with $\tilde{\f r}=\lfloor \f r \rfloor_{\nu}$.
We hence deduce \eqref{Lemma 5.19'} from \eqref{(5.47')}, \eqref{(5.48')}, \eqref{(5.49')}, and \eqref{(5.50')}.

\paragraph{Step 3}
Let us define
\begin{equation}
\label{cal_I}
\cal I^{\epsilon} \deq \nu^p\,\sum_{\f r \in [\delta,1/\delta]_{\nu}^p}\,\ee^{-\kappa |\f r|}\,y^{\epsilon}_{\f x,\tilde{\f x}}(\f r,\f r)
\end{equation}
and
\begin{equation}
\label{cal_J}
\cal J \deq \nu^p\,\sum_{\f r \in [\delta,1/\delta]_{\nu}^p}\,\ee^{-\kappa |\f r|}\,Y_{\f x,\tilde{\f x}}(\f r)\,.
\end{equation}
Recalling \eqref{Theta_function}, we show that
\begin{equation}
\label{Step_3_bound}
\bigl|\cal I^{\epsilon}-\cal J\bigr| \leq C_{\kappa,v}^p\, \ee^{T^{\epsilon}}\,\epsilon^{-\frac{5d+1}{2}}\,\delta^{-\frac{dp}{2}}\,\Theta^{1/2}(\nu)\,.
\end{equation}
In order to show \eqref{Step_3_bound}, we first recall \eqref{y_definition2}, Lemma \ref{Lemma 5.2'}, and \eqref{sigma_time_average_observation}, and compare \eqref{cal_I} with
\begin{equation}
\label{cal_I'}
\tilde{\cal I} \deq \nu^p\,\sum_{\f r \in [\delta,1/\delta]_{\nu}^p}\,\ee^{-\kappa |\f r|}\, \int \mu_{\cal C_{\eta}} (\dd \sigma)\, \ee^{T^{\epsilon}-\frac{\ii \tau^\epsilon [\sigma]}{\nu}}\,\ee^{F_2(\sigma)}\,
\prod_{i=1}^{p} \Biggl[\int \mathbb{W}^{r_i,0}_{x_i,\tilde{x}_i}(\dd \omega_i)\, \Bigl(\ee^{\ii \int_{0}^{r_i} 
\dd t\, \langle \sigma \rangle (\omega_i(t))}-1\Bigr)\Biggr]\,.
\end{equation}
In other words, \eqref{cal_I'} is obtained  by replacing the factor by $\ee^{f_2(\langle \sigma \rangle)}$ by $\ee^{F_2(\sigma)}$ in the $\sigma$ integral representation of \eqref{cal_I}.
Using \eqref{Re_F2}, \eqref{Re_f2}, the Cauchy-Schwarz inequality, Lemmas \ref{Partition_function_rate_of_convergence_4}--\ref{Partition_function_rate_of_convergence_5}, and Lemma \ref{Heat_kernel_estimates} \ref{itm:heatkernel1}, we deduce that
\begin{equation}
\label{cal_I_difference}
\bigl|\cal I^{\epsilon}-\tilde{\cal I}\bigr| \leq C_{\kappa,v}^p\, \ee^{T^{\epsilon}}\,\epsilon^{-\frac{5d+1}{2}}\,\delta^{-\frac{dp}{2}}\,\Theta^{1/2}(\nu)\,.
\end{equation}
We now compare \eqref{cal_J} and \eqref{cal_I'}. By \eqref{Re_F2} and the Cauchy-Schwarz inequality, 
we have that
\begin{multline}
\label{cal_I_difference2}
\bigl|\tilde{\cal I}-\cal J\bigr| \leq \ee^{T^{\epsilon}}\,
\nu^p\,\sum_{\f r \in [\delta,1/\delta]_{\nu}^p}\,\ee^{-\kappa |\f r|}\, 
\prod_{i=1}^{p} \int \mathbb{W}^{r_i,0}_{x_i,\tilde{x}_i}(\dd \omega_i)
\\
\times 
\Biggl( \int \mu_{\cal C_{\eta}} (\dd \sigma)\,
\Biggl| \prod_{i=1}^{p} \Bigl(\ee^{\ii \int_{0}^{r_i} 
\dd t\, \langle \sigma \rangle (\omega_i(t))}-1\Bigr)-\prod_{i=1}^{p} \Bigl(\ee^{\ii \int_{0}^{r_i} 
\dd t\, \sigma ([t]_{\nu},\omega_i(t))}-1\Bigr) \Biggr|^2 \Biggr)^{1/2}\,.
\end{multline}
By arguing analogously as in the proof of \cite[Lemma 5.20]{FKSS_2020}, we have that for fixed $1 \leq i \leq p$, 
\begin{multline}
\label{cal_I_difference3}
\int \mu_{\cal C_{\eta}} (\dd \sigma)\,
\Bigl|\ee^{\ii \int_{0}^{r_i} 
\dd t\, \langle \sigma \rangle (\omega_i(t))}-\ee^{\ii \int_{0}^{r_i} 
\dd t\, \sigma ([t]_{\nu},\omega_i(t))}\Bigr|^2 \leq \frac{C}{\nu} \sum_{s,\tilde s \in [0,r_i)_{\nu}} \int_{0}^{\nu} \dd u_{1}\,\dd u_{2}\,\dd \tilde{u}_{1}\,\dd \tilde{u}_{2}\,\delta_{\eta}(u_{1}-\tilde{u}_{1})
\\
\times \Bigl(v^{\epsilon}\bigl(\omega_{i}(s+u_{1})-\omega_{i}(\tilde{s}+\tilde{u}_{1})\bigr)-v^{\epsilon}\bigl(\omega_{i}(s+u_{2})-\omega_{i}(\tilde{s}+\tilde{u}_{2})\bigr)\Bigr)\,,
\end{multline}
where $[a,b)_{\nu} \deq [a,b-\nu]_{\nu}$. Using Lemma \ref{v^{epsilon}_properties} \ref{itm:veps2} and $\delta_{\eta} \geq 0$, it follows that 
\begin{multline}
\label{cal_I_difference4}
\eqref{cal_I_difference3} \leq \frac{C}{\epsilon^{d+1}\nu}\,\sum_{s,\tilde s \in [0,r_i)_{\nu}} \int_{0}^{\nu} \dd u_{1}\,\dd u_{2}\,\dd \tilde{u}_{1}\,\dd \tilde{u}_{2}\,\delta_{\eta}(u_{1}-\tilde{u}_{1})\, 
\\
\times
\Bigl[ \bigl| \omega_i (s+u_1)-\omega_i (s+u_2)\bigr|_{\Lambda}+\bigl| \omega_i (\tilde s+\tilde{u}_1)-\omega_i (\tilde s+\tilde{u}_2)\bigr|_{\Lambda}
\Bigr]\,.
\end{multline}
We now use \eqref{cal_I_difference4} and a telescoping argument in \eqref{cal_I_difference2}. In particular, by using the Cauchy-Schwarz inequality in $\mathbb{W}^{r_i,0}_{x_i,\tilde{x}_i}(\dd \omega_i)$, applying Lemma \ref{Heat_kernel_estimates} \ref{itm:heatkernel1}--\ref{itm:heatkernel2} and \eqref{delta_nu_integral},  we deduce that 
\begin{equation}
\label{cal_I_difference5}
\bigl|\tilde{\cal I}-\cal J\bigr|  \leq C_{\kappa,v}^p\,\ee^{T^{\epsilon}}\,\epsilon^{-\frac{d+1}{2}}\,\delta^{-\frac{dp}{2}}\,\sqrt{\nu}\,.
\end{equation}
Above we applied \eqref{Heat_kernel_estimates_ii_1} when using Lemma \ref{Heat_kernel_estimates} \ref{itm:heatkernel2}.
We hence deduce \eqref{Step_3_bound} from \eqref{cal_I_difference} and \eqref{cal_I_difference5}. Here, we also recall \eqref{Theta_function}.

We combine \eqref{Lemma 5.15'-5.16'}, \eqref{Lemma 5.19'}, \eqref{Step_3_bound}, and optimize in $\delta$, keeping  \eqref{delta_condition} in mind. We hence take
$\delta \sim \nu^{\frac{1}{4p+4}}$ when $d=2$ and $\delta \sim \nu^{\frac{1}{6p+4}-}$ when $d=3$.
Note that $\delta$ then indeed satisfies \eqref{delta_condition}.
Putting everything together, and recalling \eqref{T^epsilon_bound}, we obtain \eqref{Q_and_q_estimates_ii}, with $\theta(d,p)$ as in \eqref{theta_definition} which concludes the proof of claim \ref{itm:Q2}.
\end{proof}

\begin{proof}[Proof of Lemma \ref{lower_bound}]
We first prove \ref{itm:lb1}. We recall \eqref{W^epsilon}. From Jensen's inequality we get
\begin{equation}
\label{lower_bound_classical_1}
\zeta^{W^\epsilon} = \E[\ee^{-W^\epsilon}] \geq \ee^{-\E[W^{\epsilon}(\phi)]}\,.
\end{equation}
We note that 
\begin{equation}
\label{lower_bound_classical_2}
\E \biggl[\int \dd x\, \wick{|\phi(x)|^2}\biggr]=0\,.
\end{equation}
Furthermore, by using Wick's theorem and arguing analogously as in \cite[Section 3.1]{frohlich2017gibbs}, we have that 
\begin{equation*}
\E \biggl[\int \dd x\,\dd \tilde x\, \wick{|\phi(x)|^2} \,v^{\epsilon}(x-\tilde x)\, \wick{|\phi(\tilde x)|^2}\biggr] =
\int \dd x\,\dd \tilde x\, v^{\epsilon}(x-\tilde x)\,G(x-\tilde x)\,G(\tilde x-x)\,.
\end{equation*}
By using that $G$ is even, followed by \eqref{v_epsilon}, Lemma \ref{v^{epsilon}_properties} \ref{itm:veps1}, the above expression is in absolute value
\begin{equation}
\label{lower_bound_classical_3}
\leq \int \dd x\,|v^{\epsilon}(x)|\, G(x)^2  \lesssim_v \frac{1}{\epsilon^d}\,\int_{|x| \lesssim \epsilon} \dd x\, G(x)^2\,.
\end{equation}
Using Lemma \ref{lem:G_smoothness} when $d=2$ and Remark \ref{G_smoothness_3D_remark} when $d=3$, a direct calculation shows that 
\begin{equation}
\label{lower_bound_classical_4}
\eqref{lower_bound_classical_3} \lesssim_{\kappa,v} \chi(\epsilon)^2\,.
\end{equation}
Here, we recall \eqref{sigma(epsilon)}. Claim \ref{itm:lb1} now follows by recalling \eqref{W^epsilon} and substituting \eqref{lower_bound_classical_2} and \eqref{lower_bound_classical_4} into \eqref{lower_bound_classical_1}.

We now show \ref{itm:lb2}. We show that for $\alpha \in (0,1-\frac{d}{4})$, we have
\begin{equation}
\label{5.20_(ii)}
\cal Z \geq \exp\Biggl[-c \biggl(\chi(\epsilon)^2+\frac{\nu^{\alpha}}{\epsilon^d}\biggr)\Biggr]\,,
\end{equation}
for some $c>0$ depending on $\kappa$.
Note that \eqref{5.20_(ii)} implies claim \ref{itm:lb1}, since by \eqref{epsilon_lower_bound}--\eqref{sigma(epsilon)}, we have that $\chi(\epsilon)^2 \gtrsim_{\alpha,d}\frac{\nu^{\alpha}}{\epsilon^d}$.

We recall \eqref{Hamiltonian_H} and \eqref{free_Hamiltonian_H}, and set 
$\cal W \deq H^{\epsilon}-H^{(0)}$.
By using the Peierls-Bogoliubov inequality (see \cite[Section 2.5]{ruelle1999statistical}) in \eqref{Z^epsilon_quantum}, we have that
\begin{equation}
\label{lower_bound_quantum_1}
\cal Z \geq \exp\Biggl[\tr_{\cal F} \Biggl(-\cal W \, \frac{\ee^{-H^{(0)}}}{
\tr_{\cal F}(\ee^{-H^{(0)}})}\Biggr)\Biggr]=\exp\Bigl[-\varrho_{\nu}^{(0)}(\cal W)\Bigr]\,,
\end{equation}
where for a closed operator $\cal A$ on $\cal F$, we define 
\begin{equation*}
\varrho_{\nu}^{(0)}(\cal A) \deq   \frac{\tr_{\cal F}(\cal A\,\ee^{-H^{(0)}})}{{\tr_{\cal F}(\ee^{-H^{(0)}})}}\,.
\end{equation*}
Since 
$
\varrho_{\nu}^{(0)} [\int \dd x \, \pb{\nu  a^*(x) a(x) - \varrho_\nu}]=0\,,
$
we get that
\begin{equation}
\label{lower_bound_quantum_2}
\varrho_{\nu}^{(0)}(\cal W)= \frac{1}{2}\,\varrho_{\nu}^{(0)} \biggl[\int \dd x \, \dd \tilde x \, \pb{\nu a^*(x) a_\nu(x) - \varrho_\nu} \, v^\epsilon(x - \tilde x) \, \pb{\nu a^*(\tilde x) a (\tilde x) - \varrho_\nu}\biggr]\,,
\end{equation}
By using the quantum Wick theorem (see \cite[Lemma B1]{frohlich2017gibbs}) and using \cite[Lemma 2.10]{frohlich2017gibbs}, the right-hand side of \eqref{lower_bound_quantum_2} is
\begin{equation}
\label{lower_bound_quantum_3}
\sim \int \dd x\,\dd \tilde x\,G_{[\nu]}(x-\tilde x) \Bigl[G_{[\nu]} (\tilde x-x) +\nu \delta (\tilde x-x)\Bigr]\,v^{\epsilon}(x-\tilde x)\,,
\end{equation}
where the quantum Green function $G_{[\nu]}$ is given by
\begin{equation}
\label{lower_bound_quantum_4}
G_{[\nu]}(x) \deq \sum_{k \in \Z^d} \frac{\nu}{\ee^{\lambda_k \nu}-1}\,u_k\,,
\end{equation}
with notation as in \eqref{lambda_k}.
From \eqref{lower_bound_quantum_4}, we have that $G_{[\nu]}$ is even and therefore by Lemma \ref{v^{epsilon}_properties} \ref{itm:veps1}, we have
\begin{equation}
\label{lower_bound_quantum_5}
|\eqref{lower_bound_quantum_3}| \lesssim \int \dd x\, |v^{\epsilon}(x)|\,G_{[\nu]}(x)^2+ \frac{\nu}{\epsilon^d} \, G_{[\nu]}(0)\,.
\end{equation}
We now bound each of the two terms on the right-hand side of \eqref{lower_bound_quantum_5} separately.
For the first term, we note that for $\alpha \in (0,1-\frac{d}{4})$, $G=\sum_{k \in \Z^d} \frac{1}{\lambda_k}\,\ee^{2\pi \ii k \cdot x} \in L^2(\Lambda)$ satisfies
\begin{equation}
\label{lower_bound_quantum_6}
\bigl\|G_{[\nu]}-G\bigr\|_{L^2} \lesssim_{\kappa} \nu^{\alpha}\,.
\end{equation}
In order to obtain \eqref{lower_bound_quantum_6}, we note that
for all $\delta \in [0,1]$, we have 
\begin{equation}
\label{lower_bound_quantum_7}
\biggl|\frac{\nu}{\ee^{\lambda_k \nu}-1}-\frac{1}{\lambda_k} \biggr| \leq \nu^{1-\delta}\,\frac{1}{\lambda_k^{\delta}}\,.
\end{equation}
Using Plancherel's theorem and taking $\delta>d/4$ in \eqref{lower_bound_quantum_7}, we deduce \eqref{lower_bound_quantum_6}.
We now use H\"{o}lder's inequality, Lemma \ref{v^{epsilon}_properties} \ref{itm:veps1},  \eqref{lower_bound_quantum_6}, and recall \eqref{lower_bound_classical_3}--\eqref{lower_bound_classical_4} to deduce that the first term on the right-hand side of \eqref{lower_bound_quantum_5} is
\begin{equation}
\label{lower_bound_quantum_8}
\lesssim_{\kappa} \int \dd x\, |v^{\epsilon}(x)|\,G(x)^2+\frac{\nu^{\alpha}}{\epsilon^d} \lesssim_{\kappa,v} \chi(\epsilon)^2+\frac{\nu^{\alpha}}{\epsilon^d}\,.
\end{equation}
From \eqref{lower_bound_quantum_4}, by considering the terms with $\lambda_k \leq \frac{1}{\nu}$ and $\lambda_k > \frac{1}{\nu}$ separately, we deduce that the second term on the right-hand side of \eqref{lower_bound_quantum_5} is
\begin{equation}
\label{lower_bound_quantum_9}
\lesssim_{\kappa} \frac{\nu \chi(\sqrt{\nu})}{\epsilon^d}\lesssim_{\alpha}\frac{\nu^{\alpha}}{\epsilon^d}\,,
\end{equation}
since $\alpha<1/2$.
Estimate \eqref{5.20_(ii)} now follows from \eqref{lower_bound_quantum_5}, \eqref{lower_bound_quantum_8}, and \eqref{lower_bound_quantum_9}.
\end{proof}
We now give the details of the proof of the claim from Remark \ref{endpoint_admissible}.
\begin{proof}[Proof of claim from Remark \ref{endpoint_admissible}]
The claim follows if, instead of \eqref{nonlocal_bound_2}, we show
\begin{equation}
\label{nonlocal_bound_2'}
\|V^{\epsilon}-V\|_{L^2(\P)} \lesssim_v \|v^{\epsilon}-v\|_{H^{-1+\frac{\delta}{2}}}^{1/2}\,.
\end{equation}

Let us note that the $x$-dependence in \eqref{F^epsilon} is of the form $\bigl[v^{\epsilon}(x-y)-v(x-y)\bigr]\,G(x-x_a)\,G(x-x_b)$ for some distinct $a, b \in B \setminus \{(1,1,\pm)\}$. We consider the different possibilities for $x_a$ and $x_b$.
\paragraph{\textbf{Case 1:} $x_a=x_b=\tilde x$}
This case is easy because the variables $(x,\tilde x)$ and $(y, \tilde y)$ decouple. In particular, the contribution to \eqref{nonlocal_bound_3} is 
\begin{equation}
\label{nonlocal_bound_4}
\lesssim \biggl| \int \dd x\, \dd \tilde x \bigl[v^{\epsilon}(x-\tilde x)-v(x-\tilde x)\bigr]\,G^2(x-\tilde x) \biggr| \lesssim \|v^{\epsilon}-v\|_{H^{-1+\frac{\delta}{2}}}\, \|G\|_{H^{-1+\frac{\delta}{4}}}^2\lesssim \|v^{\epsilon}-v\|_{H^{-1+\frac{\delta}{2}}}\,.
\end{equation}
In order to prove \eqref{nonlocal_bound_4}, we used duality, 
and the bilinear estimate given in \cite[Lemma 4.2]{sohinger2019microscopic} which implies that $\|G^2\|_{H^{-1+\frac{\delta}{2}}} \lesssim_{\delta} \|G\|_{H^{-1+\frac{\delta}{4}}}^2<\infty$.

\paragraph{\textbf{Case 2:} $x_a=x_b \neq \tilde x $}

We can assume without loss of generality that $x_a=y$. The integrand that we then consider is  
\begin{equation*}
\bigl[v^{\epsilon}(x-\tilde x)-v(x-\tilde x)\bigr]\,G^2(x-y)\,\bigl[v^{\epsilon}(y -\tilde y)-v(y -\tilde y)\bigr]\,G^2(\tilde x -\tilde y)\,.
\end{equation*}
We first integrate in $x$ and use an estimate analogous to \eqref{nonlocal_bound_4}. We then integrate in $\tilde x $ and use $G \in L^2(\Lambda)$. Finally, we integrate in $y,\tilde y$ and recall \eqref{v_regularized} to deduce that the contribution to \eqref{nonlocal_bound_3} is $\lesssim \|v^{\epsilon}-v\|_{H^{-1+\frac{\delta}{2}}}^{1/2}\, \|v\|_{L^1}^{1/2}$.

\paragraph{\textbf{Case 3:} $x_a=\tilde x, x_b \neq \tilde x$ or $x_a \neq \tilde x, x_b=\tilde x$}
We can assume without loss of generality that $x_b=y$. Then, we consider the integrand
\begin{equation*}
\bigl[v^{\epsilon}(x-\tilde x)-v(x-\tilde x)\bigr]\,G(x-\tilde x)\,G(x-y)\,\bigl[v^{\epsilon}(y-\tilde y)-v(y-\tilde y)\bigr]\,G(\tilde x-\tilde y)\,G(y-\tilde y)\,.
\end{equation*}
We first fix $x, \tilde y$ and integrate in $\tilde x$ and $y$, using estimates analogous to \eqref{nonlocal_bound_4}. We then integrate in $x, \tilde y$ to deduce that the contribution to \eqref{nonlocal_bound_3} is $\lesssim \|v^{\epsilon}-v\|_{H^{-1+\frac{\delta}{2}}}$.
\paragraph{\textbf{Case 4:} $x_a \neq x_b, x_a \neq \tilde x, x_b \neq \tilde x$}

This case is similar to Case 2 and we get the same upper bound by an analogous argument.
We hence deduce \eqref{nonlocal_bound_2'}.
\end{proof}

\bibliography{bibliography} 
\bibliographystyle{amsplain}

\bigskip

\noindent
J\"urg Fr\"ohlich, ETH Z\"urich, Institute for Theoretical Physics, \href{mailto:juerg@phys.ethz.ch}{juerg@phys.ethz.ch}.
\\[0.3em]
Antti Knowles, University of Geneva, Section of Mathematics, \href{mailto:antti.knowles@unige.ch}{antti.knowles@unige.ch}.
\\[0.3em]
Benjamin Schlein, University of Z\"urich, Institute of Mathematics, \href{mailto:benjamin.schlein@math.uzh.ch}{benjamin.schlein@math.uzh.ch}.
\\[0.3em]
Vedran Sohinger, University of Warwick, Mathematics Institute, \href{mailto:V.Sohinger@warwick.ac.uk}{V.Sohinger@warwick.ac.uk}.

\bigskip

\paragraph{Acknowledgements}
We thank David Brydges for helpful correspondence and Trishen Gunaratnam for 
very useful discussions. AK acknowledges funding from the European Research Council under the European Union’s Horizon 2020 research and innovation programme (grant agreement No.\ 715539\_RandMat), funding from the Swiss National Science Foundation through the NCCR SwissMAP grant, and support from the US National Science Foundation under Grant No.\ DMS-1928930 during his participation in the program ``Universality and Integrability in Random Matrix Theory and Interacting Particle Systems'' hosted by the Mathematical Sciences Research Institute in Berkeley, California during the Fall semester of 2021. BS acknowledges partial support from the NCCR SwissMAP, from the Swiss National Science Foundation through the Grant ``Dynamical and energetic properties of Bose-Einstein condensates'' and from the European Research Council through the ERC-AdG CLaQS. VS acknowledges support of the EPSRC New Investigator Award grant EP/T027975/1.

\end{document}

%% file: hypercontractivity.pdf_tex
\begingroup%
  \makeatletter%
  \providecommand\color[2][]{%
    \errmessage{(Inkscape) Color is used for the text in Inkscape, but the package 'color.sty' is not loaded}%
    \renewcommand\color[2][]{}%
  }%
  \providecommand\transparent[1]{%
    \errmessage{(Inkscape) Transparency is used (non-zero) for the text in Inkscape, but the package 'transparent.sty' is not loaded}%
    \renewcommand\transparent[1]{}%
  }%
  \providecommand\rotatebox[2]{#2}%
  \newcommand*\fsize{\dimexpr\f@size pt\relax}%
  \newcommand*\lineheight[1]{\fontsize{\fsize}{#1\fsize}\selectfont}%
  \ifx\svgwidth\undefined%
    \setlength{\unitlength}{359.03347219bp}%
    \ifx\svgscale\undefined%
      \relax%
    \else%
      \setlength{\unitlength}{\unitlength * \real{\svgscale}}%
    \fi%
  \else%
    \setlength{\unitlength}{\svgwidth}%
  \fi%
  \global\let\svgwidth\undefined%
  \global\let\svgscale\undefined%
  \makeatother%
  \begin{picture}(1,0.35466556)%
    \lineheight{1}%
    \setlength\tabcolsep{0pt}%
    \put(0,0){\includegraphics[width=\unitlength,page=1]{hypercontractivity.pdf}}%
    \put(0.08536217,0.00460221){\makebox(0,0)[lt]{\lineheight{1.25}\smash{\begin{tabular}[t]{l}$\Pi$\end{tabular}}}}%
    \put(0.448542,0.00460221){\makebox(0,0)[lt]{\lineheight{1.25}\smash{\begin{tabular}[t]{l}$\Sigma \cup \{e\}$\end{tabular}}}}%
    \put(0.87488348,0.00460221){\makebox(0,0)[lt]{\lineheight{1.25}\smash{\begin{tabular}[t]{l}$\Sigma$\end{tabular}}}}%
    \put(0.61717856,0.22893677){\makebox(0,0)[lt]{\lineheight{1.25}\smash{\begin{tabular}[t]{l}$e$\end{tabular}}}}%
    \put(0.57005811,0.33751879){\makebox(0,0)[lt]{\lineheight{1.25}\smash{\begin{tabular}[t]{l}$(i,j)$\end{tabular}}}}%
    \put(0.57517989,0.03123554){\makebox(0,0)[lt]{\lineheight{1.25}\smash{\begin{tabular}[t]{l}$(i',j')$\end{tabular}}}}%
    \put(0.67556701,0.18284062){\makebox(0,0)[lt]{\lineheight{1.25}\smash{\begin{tabular}[t]{l}$\longrightarrow$\end{tabular}}}}%
  \end{picture}%
\endgroup%